\documentclass[11pt, DIV=15]{scrartcl}
\usepackage[T1]{fontenc}
\usepackage[english]{babel}
\usepackage{amsmath, amsfonts, amsthm, amssymb, mathtools, stmaryrd}
\usepackage[pdfusetitle,hidelinks]{hyperref}
\usepackage[capitalise,noabbrev]{cleveref}
\usepackage{csquotes}
\usepackage[osf,sc]{mathpazo}
\usepackage{relsize}
\usepackage{microtype}
\usepackage{pdfpages}
\usepackage[a4paper,margin=1in]{geometry}

\usepackage{tikz}
\usetikzlibrary{arrows,backgrounds,positioning,fit,cd}

\usepackage{subcaption}

\recalctypearea

\usepackage{comment}

\usepackage{graphicx}

\usepackage{abstract}
\usepackage[inline]{enumitem}
\usepackage{booktabs}
\usepackage{todonotes}

\usepackage{tikz}
\usetikzlibrary{cd}
\usepackage{microtype}

\theoremstyle{definition}
\newtheorem{definition}{Definition}[section]

\newtheorem{remark}[definition]{Remark}

\theoremstyle{plain}
\newtheorem{lemma}[definition]{Lemma}
\newtheorem{example}[definition]{Example}
\newtheorem{corollary}[definition]{Corollary}
\newtheorem{theorem}[definition]{Theorem}

\newtheorem{proposition}[definition]{Proposition}

\newtheorem{fact}[definition]{Fact}
\Crefname{fact}{Fact}{Facts}

\theoremstyle{remark}
\newtheorem{claim}{Claim}
\Crefname{claim}{Claim}{Claims}
\newenvironment{claimproof}[1][Proof of Claim]{\begin{proof}[#1] }{ \end{proof}}
\setlist[enumerate, 1]{font=\upshape, noitemsep, nolistsep}
\setlist[enumerate, 2]{font=\upshape, noitemsep, nolistsep}
\setlist[itemize, 1]{noitemsep, nolistsep,font=\upshape}
\setlist[itemize, 2]{noitemsep, nolistsep,font=\upshape}

\newcommand{\abs}[1]{\left| #1 \right|}

\renewcommand\phi\varphi
\renewcommand\epsilon\varepsilon

\DeclareMathOperator{\soe}{soe}

\DeclareMathOperator{\atp}{atp}
\DeclareMathOperator{\gca}{gca}

\title{Weisfeiler--Leman and Graph Spectra\footnote{A conference version \cite{rattan_weisfeiler_2023} of this article appeared at SODA 2023.}}

\newcommand{\orcid}[1]{\href{https://orcid.org/#1}{\includegraphics[height=1.8ex]{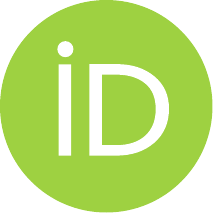}}
}
\author{Gaurav Rattan \orcid{0000-0002-5095-860X}\\
	\smaller RWTH Aachen University \\
	\smaller \texttt{rattan@cs.rwth-aachen.de}
	\and Tim Seppelt \orcid{0000-0002-6447-0568}\\
	\smaller RWTH Aachen University \\
	\smaller \texttt{seppelt@cs.rwth-aachen.de}}

\DeclareMathOperator{\id}{id}

\newcommand\multiset[1]{\left\{ \!\! \left\{ #1 \right\} \!\! \right\}}
\newcommand{\CC}{\mathbb{C}}
\usepackage{xspace}

\newcommand{\WL}{\textup{\textsf{WL}}\xspace}

\newcommand{\pin}{\textup{in}}
\newcommand{\pout}{\textup{out}}
\newcommand{\Graph}{\mathsf{Graph}}
\newcommand{\C}{\mathfrak{C}}
\newcommand{\ot}{\leftarrow}

\DeclareMathOperator{\Spec}{Spec}
\DeclareMathOperator{\Adj}{Adj}

\DeclareMathOperator{\dist}{dist}

\DeclareMathOperator{\tr}{tr}
\DeclareMathOperator{\obj}{obj}
\DeclareMathOperator{\mor}{mor}
\DeclareMathOperator{\EM}{EM}
\DeclareMathOperator{\Lab}{Labelled}

\renewcommand\vec\boldsymbol
\newcommand\matvec\mathsf
\renewcommand\phi\varphi

\setlist[enumerate, 1]{font=\upshape}  \setlist[enumerate, 2]{font=\upshape}
\setlist[itemize, 1]{font=\upshape}
\setlist[itemize, 2]{font=\upshape}

\begin{document}

	\maketitle

	\begin{abstract}

	Two simple undirected graphs are \emph{cospectral} if their respective adjacency matrices have the same multiset of eigenvalues. Cospectrality yields an equivalence relation on the family of graphs which is provably weaker than isomorphism. In this paper, we study cospectrality in relation to another well-studied relaxation of isomorphism, 
	namely $k$-dimensional Weisfeiler--Leman ($k$-\WL) indistinguishability. 
	
	Cospectrality with respect to standard graph matrices such as the adjacency or the Laplacian matrix yields a strictly finer equivalence relation than $2$-\WL indistinguishability. We show that individualising one vertex plus running $1$-\WL already subsumes cospectrality with respect to all such graph matrices. Building on this result, we resolve an open problem of F{\"u}rer~(2010) about spectral invariants.

	Looking beyond $2$-\WL, we devise a hierarchy of graph matrices generalising the adjacency matrix such that $k$-\WL indistinguishability after a fixed number of iterations can be captured as a spectral condition on these matrices. Precisely, we provide a spectral characterisation of $k$-\WL indistinguishability after $d$ iterations, for $k,d \in \mathbb{N}$. 
	Our results can be viewed as characterisations of homomorphism indistinguishability over certain graph classes in terms of matrix equations. The study of homomorphism indistinguishability is an emerging field, to which we contribute by extending the algebraic framework of Man{\v c}inska and Roberson (2020) and Grohe et al.\@ (2022).
\end{abstract}

\section{Introduction}
\label{sec:intro}

The algebraic properties of the adjacency matrix of a graph can reveal a great deal of information about the structure of the graph. In particular, the spectra of certain matrices associated with a graph, such as the adjacency matrix or the Laplacian matrix, can already determine several important graph properties such as connectedness and bipartiteness~\cite{brouwer_spectra_2012}. Moreover, they can be used to approximate several important graph parameters such as graph expansion~\cite{CHGR} and chromatic number~\cite{SGTcol}. A well-known result of this kind is Kirchhoff's Matrix Tree Theorem which asserts that the number of spanning trees of a connected graph is equal to the product of its non-zero Laplacian eigenvalues, cf.\@~\cite{godsil_algebraic_2004}. Apart from the adjacency matrix and the Laplacian matrix, several other graph matrices such as the signless or normalised Laplacian, the Seidel matrix, etc.\@ have been utilised to obtain further insights into several graph properties and parameters~\cite{von_luxburg_tutorial_2007, cvetkovic_signless_2007, van_dam_which_2003}. In algebraic graph theory, the spectra of such graph matrices are collectively referred to as spectral graph invariants. 

Two simple undirected graphs are said to be \emph{cospectral} if their respective adjacency matrices have the same multiset of eigenvalues. Whereas isomorphic graphs are necessarily cospectral, the converse does not hold in general. Nevertheless, cospectrality is an important equivalence relation on graphs which continues to be extensively studied in algebraic graph theory. The notion of cospectrality can be easily extended to graph matrices other than the adjacency matrix: each such graph matrix induces a corresponding equivalence relation on graphs~\cite{van_dam_which_2003}. The main objective of this paper is to analyse cospectrality in relation to the well-studied relaxations of isomorphism arising from algebra, combinatorics, and logic, see e.g.~\cite{immerman_canon_90, AtseriasM13, DellGR18}. In particular, we focus on the equivalence relations arising from the Weisfeiler--Leman framework~\cite{weisfeiler_construction_76}. Formally, the $k$-dimensional Weisfeiler--Leman ($k$-\WL) is a combinatorial procedure which iteratively partitions the set of $k$-tuples of vertices into finer equivalence classes based on their local neighbourhoods, cf.\@~\cite{grohe_logic_2021}. In the parlance of $k$-\WL, this yields a so-called stable colouring function which associates a colour to every $k$-tuple of vertices. Two graphs are \emph{$k$-\WL indistinguishable} if the $k$-\WL procedure generates the same multiset of stable vertex-tuple colours for the two graphs. For each $k \in \mathbb{N}$, the $k$-\WL procedure thus yields an equivalence relation on the domain of graphs, namely \emph{$k$-\WL indistinguishability}, and moreover, letting $k \in \mathbb{N}$ vary, an infinite hierarchy of strictly finer equivalence relations~\cite{cai_optimal_1992}. Hence, there is no fixed dimension $k$ such that $k$-\WL indistinguishability coincides with isomorphism. Importantly, the \WL framework has been shown to admit equivalent formalisms in descriptive complexity~\cite{immerman_canon_90}, convex optimisation~\cite{AtseriasM13}, and substructure counting~\cite{DellGR18}. The fact that $k$-WL is established as a powerful and expressive framework which captures a wide variety of equivalence relations on the domain of graphs justifies our interest in using the $k$-\WL framework as a yardstick for analysing cospectrality. It may be noted that for most natural graph classes, there exists a fixed $k \in \mathbb{N}$ such that $k$-\WL indistinguishability coincides with the isomorphism relation on the given graph class. The prominent examples include trees~\cite{AHU, immerman_canon_90}, planar graphs~\cite{MGplnr, kiefer_planar_19}, and graphs with excluded minors~\cite{grohe_excluded_12}.

Our line of investigation has a concrete practical motivation in the context of the recent work in the field of machine learning with graph data. In the emerging field of graph representation learning, the design of graph neural networks (GNNs) for achieving robust inference on graphs has received much attention~\cite{Hamilton}. There has been considerable recent work on measuring the expressive power GNNs, i.e.\@ whether a GNN can generate distinct representations for non-isomorphic graphs. It turns out that the expressive power of the most commonly used GNNs, namely \emph{message-passing} GNNs, is equal to $1$-\WL indistinguishability~\cite{XHLJ19,Morriswlgoneural}, i.e.\@ a message-passing GNN can generate distinct representations for two graphs if and only if they are $1$-\WL distinguishable. On the other hand, \emph{spectral} GNNs work in a very different manner to message-passing networks: they learn a composition of spectral transformations on the adjacency matrix of the graph. Moreover, several machine learning models directly work with the eigenvectors of the adjacency matrix, arising from transformations such as principal component analysis, matrix factorisation, etc.\@~\cite{SBGNN}. We believe that our work will serve as a contribution towards studying the expressive power of combination of spectral invariants or representations with combinatorial algorithms and message passing GNNs.

\subsection{Our Results: Graph spectra between $1$-\WL and $2$-\WL} 
\label{sec:results}
Our first set of results concerns the cospectrality relation induced by the usual graph matrices, such as the adjacency matrix and Laplacian matrix, whose distinguishing power is well-known to be subsumed by $2$-\WL.
A systematic study of spectral invariants in the regime below $2$-\WL was undertaken in \cite{furer_power_2010} where the following spectral invariant was proposed. 
\begin{definition}[Fürer's spectral invariant \cite{furer_power_2010}]\label{def:furer-inv}
Let $A(G)$ be the adjacency matrix of a graph~$G$. Let $\lambda_1 < \dots < \lambda_k$ denote the eigenvalues of $A(G)$ and write $P^{(1)}, \dots, P^{(k)}$ for the projection maps onto the respective eigenspaces. Define \emph{Fürer's spectral invariant} $\mathcal{F}(G)$ as 
\[
\mathcal{F}(G) \coloneqq
\left( \Spec A(G), 
\multiset{P_v \ \middle|\ v \in V(G)}
\right)
\]
where $P_v \coloneqq \left(p_{vv}, \multiset{p_{vw} \ \middle|\ w \in V(G)} \right)$,
$p_{vw} \coloneqq \left(P^{(1)}_{vw}, \dots, P^{(k)}_{vw}\right)$ for $v,w \in V(G)$,
and $\Spec A(G)$ is the multiset of eigenvalues of $A(G)$ with geometric multiplicities.
\end{definition}

While Fürer~\cite{furer_power_2010} showed that $\mathcal{F}$ is determined by $2$-\WL, they conjectured that the former is strictly weaker than the latter. We resolve this conjecture in the affirmative by introducing the notion of \emph{$(1,1)$-\WL indistinguishability}. 
The $(1,1)$-\WL test processes graphs $G$ and $H$ by comparing vertex-individualised copies of the graphs. For a vertex $v \in V(G)$, the \emph{vertex-individualised graph~$G_v$} is the coloured graph obtained from $G$ by colouring $v$ orange and all other vertices grey.
Note that vertex individualisation is a standard tool from the isomorphism testing literature \cite{neuen_exponential_2018}. On the practical side, augmenting GNNs with information about individualised subgraphs is a direction of current research~\cite{chendi_rattan}.

\begin{definition}\label{def:oneonewl}
Two graphs $G$ and $H$ are \emph{$(1,1)$-\WL indistinguishable} if there exists a bijection $\pi\colon V(G) \to V(H)$ such that for all $v \in V(G)$ the vertex-individualised graphs $G_v$ and $H_{\pi(v)}$ are $1$-\WL indistinguishable. Such a bijection $\pi$ is said to be \emph{colour-preserving}.
\end{definition}

It is easy to see that $(1,1)$-\WL lies between $1$-\WL and $2$-\WL in terms of its ability to distinguish non-isomorphic graphs. 
We prove that these inclusions are strict by constructing appropriate examples in \cref{ex:weaker}.
We then proceed to answer the question left open in \cite{furer_power_2010}.

\begin{theorem}\label{thm:specinv-intro}\label{cor:cospectrality-intro}
If $G$ and $H$ are $(1,1)$-\WL indistinguishable graphs then $\mathcal{F}(G) = \mathcal{F}(H)$. Hence, Fürer's spectral invariant $\mathcal{F}(\cdot)$ is strictly weaker than 
$2$-\WL.
\end{theorem}

For the proof of this theorem, we set up a general formalism of \emph{equitable matrix maps}, a generalisation of graph matrices such as the adjacency matrix or the Laplacian. An equitable matrix map~$\varphi$ associates with every graph~$G$ a matrix $\phi(G) \in \mathbb{C}^{V(G) \times V(G)}$ which preserves certain subspaces of $\mathbb{C}^{V(G)}$ induced by $1$-\WL (see \cref{sec:emm} for details). Given a graph $G$, let $\mathfrak{E}(G)$ denote the set of all matrices of the form $\varphi(G)$, where $\varphi$ is an equitable matrix map. It turns out that $\mathfrak{E}(G)$ forms an algebra and contains most of the well-known matrices associated with graphs, such as the adjacency matrix, Laplacian matrix, and the Seidel matrix. To show that $(1,1)$-\WL is sufficiently powerful to determine all of the well-known spectral invariants in this regime, we provide a spectral characterisation of $(1,1)$-\WL indistinguishable graphs in terms of equitable matrix map, cf.\@  \Cref{thm:specinv,rem:specinv}.
As an application of the above framework, we parallel a result of Godsil~\cite{godsil_equiarboreal_1981} asserting that the multiset of commute distances, a property of random walks on a graph, is determined by $2$-\WL. In fact, the hitting time, the nonsymmetric counterpart of the commute distance, is already determined by $(1,1)$-\WL.\footnote{In an earlier versions of this work including \cite{rattan_weisfeiler_2023}, we wrongly claimed that $(1,1)$-\WL indistinguishable graphs have the same multiset of commute distances. 
	Floris Geerts and Marek Cerný pointed out to us that this is not true.
	In fact, the graphs in \cref{fig:cntex} are a counterexample.}

\begin{theorem}\label{thm:rw}\label{prop:main}\label{thm:comdist}
If graphs $G$ and $H$ are $(1,1)$-\WL indistinguishable then they have the same multiset of hitting times. 
Hence, the multiset of hitting times is a strictly weaker invariant than~$2$-\WL.
\end{theorem}

From a practical viewpoint, our results diminish hopes to simulate $2$-\WL by augmenting $1$-\WL using spectra-based initial vertex colourings. Indeed, such a strategy seems infeasible since the supplementation of spectral information will not result in surpassing the $(1,1)$-\WL barrier. In particular, all invariants based on spectral properties of common graph matrices such as eigenvalues and entries of projections onto eigenspace, or even hitting times are already determined by $(1,1)$-\WL, which is strictly weaker than $2$-\WL. On the positive side, $(1,1)$-\WL indistinguishability constitutes an important equivalence relation on graphs in its own right. It is rich enough to capture a wide range of spectral properties and, in contrary to $2$-\WL, decidable in linear space.

More precisely, while the time complexity of testing two $n$-vertex $m$-edge graphs $G$ and $H$ for $2$-\WL indistinguishability is $O(n^3 \log n)$, it is $O((n+m)\log n)$ for $1$-\WL, cf.\@ \cite{berkholz_tight_2017},  \cite[Proposition~V.4]{grohe_logic_2021}, and \cite[Theorem~4.9.5]{immerman_canon_90}. In particular, for dense graphs, the complexity of $1$-\WL can be $O(n^2 \log n)$. By considering each pair of vertices in $V(G) \times V(H)$ separately, $(1,1)$-\WL can be implemented in time $O(n^2(n+m)\log n)$ and linear space. Crucially, in this case graphs are compared (rather than that one graph is canonised) as it is then not necessary to store the computed colours in a universal way, i.e.\@ such that they are meaningful across all possible input graphs. The time required by $(1,1)$-\WL can be reduced by running $1$-\WL on the pair of graphs obtained from the original input graphs by copying them $n$ times and individualising one vertex in each copy. This yields an $O(n(n+m)\log n)$ algorithm, which requires however quadratic space. The trade-off between space and time complexity merits further investigations.

\subsection{Our Results: Spectral characterisation of $k$-\WL indistinguishability after $d$~iterations}
The set of results described above leads to the following question: do there exist graph matrices such that their spectral properties yields a finer equivalence relation than $2$-\WL, and even $k$-\WL?

We provide an affirmative answer to this question by characterising indistinguishability after $d$~iterations of $k$-\WL in terms of spectral properties of certain matrices. These results depart from the results presented earlier in two notable ways. Firstly, it turns out that we need stronger spectral conditions than just the cospectrality of certain graph matrices. Nevertheless, our stated conditions on the devised graph matrices will be later presented as a spectral characterisation in the sense of representation theory, cf.\@ \cite{grohe_homomorphism_2021}. Secondly, our results stipulate the number of $k$-\WL iterations as an additional parameter, as opposed to examining stable colourings. Due to the fact~\cite{cai_optimal_1992} that two graphs are indistinguishable after $d$~iterations of $k$-\WL if and only if they satisfy the same formulae of the $(k+1)$-variable quantifier-depth-$d$ fragment of first order logic with counting quantifiers, the notion is of much theoretical interest~\cite{dawar_lovasz_2021,lichter_walk_2019,kiefer_iteration_2020}.
On the practical side, it was recently shown~\cite{Morriswlgoneural} that two graphs are indistinguishable by a message passing GNN (resp.\@ $k$-GNN) with $d$ layers if and only if they are indistinguishable after $d$ iterations of $1$-\WL (resp.\@ $k$-\WL). Due to resource limitations, $k$ and $d$ are usually chosen to be small constants, e.g.\@ $k \in \{1,2\}$ and $d \in \{3,4,5\}$, in application scenarios. This applies both when executing \WL and when deploying GNNs. Hence, the distinguishing power of $k$-\WL after a fixed number of iterations is of practical interest. 
Moreover, our results serve as a possible pathway for incorporating spectral tools in the design and analysis of more powerful variants of the $k$-\WL algorithm and higher-order GNNs. Devising such variants has recently received much attention \cite{morris_sparsewl,chendi_rattan}.

We return to a formal description of our results. The aforementioned generalisations of the adjacency matrix that we devise are instances of \emph{matrix maps}: they assign to every graph $G$ a matrix of certain dimension. For example, the \emph{adjacency matrix map} assigns to a graph $G$ its adjacency matrix $\boldsymbol{A}_G$. Aiming at comparability with the $k$-\WL hierarchy, we take a fine-grained approach and construct in \cref{ssec:fingen} for every $k \geq 1$ and $d \geq 0$ a family of matrix maps $\mathcal{A}_k^d$
whose spectra encode the information extracted by $k$-\WL after $d$ iterations.

Our efforts result in the following theorem. 
Note that its second assertion implies that the \emph{main part}~\cite{cvetkovic_eigenspaces_1997} of the spectrum of $\widehat{\boldsymbol{A}}_G$ and  $\widehat{\boldsymbol{A}}_H$ are the same, cf.\@ \cite[Proof of Theorem 2]{DellGR18}. A sufficient condition on the main parts of the spectra is harder to formulate. Nevertheless, \cref{thm:main-intro} is a spectral characterisation in the sense that it is essentially a statement on characters of certain representations, cf.\@ \cite{grohe_homomorphism_2021}.
Here, a matrix is \emph{pseudo-stochastic} if its row and column sums are one.
\begin{theorem}\label{thm:main-intro}
 Let $k \geq 1$ and $d \geq 0$.
 Given two graphs $G$ and $H$, the following are equivalent:
 \begin{enumerate}
 \item $G$ and $H$ are indistinguishable after $d$ iterations of $k$-\WL. 
 \item There exists a pseudo-stochastic matrix $X$ such that for every matrix map 
 $\widehat{\boldsymbol{A}} \in \mathcal{A}_k^d$, the associated matrices 
 $\widehat{\boldsymbol{A}}_G$ and $\widehat{\boldsymbol{A}}_H$ for graphs $G$ and $H$ satisfy the conditions 
 $X \widehat{\boldsymbol{A}}_G = \widehat{\boldsymbol{A}}_H X$.
 \end{enumerate}
\end{theorem}

The proof of \cref{thm:main-intro} builds on a framework which has recently been developed for studying homomorphism indistinguishability. Two graphs $G$ and $H$ are \emph{homomorphism indistinguishable} over a class of graphs $\mathcal{F}$ if for every $F \in \mathcal{F}$ the number of homomorphism $F \to G$  equals the number of homomorphisms $F \to H$. For example, $G$ and $H$ are homomorphism indistinguishable over cycles iff their adjacency matrices $\boldsymbol{A}_G$ and $\boldsymbol{A}_H$ are cospectral, which in turn can be characterised by the existence of an orthogonal matrix $X$ such that $X\boldsymbol{A}_G = \boldsymbol{A}_H X$, cf.\@ e.g.\@ \cite{DellGR18,grohe_homomorphism_2021}. It has emerged recently that homomorphism indistinguishability over various graph classes characterises a diverse range of equivalence relations on graphs, of which many can be represented as matrix equations  \cite{Lovasz67,DellGR18,dvorak_recognizing_2010,grohe_counting_2020,dawar_lovasz_2021,mancinska_quantum_2019,grohe_homomorphism_2021,atserias_expressive_2021,spitzer_characterising_2022,roberson_oddomorphisms_2022}. In particular, $k$-\WL indistinguishability after $d$~iterations has been characterised as homomorphism indistinguishability over graphs which admit a certain decomposition of bounded size called \emph{$(k+1)$-pebble forest cover of depth~$d$} \cite{abramsky_pebbling_2017,dawar_lovasz_2021}.

For proving \cref{thm:main-intro}, we combine the approach of \cite{mancinska_quantum_2019,grohe_homomorphism_2021} for constructing matrix equations capturing homomorphism indistinguishability with the categorical framework of \cite{dawar_lovasz_2021,abramsky_structure_2022}. While \cref{thm:main-intro} is the concrete outcome of this strategy, we believe that it conceptually enhances the available toolkit for proving characterisations of homomorphism indistinguishability in terms of matrix equations. To sketch these contributions, we outline the strategy used by \cite{mancinska_quantum_2019,grohe_homomorphism_2021} to prove such results. 
Each step is exemplified by the arguments used to derive a characterisation of homomorphism indistinguishability over trees in \cite[Section 3]{grohe_homomorphism_2021}.
\begin{enumerate} \label{recipe}
\item Define a family of labelled graphs and their homomorphism vectors\label{step1}

Labelled graphs are tuples $\boldsymbol{F} = (F, u)$ where $F$ is a graph and $u \in V(F)$. For every graph~$G$, a labelled graph $\boldsymbol{F}$ gives rise to a \emph{homomorphism vector} $\boldsymbol{F}_G \in \mathbb{C}^{V(G)}$ where $\boldsymbol{F}_G(v)$ for $v \in V(G)$ is defined as the number of \emph{homomorphisms} $h \colon F \to G$ such that $h(u) = v$, i.e.\@ mappings $h \colon V(F) \to V(G)$ such that $h(w_1w_2) \in E(G)$ for every $w_1w_2 \in E(F)$.

For the example, let $\mathcal{T}$ denote the family of labelled trees $(T, u)$, i.e.\@ of trees $T$ with an arbitrary labelled vertex $u \in V(T)$.

\item Define suitable combinatorial and algebraic operations\label{step2}

It was observed in \cite{mancinska_quantum_2019} that there is a close correspondence between combinatorial operations on labelled graphs and algebraic operations on homomorphism vectors. An instance is the \emph{gluing product} of $\boldsymbol{F} = (F, u)$ and $\boldsymbol{F}' = (F', u')$ which results in the labelled graph $\boldsymbol{F} \odot \boldsymbol{F}'$ obtained by taking the disjoint union of $F$ and $F'$, merging the vertices $u$ and $u'$ to yield the new labelled vertex. It can be shown that the entry-wise product  $\boldsymbol{F}_G \odot \boldsymbol{F}'_G$ of  the homomorphism vectors $\boldsymbol{F}_G$ and $\boldsymbol{F}'_G$ yields the homomorphism vector $(\boldsymbol{F} \odot \boldsymbol{F}')_G$ of the gluing product $\boldsymbol{F} \odot \boldsymbol{F}'$. 

For the example, we consider gluing and the operation of attaching a fresh edge $uv$ to the labelled vertex $u$ of $(F, u)$ and moving the label to $v$.
\item Prove finite generation\label{step3}

Aiming at a matrix equation with finitely many constraints, it is typically proven that the family of labelled graphs in step~\ref{step1} is finitely generated by certain \emph{basal graphs} under the operation in step~\ref{step2}. For trees, it can be easily seen that $\mathcal{T}$ is generated under gluing and attaching edges by the single-vertex labelled graph without any edges.

\item Recover a matrix equation using algebraic and representation-theoretic techniques. 

For trees, one may use the techniques in \cite{grohe_homomorphism_2021} to discover the fractional isomorphism matrix equation.\label{step4}
\end{enumerate}

The main difficulty arising when applying this strategy to pebble forest covers is that there is no obvious way to define labelling and operations such that the resulting class of labelled graphs is closed under these. We overcome this obstacle by promoting the labels from mere distinguished vertices to objects encoding the role the selected vertices play in the associated pebble forest cover. Subsequently, the gluing operation can be restricted to pairs of labelled graphs whose labels play the same role. While this resolves the combinatorial problems in steps~\ref{step1}--\ref{step3}, alterations have to be made also on the algebraic side in order to meet the requirements of step~\ref{step4}. For this purpose, we introduce representations of labelled graphs in terms of matrices which not only encode the homomorphism counts but also the role of the labelled vertex in the decomposition.

Although this approach might appear to be tailored to the graph class relevant for \cref{thm:main-intro}, it is in fact inspired by categorical principles laid out
in \cite{abramsky_pebbling_2017,abramsky_relating_2021,dawar_lovasz_2021}. 
Their framework of comonads on the category of relational structures has given rise to a categorical language for capturing natural graph classes and decompositions leading moreover to results in homomorphism indistinguishability \cite{dawar_lovasz_2021,MontacuteS21,abramsky_discrete_2022}.
For proving \cref{thm:main-intro}, we introduce bilabelled graphs augmented by additional information accompanied by corresponding representations and operations. In \cref{app:comonadic}, we argue that this can be viewed as an instantiation of a comonadic strategy for homomorphism indistinguishability in the case of the pebbling comonad~\cite{abramsky_pebbling_2017,dawar_lovasz_2021}. 
We believe that our approach may facilitate the study of matrix equations for homomorphism indistinguishability over graph classes which are more intricate with respect to labellings, operations, and generation.

\section{Preliminaries}
\label{sec:prelims}

We assume familiarity with the standard notation and terminology from graph theory. The reader is referred to the monograph \cite{Diestel} for a detailed presentation. In what follows, all graphs are assumed to be simple undirected graphs. As usual, let $[n]$ for $n \in \mathbb{N}$ denote the set $\{1,\dots,n\}$  and let $\mathbb{C}$ denote the field of complex numbers. 

\subsection{Bilabelled Graphs and Homomorphism Matrices}
A \emph{homomorphism} $h \colon F \to G$ between graphs $F$ and $G$ is a map $V(F) \to V(G)$ such that $h(u)h(v) \in E(G)$ for all $uv \in E(F)$.
For $k, \ell \in \mathbb{N}$,
a \emph{$(k,\ell)$-bilabelled graph} is a tuple $\boldsymbol{F} = (F, \vec{u}, \vec{v})$ of a graph $F$ and tuples of \emph{in-labelled} vertices $\vec{u} \in V(F)^k$  and of \emph{out-labelled} vertices $\vec{v} \in V(F)^\ell$.
For every graph~$G$, it gives rise to a \emph{homomorphism matrix} $\boldsymbol{F}_G \in \mathbb{C}^{V(G)^k \times V(G)^\ell}$ where $\boldsymbol{F}_G(\vec{x}, \vec{y})$ is the number of homomorphisms $h \colon F \to G$ such that $h(\vec{u}) = \vec{x}$ and $h(\vec{v}) = \vec{y}$ entry-wise. 
These definitions readily extend to \emph{$k$-labelled graphs} $\boldsymbol{F} = (F, \vec{u})$ and their \emph{homomorphism vectors} $\boldsymbol{F}_G \in \mathbb{C}^{V(G)^k}$.
The \emph{reverse} of an $(\ell_1, \ell_2)$-bilabelled
graph $\boldsymbol{F} =(F,\vec{u}, \vec{v})$ is defined to be the $(\ell_2,\ell_1)$-bilabelled graph $\boldsymbol{F}^* = (F, \vec{v}, \vec{u})$ with roles of in- and out-labels interchanged. 
The \emph{series composition} of an $(\ell_1, \ell_2)$-bilabelled graph $\boldsymbol{F} = (F, \vec{u}, \vec{v})$ and an $(\ell_2, \ell_3)$-bilabelled graph $\boldsymbol{F}' = (F', \vec{u}', \vec{v}')$, denoted by $\boldsymbol{F} \cdot \boldsymbol{F}'$ is the $(\ell_1, \ell_3)$-bilabelled graph obtained by taking the disjoint union of $F$ and $F'$ and identifying for all $i \in [\ell_2]$ the vertices $\vec{v}_i$ and $\vec{u}'_i$. The in-labels of $\boldsymbol{F} \cdot \boldsymbol{F}'$ lie on $\vec{u}$ while its out-labels are positioned on $\vec{v}'$.
See \cite{mancinska_quantum_2019,grohe_homomorphism_2021} for further details.

\subsection{Pebble Forest Covers} We recall the following notions from  \cite{abramsky_relating_2021}. A \emph{chain} in a poset $\mathcal{P}=(P, \leq)$ is a subset 
$C \subseteq P$ such that for all $x,y \in C$, either $x \leq y$ or $y \leq x$. 
A \emph{forest} is a poset $\mathcal{F} = (F, \leq)$ such that, for all $x \in \mathcal{F}$, the set $\{y \in F \,\vert\, y \leq x\}$ of predecessors of $x$ is a finite chain. 
The \emph{depth} of a forest $\mathcal{F}$ is the length of a maximal chain in $\mathcal{F}$. 
A \emph{tree} is a forest $(F,\leq)$ with a \emph{root} element $x$ such that $x \leq y$ for all $y \in F$.
A \emph{leaf} of a tree $(T,\leq)$ is an element $x$ of $T$ such the set $\{y \in T \,\vert\, x \leq y\}$ of successors of $x$ is $\{x\}$. 
The \emph{depth of a leaf $x$} of a tree $(T,\leq)$ is the length of the unique maximal chain from the root element to the leaf element. 

A \emph{forest cover} of a graph $G$ is a forest $\mathcal{F}=(F,\leq)$ with $F = V(G)$
such that for every edge $uv \in E(G)$, it must hold that either $u \leq v $ or $ v \leq u$. 
A \emph{tree cover} of a graph $G$ is a tree $\mathcal{T}$ which is a forest cover of $G$. 

Let $\mathcal{F} = (F, \leq)$ be a forest cover for a graph $G$. 
A \emph{$k$-pebbling function} for the pair $(G, \mathcal{F})$ is a mapping $p \colon V(G) \to [k]$ satisfying that for all vertices $u, v \in V(G)$ such that $u\leq v$ and $uv \in E(G)$, it holds that $p(x) \neq p(u)$ for every $x \neq u$ such that $u \leq x \leq v$. 
A \emph{$k$-pebble forest cover} of $G$ is a pair $(\mathcal{F},p)$ such that $\mathcal{F}$ is a forest cover for $G$ and $p$ is a $k$-pebbling function for $(G,\mathcal{F})$. It can be shown \cite[Theorem~6.4]{abramsky_relating_2021} that a graph $G$ has a $k$-pebble forest cover iff it has treewidth less than $k$.

\begin{definition}[Pebble Forest Covers for Labelled Graphs] \label{def:pfc-labelled}
	Let $k \geq \ell \geq 1$.
	A $k$-pebble forest cover of an $\ell$-labelled graph $\boldsymbol{F}=(F,\vec{u})$ is a $k$-pebble forest cover $(\mathcal{F}, p)$ of $F$ such that
	\begin{enumerate}
		\item the labelled vertices $\vec{u}_1, \dots, \vec{u}_\ell$ form a chain, i.e.\@ are mutually comparable in $\mathcal{F}$,
		\item all other vertices lie below the labelled vertices, i.e.\@ for all $v \in V(F) \setminus \{\vec{u}_1, \dots, \vec{u}_\ell\}$ it holds that $v \geq \vec{u}_1, \dots, \vec{u}_\ell$,
		\item the pebbling function $p$ is injective on $\{\vec{u}_1, \dots, \vec{u}_\ell\}$, i.e.\@ if $p(\vec{u}_i) = p(\vec{u}_j)$ then $\vec{u}_i = \vec{u}_j$ for all $i, j \in [\ell]$.
	\end{enumerate}
	The \emph{depth} of such an $(\mathcal{F}, p)$ is the length of the longest chain in $V(F) \setminus \{\vec{u}_1, \dots, \vec{u}_\ell\}$.
	Write $\mathcal{LPFC}_k^d$ for the class of $k$-labelled graphs admitting a $(k+1)$-pebble forest cover of depth $d$.
	Write $\mathcal{PFC}_k^d$ for the unlabelled graphs underlying the labelled graphs in $\mathcal{LPFC}_k^d$.
\end{definition}

\subsection{Weisfeiler--Leman}
Let $k \geq 1$.
Given a graph $G$, the \emph{$k$-dimensional Weisfeiler--Leman algorithm} ($k$-\WL) colours $k$-tuples of its vertices in the following iterative manner.
Let $\atp$ denote a function which maps a graph $G$ and $\vec{x} \in V(G)^\ell$ to some value such that $\atp(G, \vec{x}) = \atp(H, \vec{y})$ if and only if $\vec{x}_i \mapsto \vec{y}_i$, $i \in [\ell]$ is a partial isomorphism, cf.\@ \cite[Section~IV]{grohe_logic_2021}.
Define $\WL_k^0(G, \vec{v}) \coloneqq \atp(G, \vec{v})$ for all $\vec{v} \in V(G)$ and for $d \geq 0$,
\[ \WL_k^{d+1}(G, \vec{v}) \coloneqq \left( \WL_k^d(G, \vec{v}), \multiset{ \left( \atp(G, \vec{v}w), \WL_k^d(G, \vec{v}[1/w]) , \dots, \WL_k^d(G,\vec{v}[k/w]) \right) \ \middle|\ w \in V(G)  } \right).\]
Here, $\vec{v}[i/w]$ for $i \in [k]$ denotes the $k$-tuple obtained from $\vec{v}$ by replacing its $i$-th entry by $w$.
The procedure terminates when the colouring is stable, i.e.\@ once the partition of $V(G)^k$ into colour classes prescribed by this colouring is not refined by an additional iteration of $k$-\WL. 
We write $\WL_k^\infty$ for this stable colouring.
For further details, see e.g.\@ \cite{cai_optimal_1992,grohe_logic_2021}.

Two graphs $G$ and $H$ are \emph{indistinguishable by $k$-\WL after $d$~iterations} if the multisets of the $\WL_k^d(G, \vec{v})$ for $\vec{v} \in V(G)^k$ and of the $\WL_k^d(H, \vec{w})$ for $\vec{w} \in V(H)^k$ are the same.
This equivalence relation is a well-studied notion which foremost admits a logical characterisation.
Let $\mathsf{C}_k$ denote the $k$-variable fragment of first-order logic with counting quantifiers.  Let $\mathsf{C}_k^d$ denote the quantifier rank-$d$ fragment of $\mathsf{C}_k$. 
The equivalence of \cref{camb1,camb2} in the following theorem characterising $k$-\WL indistinguishability is implicit in \cite{abramsky_pebbling_2017,dawar_lovasz_2021}. The equivalence of \cref{camb2,camb3} is due to \cite{cai_optimal_1992}. For completeness, we give a self-contained proof in \cref{app:thm:camb}.

\begin{theorem}[{\cite[Section IV.B]{dawar_lovasz_2021}, \cite[Theorem~18]{abramsky_pebbling_2017}, \cite[Theorem~5.2]{cai_optimal_1992}}]\label{thm:camb}
Let $k \geq 1$ and $d \geq 0$.
Given two graphs $G$ and $H$ with $\vec{v} \in V(G)^k$ and $\vec{w} \in V(H)^k$, the following are equivalent. 
\begin{enumerate}
	\item $\WL_k^d(G, \vec{v}) = \WL_k^d(H, \vec{w})$, \label{camb3}
	\item For all $\phi(\vec{x}) \in \mathsf{C}_{k+1}^d$ with $k$~free variables, $G \models \phi(\vec{v})$ if and only if $H \models \phi(\vec{w})$,\label{camb2}
	\item For all $\boldsymbol{F} \in \mathcal{LPFC}_k^d$, it holds that $\boldsymbol{F}_G(\vec{v}) = \boldsymbol{F}_H(\vec{w})$. \label{camb1}
\end{enumerate}
\end{theorem}

\begin{corollary} \label{cor:camb}
	Let $k \geq 1$ and $d \geq 0$. Two graphs $G$ and $H$ are indistinguishable after $d$ iterations of $k$-\WL if and only if they are homomorphism indistinguishable over $\mathcal{PFC}_k^d$.
\end{corollary}

\subsection{Specht--Wiegmann Theorem}
A variation of the Specht--Wiegmann Theorem~\cite{specht_zur_1940,wiegmann_necessary_1961} has proven useful for deriving matrix equations for homomorphism indistinguishability in \cite{grohe_homomorphism_2021}. It also serves as the representation-theoretic core of the proof of \cref{thm:main-intro}.
In order to state it, let $\Sigma_{2m}$ denote the finite alphabet $\{x_i, y_i \mid i \in [m]\}$ for some $m \in \mathbb{N}$. Let $\Gamma_{2m} \coloneqq \Sigma_{2m}^*$ denote the set of all finite words over $\Sigma_{2m}$.
Given a family of complex-valued matrices $\matvec{A} = (A_1, \dots, A_m)$, a word $w \in \Gamma_{2m}$ can be represented as a matrix $w_{\matvec{A}}$ by sending $x_i$ to $A_i$, $y_i$ to $A_i^*$, the conjugate-transpose of $A_i$, and interpreting concatenation as matrix product.
Let $\soe(M)$ denote the sum of all entries of a complex matrix $M$. 
Recall that a matrix is pseudo-stochastic if its row and column sums are one.

\begin{theorem}[{\cite[Theorem~20]{grohe_homomorphism_2021}}]\label{thm:soe}
Let $I$ and $J$ be finite index sets.
Let $\matvec{A} = (A_1,\dots,A_m)$ and $\matvec{B} = (B_1,\dots,B_m)$  be two sequences of matrices such that $A_i \in \CC^{I \times I}$ and $B_i \in \CC^{J \times J}$  for $i \in [m]$. Then the following are equivalent: 
\begin{enumerate}
 \item There exists a pseudo-stochastic matrix $X \in \CC^{J \times I}$ such that $X A_i  =  B_iX $ and $X A_i^*  =  B_i^* X$ for all $i \in [m]$. 
 \item For every word $w \in \Gamma_{2m}$, $\soe(w_{\matvec{A}}) = \soe(w_{\matvec{B}})$. 
\end{enumerate} 
\end{theorem}

\section{Graph Spectra between 1-WL and 2-WL}
\label{sec:onetwowl}

It is well-known \cite[discussion before Theorem 2.10]{dawar_generalizations_2020} that $2$-\WL indistinguishable graphs are cospectral with respect to all matrices that lie in their adjacency algebras, which are instances of coherent algebras, cf.\@ \cref{ssec:adjacency-algebra}. This argument extends to graph parameters, like hitting times of random walks, which are determined by such spectral data \cite{lovasz_random_1996,godsil_equiarboreal_1981}. 

We aim at systematising these results by introducing $(1,1)$-\WL \emph{indistinguishability} in \cref{def:oneonewl}, a novel algorithmic equivalence relation on graphs strictly weaker than $2$-\WL. 
In \cref{ssec:eqmatmaps}, we define the notion of \emph{equitable matrix maps} and derive a spectral characterisation of $(1,1)$-\WL in terms of the spectra of these maps, stated as \cref{thm:specinv,rem:specinv}.
In \cref{ssec:septhms}, we prove \cref{ex:weaker} which states that $(1,1)$-\WL is strictly weaker than $2$-\WL in terms of distinguishing non-isomorphic graphs. 
In conjunction with \cref{thm:specinv}, this yields the desired \cref{thm:specinv-intro} which constitutes the main result of this section.
In \cref{ssec:app-algoneonewl}, $(1,1)$-\WL indistinguishability is characterised in terms of matrix equations and homomorphism indistinguishability. 

\subsection{Spectral Characterisation of $(1,1)$-\WL Indistinguishable Graphs}
\label{sec:emm}\label{ssec:eqmatmaps}

Let $G$ be a graph. 
Let $\mathcal{C}$ denote a \emph{vertex colouring}, i.e.\@ a map from $V(G)$ to a (unspecified) set of colours. 
Abusing notation, we sometimes treat $\mathcal{C}$ as a partition of $V(G)$ into colour classes $C \in \mathcal{C}$.
Define $\mathbb{C}\mathcal{C} \leq \mathbb{C}^{V(G)}$ to be the vector space spanned by colour class indicator vectors~$\boldsymbol{1}_C$, $C \in \mathcal{C}$.
The projection onto $\mathbb{C}\mathcal{C}$, typically denoted by~$M$, is called the \emph{averaging matrix} of $\mathcal{C}$ as $M = \sum_{C \in \mathcal{C}} \abs{C}^{-1} \boldsymbol{1}_C\boldsymbol{1}_C^T$ acts by averaging vectors over colour classes.
For $v \in V(G)$,
let $\mathcal{C}_G^v$ denote the vertex colouring obtained by running $1$-\WL on the vertex-individualised graph $G_v$, i.e.\@ the coloured graph obtained from $G$ by colouring $v$ orange and all other vertices grey.
\begin{definition}\label{def:emm}
A matrix map $\phi$ which associates with every graph $G$ a matrix in $\mathbb{C}^{V(G) \times V(G)}$ is an \emph{equitable matrix map} if for all $(1,1)$-\WL indistinguishable graphs $G$ and $H$ and all colour-preserving bijections $\pi \colon V(G) \to V(H)$ it holds that
\begin{enumerate}[label=(E\arabic*)]
\item for all $v\in V(G)$, $M^v \phi(G) = \phi(G) M^v$ for $M^v$ the averaging matrix of $\mathcal{C}_G^v$,\label{l1}
\item for all $v \in V(G)$ and permutation transformations $P \colon \mathbb{C}^{V(G)} \to \mathbb{C}^{V(H)}$ such that for all $x \in V(G)$ there exists $y \in V(H)$ such that $Pe_x = e_y$ and $\mathcal{C}_G^v(x) = \mathcal{C}_H^{\pi(v)}(y)$ it holds that $P\phi(G)M^v = \phi(H)PM^{v}$.\label{l2}
\end{enumerate}
Write $\mathfrak{E}$ for the set of all equitable matrix maps.
\end{definition}

Assertion~\ref{l1} guarantees that equitable matrix maps respect the subspaces $\mathbb{C}\mathcal{C}^v_G \leq \mathbb{C}^{V(G)}$ induced by $(1,1)$-\WL. That a similar statement holds for the adjacency matrix and $1$-\WL was first observed in \cite[Lemmata~3.1,~3.2]{furer_graph_1995}. \ref{l2} relates the entries of equitable matrix maps evaluated at $(1,1)$-\WL indistinguishable graphs turning them essentially into functions of the $(1,1)$-\WL colours.
\begin{lemma}\label{lem:adjacencymatrix}
The adjacency matrix map~$A$ and the degree matrix map~$D$ are equitable matrix maps.
\end{lemma}

\begin{proof}
	For \ref{l1}, let $G$ be a graph with equitable vertex colouring $\mathcal{C}$ whose averaging matrix is~$M$. 
	For sets $S, T\subseteq V(G)$ write $E_G(S, T)$ for the number of edges connecting vertices in $S$ with vertices in $T$.
	Then for vertices $u, v \in V(G)$, letting $\mathcal{C}(v) \subseteq V(G)$ denote the $\mathcal{C}$-colour class to which $v$ belongs,
	\[
	(M A(G))(u,v)
	= \frac{E_G(\{v\},\mathcal{C}(u))}{\abs{\mathcal{C}(u)}} 
	= \frac{E_G(\mathcal{C}(v), \mathcal{C}(u))}{\abs{\mathcal{C}(u)} \abs{\mathcal{C}(v)}}
	= \frac{E_G(\{u\}, \mathcal{C}(v))}{\abs{\mathcal{C}(v)}} 
	= (A(G) M)(u,v).
	\]
	
	For \ref{l2}, let $H$ be a second graph such that $G$ and $H$ are $(1,1)$-\WL indistinguishable. Let $\pi \colon V(G) \to V(H)$ be colour-preserving, $v \in V(G)$, and $M^v$ and $P$ as in \ref{l2}. Let $C \in \mathcal{C}_G^v$ be a colour class. Since $G$ and $H$ are $(1,1)$-\WL indistinguishable, there is a corresponding colour class $C' \in \mathcal{C}_H^{\pi(v)}$. Then
	\[
	P A(G) \boldsymbol{1}_C 
	= \sum_{D \in \mathcal{C}_G^v} E_G(D, C) P\boldsymbol{1}_D
	= \sum_{D' \in \mathcal{C}_H^{\pi(v)}} E_H(D', C') \boldsymbol{1}_{D'}
	= A(H) P\boldsymbol{1}_C
	\]
	since $P\boldsymbol{1}_C = \boldsymbol{1}_{C'}$. This implies \ref{l2}.
	Similar arguments show that $D$ is an equitable matrix map.
\end{proof}

The following \cref{lem:closure} asserts that $\mathfrak{E}$ has enjoyable algebraic properties.

\begin{lemma}\label{lem:closure}
	$\mathfrak{E}$ is closed under linear combinations, matrix products, and transposes.
\end{lemma}
\begin{proof}
	Since \ref{l1} and \ref{l2} are linear conditions, the first claim is immediate. 
	
	For closure under matrix products, let $\phi, \psi \in \mathfrak{E}$. Let $G$ be any graph and $M^v$ be as in \ref{l1}. Then
	$M^v(\phi \cdot \psi)(G) = M^v\phi(G) \psi(G) = \phi(G)M^v\psi(G) = \phi(G) \psi(G) M^v = (\phi \cdot \psi)(G) M^v$ by \ref{l1}.
	Hence, $\phi \cdot \psi$ satisfies \ref{l1}.
	For \ref{l2}, let $H$ and $P$ be as in \cref{def:emm}. Then 
	\begin{align*} 
	P (\phi \cdot \psi)(G)M^v & = P \phi(G) \psi(G) M^v \overset{\text{\ref{l1}}}{=} P \phi(G) M^v \psi(G) \overset{\text{\ref{l2}}}{=} \phi(H) PM^v \psi(G) \\
	&\overset{\text{\ref{l1}}}{=} \phi(H) P \psi(G)M^v \overset{\text{\ref{l2}}}{=} (\phi \cdot \psi)(H) PM^v.
	\end{align*}
	Hence, \ref{l2} holds.
	
	For closure under transposes, let $\phi \in \mathfrak{E}$.
	Observe that the $M^v$ are symmetric matrices. Hence, $\phi^T$ satisfies \ref{l1}.
	Furthermore, if $P \colon \mathbb{C}^{V(G)} \to \mathbb{C}^{V(H)}$ is a permutation transformation such that for all $x \in V(G)$ there exists $y \in V(H)$ such that $Pe_x = e_y$ and $\mathcal{C}_G^v(x) = \mathcal{C}_H^{\pi(v)}(y)$ then $P^T \colon \mathbb{C}^{V(H)} \to \mathbb{C}^{V(G)}$ is such that for all $y \in V(H)$ there exists $x \in V(G)$ such that $Pe_y = e_x$ and $\mathcal{C}_G^{\pi^{-1}(v)}(x) = \mathcal{C}_H^{v}(y)$ where $\pi^{-1} \colon V(H) \to V(G)$ is a colour-preserving bijection. It follows that $P^T\phi(H)M^v = \phi(G) P^T M^v$. Hence, $M^v \phi^T(H) P = M^v P \phi^T(G)$.
	It remains to observe that $M^v P = PM^{\pi^{-1}(v)}$ and that \ref{l2} holds.
\end{proof}

\Cref{lem:adjacencymatrix,lem:closure} imply that many frequently studied matrices associated with graphs are equitable matrix maps, cf.\@ \cite{van_dam_which_2003,cvetkovic_signless_2007,von_luxburg_tutorial_2007}.

\begin{example}
The graph Laplacian $L \coloneqq D-A$, the signless Laplacian $\abs{L} \coloneqq D + A$,  the adjacency matrix of the complement $A^c \coloneqq J - A - I$, the Seidel matrix $S \coloneqq A^c - A$, the random walk normalised Laplacian $L_\text{rw} \coloneqq D^{-1}A$, and the symmetric normalised Laplacian $L_\text{sym} \coloneqq D^{-1/2}LD^{-1/2}$ are equitable matrix maps.
\end{example}

We now turn to the spectral properties of equitable matrix maps. Key here is \cref{lem:proj}, which allows to pass from an equitable matrix map to the projections onto its eigenspaces. As a corollary, it follows for example that \emph{time-$t$ heat kernels} $H_t(G)$ for $t \in \mathbb{R}$ as considered in \cite{feldman_weisfeiler_2022} are equitable matrix maps. In the notation of \cref{lem:proj}, they are defined as the matrix map $H_t(G) \coloneqq \sum_{\lambda \in \mathbb{C}} e^{-\lambda t} L_\lambda(G)$ where $L$ is the graph Laplacian.

\begin{lemma}\label{lem:proj}
Let $\phi \in \mathfrak{E}$ and $\lambda \in \mathbb{C}$. For a graph $G$, define $\phi_\lambda(G)$ to be the projection onto the space $E_\lambda$ spanned by the eigenvectors of $\phi(G)$ with eigenvalue $\lambda$. Then $\phi_\lambda \in \mathfrak{E}$.
\end{lemma}

\begin{proof}
	By~\ref{l1}, $\phi(G)$ preserves the direct sum decomposition $\mathbb{C}^{V(G)} = \mathbb{C}\mathcal{C}^v_G \oplus (\mathbb{C}\mathcal{C}^v_G)^\perp$ for the individualised vertex colouring of any $v \in V(G)$.
	Furthermore, $M^v$ preserves the direct sum decomposition $\mathbb{C}^{V(G)} = E_\lambda \oplus (E_\lambda)^\perp$. 
	Indeed, if $x \in E_\lambda$ then $\phi(G)M^vx = M^v\phi(G)x = \lambda M^vx$, so $M^v x \in E_\lambda$.
	This implies~\ref{l1}.
	
	For \ref{l2}, observe that if $x \in E_\lambda$ then $\lambda P M^vx = P\phi(G)M^v x = \phi(H)PM^vx = \phi(H)PM^vx$ by~\ref{l2}, so $PM^vx$ is in the eigenspace of $\phi(H)$ with eigenvalue $\lambda$. Hence, $P\phi_\lambda(G)M^v x = PM^v x = \phi_\lambda(H) PM^vx$. If $x \in (E_\lambda)^\perp$ then $P\phi_\lambda(G)M^v x = 0 = \phi_\lambda(H) PM^vx$. This implies~\ref{l2}.
\end{proof}

In \cref{thm:colentry}, the correspondence between the colouring algorithm $(1,1)$-\WL and equitable matrix maps is established.
As a consequence, it is proved in \cref{thm:specinv} that a spectral invariant generalising Fürer's spectral invariant (\cref{def:furer-inv}, \cite{furer_power_2010}) is determined by $(1,1)$-\WL.

\begin{proposition}\label{thm:colentry}
Let $\phi \in \mathfrak{E}$.
Let $G$ and $H$ be graphs, $v, w \in V(G)$, and $x, y \in V(H)$. 
If the $(1,1)$-\WL colours $\mathcal{C}_G^v(w)$ and $\mathcal{C}_H^x(y)$ coincide 
then the $(v,w)$-th entry of $\phi(G)$
and the $(x,y)$-th entry of $\phi(H)$ are equal.
\end{proposition}
\begin{proof}
	Let $P \colon \mathbb{C}^{V(G)} \to \mathbb{C}^{V(H)}$ be a permutation transformation such that for all $u \in V(G)$ and $z \in V(H)$ such that $Pe_u = Pe_z$ it holds that $\mathcal{C}^v_G(u) = \mathcal{C}^x_H(z)$. Then,
	\begin{align*}
		\phi(G, v, w) 
		&= e_v^T \phi(G) e_w 
		= e_v^T M^v \phi(G) e_w 
		\overset{\text{\ref{l1}}}{=} e_v^T \phi(G) M^v e_w 
		= e_x^T P^T \phi(G) M^v e_w \\
		&\overset{\text{\ref{l2}}}{=} e_x^T \phi(H) P M^v e_w 
		= e_x^T \phi(H) M^{x} e_y 
		\overset{\text{\ref{l1}}}{=} e_x^T M^x \phi(H) e_y
		= \phi(H, x, y),
	\end{align*}
	since $M^v e_v = e_v$ for all $v \in V(G)$ because $v$ lies in a singleton colour class in $\mathcal{C}^v_G$.
\end{proof}

\begin{theorem}\label{thm:specinv}\label{cor:cospectrality}
If $G$ and $H$ are $(1,1)$-\WL indistinguishable graphs then $\mathcal{S}(G) = \mathcal{S}(H)$ for
\[
\mathcal{S}(G) \coloneqq \left(
\phi \mapsto \Spec \phi(G),
\multiset{S^G_v \mid v \in V(G)} \right)
\]
where $S^G_v \coloneqq \left( \phi \mapsto \phi(G, v,v), \multiset{\phi \mapsto \phi(G, v, w) \mid w \in V(G)}\right)$
and the maps indicated by $\mapsto$ have domain $\mathfrak{E}$.
In particular, for every $\phi \in \mathfrak{E}$, the matrices $\phi(G)$ and $\phi(H)$ are cospectral.
\end{theorem}

\begin{proof}
	It is first shown that the second components $\mathcal{S}(G)$ and $\mathcal{S}(H)$ are equal by constructing suitable bijections which witness the equality of the involved multisets.
	Since $G$ and $H$ are $(1,1)$-\WL indistinguishable, there exist bijections $\pi \colon V(G) \to V(H)$ and $\sigma_v \colon V(G) \to V(H)$ for every $v \in V(G)$ such that $\mathcal{C}^v_G(w) = \mathcal{C}^{\pi(v)}_H(\sigma_v(w))$ for all $v, w \in V(G)$. Observe that $\sigma_v(v) = \pi(v)$ for all $v\in V(G)$. By \cref{thm:colentry}, $\phi(G, v, v) = \phi(H, \pi(v), \pi(v))$ and $\phi(G, v, w) = \phi(H, \pi(v), \sigma_v(w))$. This implies that the second component of $\mathcal{S}(G)$ is determined by $(1,1)$-\WL.
	
	It remains to show that the first component of $\mathcal{S}(G)$ is determined by its second component. Let $\phi \in \mathfrak{E}$. By \cref{lem:proj}, for every $\lambda \in \mathbb{C}$, the projection $\phi_\lambda$ onto the span $E_\lambda$ of the eigenvectors of $\phi$ with eigenvalue $\lambda$ is in $\mathfrak{E}$. Observe that $\dim E_\lambda = \tr \phi_\lambda(G) = \sum_{v \in V(G)} \phi_\lambda(G, v, v)$. This is determined by the second component of $\mathcal{S}(G)$. Hence, for every $\phi \in \mathfrak{E}$ and every $\lambda \in \mathbb{C}$ the spans of the eigenvectors of $\phi(G)$ and of the eigenvectors of $\phi(H)$ with eigenvalue $\lambda$ have the same dimension. This implies that $\Spec \phi(G) = \Spec \phi(H)$.
\end{proof}

\begin{remark}\label{rem:specinv}
There is a rather trivial converse to \cref{thm:specinv,cor:cospectrality}: For a graph $G$, let $\omega_G$ be the matrix map that takes as value either the zero matrix or the identity matrix of appropriate size such that $\omega_G(G) = \omega_G(H)$ if and only if $G$ and $H$ are $(1,1)$-\WL indistinguishable. Clearly, $\omega_G \in \mathfrak{E}$.
\end{remark}

\subsection{Separation of $(1,1)$-\WL and $2$-\WL} 
\label{ssec:septhms}

It is well-known that there exist $1$-\WL indistinguishable graphs which are not cospectral, e.g.\@ the cycle graph on six vertices and the disjoint union of two triangle graphs. 
Hence, $1$-\WL is strictly less powerful than $(1,1)$-\WL in terms of its ability to distinguish non-isomorphic graphs. The following theorem shows that $(1,1)$-\WL is strictly less powerful than $2$-\WL in terms of distinguishing power.\footnote{This fact is probably known to many experts in the field. However, we could not find any reference for it.}
In conjunction with \cref{thm:specinv}, this yields the desired \cref{thm:specinv-intro}. The graphs constructed in \cref{ex:weaker} are depicted in \cref{fig:cntex}.

\begin{theorem}\label{ex:weaker}
There exist graphs $G$ and $H$ which are $(1,1)$-\WL indistinguishable but not $2$-\WL indistinguishable.
\end{theorem}

\begin{figure}[t]
	\centering
	\captionsetup[subfigure]{justification=centering}
	\tikzset{
		vertex/.style = {fill,circle,inner sep=0pt,minimum size=5pt},
		edge/.style = {-,thick},
		link/.style = {-, semithick},
		lbl/.style={color=lightgray},
		node/.style = {inner sep=0.7pt,circle,draw,fill}
	}
\begin{subfigure}[t]{0.35 \textwidth}
\centering
\begin{tikzpicture}
	\node[vertex] (a) [label = {}] {} ;
	\node[vertex] (b) [label = {}] [right = of a] {}; 
	\node[vertex] (c) [label = {}] [below = of b] {}; 
	\node[vertex] (d) [label = {}] [left = of c] {}; 
	
	\draw[edge] (a) -- (b) {}; 
	\draw[edge] (b) -- (c) {}; 
	\draw[edge] (c) -- (d) {}; 
	\draw[edge] (d) -- (a) {}; 
	
\begin{scope}[xshift = -0.99cm, yshift=0.65cm]
		\node[node] (aa) at (60:0.25cm) {};
		\node[node] (bb) at (180:0.25cm) {};
		\node[node] (cc) at (300:0.25cm) {};
		\draw[link] (aa) -- (bb) {};
		\draw[link] (cc) -- (bb) {};
		\draw[link] (aa) -- (cc) {};
		
		\begin{scope}[xshift = 0.5cm]
			\node[node] (aa) at (0:0.25cm) {};
			\node[node] (bb) at (120:0.25cm) {};
			\node[node] (cc) at (240:0.25cm) {};
			\draw[link] (aa) -- (bb) {};
			\draw[link] (cc) -- (bb) {};
			\draw[link] (aa) -- (cc) {};
		\end{scope}
	\end{scope}
	
\begin{scope}[xshift = 1.65cm, yshift=0.65cm]
		\node[node] (aa) at (60:0.25cm) {};
		\node[node] (bb) at (180:0.25cm) {};
		\node[node] (cc) at (300:0.25cm) {};
		\draw[link] (aa) -- (bb) {};
		\draw[link] (cc) -- (bb) {};
		\draw[link] (aa) -- (cc) {};
		
		\begin{scope}[xshift = 0.5cm]
			\node[node] (aa) at (0:0.25cm) {};
			\node[node] (bb) at (120:0.25cm) {};
			\node[node] (cc) at (240:0.25cm) {};
			\draw[link] (aa) -- (bb) {};
			\draw[link] (cc) -- (bb) {};
			\draw[link] (aa) -- (cc) {};
		\end{scope}
	\end{scope}
	
\begin{scope}[xshift = 1.80cm, yshift=-1.80cm]
		\newdimen\R
		\R=0.3cm
		\draw (0:\R) \foreach \x in {60,120,...,360} {  -- (\x:\R) };
		\foreach \x/\l/\p in
		{ 60/{}/above, 120/{}/above, 180/{}/left, 240/{}/below, 300/{}/below, 360/{}/right }
		\node[inner sep=0.7pt,circle,draw,fill,label={\p:\l}] at (\x:\R) {};
	\end{scope}
	
\begin{scope}[xshift = -0.65cm, yshift=-1.8cm]
		\newdimen\R
		\R=0.3cm
		\draw (0:\R) \foreach \x in {60,120,...,360} {  -- (\x:\R) };
		\foreach \x/\l/\p in
		{ 60/{}/above, 120/{}/above, 180/{}/left, 240/{}/below, 300/{}/below, 360/{}/right }
		\node[inner sep=0.7pt,circle,draw,fill,label={\p:\l}] at (\x:\R) {};
	\end{scope}
	
	\begin{scope}[on background layer]
\draw[dashed, thin ,fill = black!10] (0,0) -- (105:15mm)
		arc [start angle=105, end angle=165, radius=15mm] -- (0,0); 
		\draw[dashed, thin ,fill = black!10, xshift = 1.15cm] (0,0) -- (15:15mm)
		arc [start angle=15, end angle=75, radius=15mm] -- (0,0); 
		\draw[dashed, thin ,fill = black!10, yshift = -1.15cm] (0,0) -- (195:15mm)
		arc [start angle=195, end angle=255, radius=15mm] -- (0,0); 
		\draw[dashed, thin ,fill = black!10, xshift = 1.15cm, yshift= -1.15cm] 
		(0,0) -- (285:15mm)
		arc [start angle=285, end angle=345, radius=15mm] -- (0,0); 
		\node[] at (-1.5cm,1.5cm) {$Y_2$};  
		\node[] at (2.5cm,1.5cm) {$Y_1$};  
		\node[] at (-1.5cm,-2.5cm) {$X_1$};  
		\node[] at (2.5cm,-2.5cm) {$X_2$};  
	\end{scope}
\end{tikzpicture}
\caption{Graph $G$}
	\label{fig:a}
\end{subfigure}
\quad\quad\quad
\begin{subfigure}[t]{0.35 \textwidth}
		\centering
		\begin{tikzpicture}
			
			\node[vertex] (a) [label = {}] {} ;
			\node[vertex] (b) [label = {}] [right = of a] {}; 
			\node[vertex] (c) [label = {}] [below = of b] {}; 
			\node[vertex] (d) [label = {}] [left = of c] {}; 
			
			\draw[edge] (a) -- (b) {}; 
			\draw[edge] (b) -- (c) {}; 
			\draw[edge] (c) -- (d) {}; 
			\draw[edge] (d) -- (a) {}; 
			
\begin{scope}[xshift = -0.99cm, yshift=0.65cm]
				\node[node] (aa) at (60:0.25cm) {};
				\node[node] (bb) at (180:0.25cm) {};
				\node[node] (cc) at (300:0.25cm) {};
				\draw[link] (aa) -- (bb) {};
				\draw[link] (cc) -- (bb) {};
				\draw[link] (aa) -- (cc) {};
				
				\begin{scope}[xshift = 0.5cm]
					\node[node] (aa) at (0:0.25cm) {};
					\node[node] (bb) at (120:0.25cm) {};
					\node[node] (cc) at (240:0.25cm) {};
					\draw[link] (aa) -- (bb) {};
					\draw[link] (cc) -- (bb) {};
					\draw[link] (aa) -- (cc) {};
				\end{scope}
			\end{scope}
			
\begin{scope}[xshift = 1.65cm, yshift=-1.75cm]
				\node[node] (aa) at (60:0.25cm) {};
				\node[node] (bb) at (180:0.25cm) {};
				\node[node] (cc) at (300:0.25cm) {};
				\draw[link] (aa) -- (bb) {};
				\draw[link] (cc) -- (bb) {};
				\draw[link] (aa) -- (cc) {};
				
				\begin{scope}[xshift = 0.5cm]
					\node[node] (aa) at (0:0.25cm) {};
					\node[node] (bb) at (120:0.25cm) {};
					\node[node] (cc) at (240:0.25cm) {};
					\draw[link] (aa) -- (bb) {};
					\draw[link] (cc) -- (bb) {};
					\draw[link] (aa) -- (cc) {};
				\end{scope}
			\end{scope}
			
\begin{scope}[xshift = 1.85cm, yshift=0.70cm]
				\newdimen\R
				\R=0.3cm
				\draw (0:\R) \foreach \x in {60,120,...,360} {  -- (\x:\R) };
				\foreach \x/\l/\p in
				{ 60/{}/above, 120/{}/above, 180/{}/left, 240/{}/below, 300/{}/below, 360/{}/right }
				\node[inner sep=0.7pt,circle,draw,fill,label={\p:\l}] at (\x:\R) {};
			\end{scope}
			
\begin{scope}[xshift = -0.65cm, yshift=-1.8cm]
				\newdimen\R
				\R=0.3cm
				\draw (0:\R) \foreach \x in {60,120,...,360} {  -- (\x:\R) };
				\foreach \x/\l/\p in
				{ 60/{}/above, 120/{}/above, 180/{}/left, 240/{}/below, 300/{}/below, 360/{}/right }
				\node[inner sep=0.7pt,circle,draw,fill,label={\p:\l}] at (\x:\R) {};
			\end{scope}
			
			\begin{scope}[on background layer]
\draw[dashed, thin ,fill = black!10] (0,0) -- (105:15mm)
				arc [start angle=105, end angle=165, radius=15mm] -- (0,0); 
				\node[] at (-1.5cm,1.5cm) {$Y_1$};  
				
				\draw[dashed, thin ,fill = black!10, xshift = 1.15cm] (0,0) -- (15:15mm)
				arc [start angle=15, end angle=75, radius=15mm] -- (0,0); 
				\node[] at (2.5cm,1.5cm) {$X_1$};  
				
				\draw[dashed, thin ,fill = black!10, yshift = -1.15cm] (0,0) -- (195:15mm)
				arc [start angle=195, end angle=255, radius=15mm] -- (0,0); 
				\node[] at (-1.5cm,-2.5cm) {$X_2$};  
				
				\draw[dashed, thin ,fill = black!10, xshift = 1.15cm, yshift= -1.15cm] 
				(0,0) -- (285:15mm)
				arc [start angle=285, end angle=345, radius=15mm] -- (0,0); 
				\node[] at (2.5cm,-2.5cm) {$Y_2$};  
				
			\end{scope}
		\end{tikzpicture}
		\caption{Graph $H$}
		\label{fig:b}
	\end{subfigure}
	\caption{Pair of graphs which are $(1,1)$-\WL indistinguishable but $2$-\WL distinguishable from \cref{ex:weaker}. Shaded sector denote bicliques $K_{1,6}$.} \label{fig:cntex}
\end{figure}
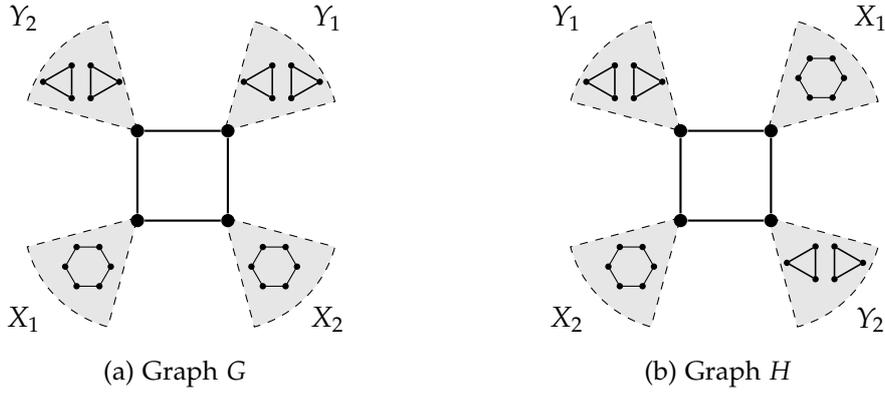 
\begin{proof}
	Let $X$ denote the six-cycle graph while $Y$ denotes the disjoint union of two triangles. 
	Define the graph $G$ as the disjoint union of a \emph{backbone} cycle $C$ on four vertices $u_1,u_2,u_3,u_4$,  two copies $X_1,X_2$ of $X$, and two copies $Y_1,Y_2$ of $Y$. Add an add an edge to $G$ from every vertex in $X_1$ to $u_1$, from every vertex in $X_2$ to $u_2$, from every vertex in $Y_1$ to $u_3$, and from every vertex in $Y_2$ to $u_4$. We say that vertex $u_1$ \emph{points} to the copy $X_1$ and so on. 
	
	For the graph $H$, the construction is identical to $G$ except for the cyclic arrangement of the copies of $X$ and~$Y$. Formally, edges are added to $H$ from every vertex in $X_1$ to $u_1$, from every vertex in $X_2$ to $u_3$, from every vertex in $Y_1$ to $u_2$, and from every vertex in $Y_2$ to $u_4$. 
	
	It is claimed that $2$-\WL can distinguish the graphs $G$ and $H$. Indeed, in the $3$-pebble bijective game on $G$ and $H$, cf.\@ \cite{cai_optimal_1992}, Spoiler can dictate a play on the backbone cycles of $G$ and $H$ and pebble vertices $u \in V(G)$ and $v \in V(H)$ such that $u$ points to an $X$-copy and $v$ points to a $Y$-copy. Then, Spoiler can use the two remaining pebbles to mark an edge of a triangle in that $Y$-copy. For any response of Duplicator, Spoiler can then pebble the remaining vertex of the triangle to expose the non-isomorphism of $X$ and $Y$, and hence of $G$ and $H$. 
	
	Finally, it is claimed that $(1,1)$-\WL cannot distinguish graphs $G$ and $H$. 
	Let $u \in V(G)$ and $v \in V(H)$. Let $G_u$ and $H_v$ denote graphs 
	$G$ and $H$ where $u$ and $v$ respectively are individualised, i.e.\@
	they are coloured with a separate common colour. 
	In each of the following cases, $G_u$ and $H_v$ are easily seen to be 
	$1$-\WL indistinguishable:
	\begin{enumerate*}[label=(\alph*)]
		\item Both $u$ and $v$ are backbone vertices pointing to $X$-copies,
		\item Both $u$ and $v$ are backbone vertices pointing to $Y$-copies, 
		\item Both $u$ and $v$ are non-backbone vertices lying in $X$-copies, 
		\item Both $u$ and $v$ are non-backbone vertices lying in $Y$-copies. 
	\end{enumerate*}
	Hence, it is easy to find the desired bijection between $V(G)$ and $V(H)$ 
	witnessing the $(1,1)$-\WL indistinguishability of $G$ and $H$.
\end{proof}

\label{ssec:adjacency-algebra}
\Cref{ex:weaker} can be paralleled by a separation of the algebraic object typically associated with $2$-\WL, i.e.\@ the adjacency algebra \cite{chen_lectures_2018}, and the algebra of equitable matrix maps. For a graph~$G$,
the \emph{adjacency algebra}  $\Adj(G)$ is defined as the subalgebra of $\mathbb{C}^{V(G) \times V(G)}$ on the $\mathbb{C}$-span of the colour class indicator matrices of the $2$-\WL colours of $G$.
For a graph $G$, write $\mathfrak{E}(G)$ for the set of matrices $\phi(G)$ for $\phi \in \mathfrak{E}$.
Since $2$-\WL colours determine $(1,1)$-\WL colours, cf.\@ \cref{lem:2wl-11wl}, \cref{thm:colentry} implies that $\mathfrak{E}(G)$ is a subalgebra of the adjacency algebra of $G$.
Analysing the graphs constructed in \cref{ex:weaker} more thoroughly shows that both graphs exhibit twenty $(1,1)$-\WL colours and thus $\dim \mathfrak{E}(F) \leq 20$ for $F \in \{G, H\}$. At the same time, the $\Adj(G)$ is $37$-dimensional, while $\Adj(H)$ is $29$-dimensional. This shows that there are graphs $G$ for which $\mathfrak{E}(G)$ is a proper subalgebra of~$\Adj(G)$.

\begin{lemma} \label{lem:2wl-11wl}
	Let $G$ and $H$ be graphs with vertices $v,w\in V(G)$ and $x, y \in V(H)$. Then
	\[
		\WL_2^\infty(G, vw) = \WL_2^\infty(H, xy) \implies \mathcal{C}_G^v(w) = \mathcal{C}_H^x(y).
	\]
\end{lemma}
\begin{proof}
	The proof is by induction on the number of rounds of the computation of the colourings $\mathcal{C}_G^v$ and $\mathcal{C}_H^x$. Observe that by assumption there exists a bijection $\pi \colon V(G) \to V(H)$ such that
	\begin{align*}
		\atp(G, vwz) = \atp(H, xy \pi(z)), \quad 
		\WL_2^\infty(G, vz) = \WL_2^\infty(H, x\pi(z)), \quad
		\WL_2^\infty(G, zw) = \WL_2^\infty(H, \pi(z)y)
	\end{align*}
	for all $z \in V(G)$.
	In the zero-th round, $(\mathcal{C}_G^v)^0(w) = (\mathcal{C}_G^x)^0(y)$ if and only if $v=w \Leftrightarrow x = y$. This is clearly implied by $\WL_2^\infty(G, vw) = \WL_2^\infty(H, xy)$.
	In later rounds,
	\[
		(\mathcal{C}_G^v)^{i+1}(w) = \multiset{\left( \atp(G_v, wz), (\mathcal{C}_G^v)^{i}(z) \right)\mid z \in V(G)   }.
	\]
	Observing that $\atp(G, vwz) = \atp(H, xy \pi(z))$ implies that $\atp(G_v, wz) = \atp(H_x, y\pi(z))$, it can be readily seen that $\pi$ from above witnesses that $(\mathcal{C}_G^v)^{i+1}(w) = (\mathcal{C}_H^x)^{i+1}(y)$.
\end{proof}

\subsection{Hitting Times}
\label{ssec:app-comdist}

Random walks on graphs constitute one of the best-studied random processes, with myriad applications to statistics and computing \cite{lovasz_random_1996}. Given a graph $G$, at any current position $v \in V(G)$ the walk proceeds to a neighbour $w$ of $v$ with probability $1/d_G(v)$. Along with cover time and mixing time, the \emph{commute time} is a central parameter in the quantitative study of random walks. Recall that the \emph{hitting time} $H(s, t)$ of two vertices $s, t \in V(G)$ is the expected duration of a random walk starting in $s$ and ending in $t$. The \emph{commute time} or \emph{commute distance} $\kappa(s, t)$ is defined to be $H(s, t) + H(t, s)$. Godsil~\cite{godsil_equiarboreal_1981} showed that two $2$-\WL indistinguishable graphs have the same multiset of commute distances. We devote the rest of this subsection to the proof of \cref{thm:comdist} which parallels this result.

First, we show in \cref{prop:induced} that the hitting times $H$ can be read off an equitable matrix map. 
With some technical adjustments for disconnected graphs in \cref{lemma:connected}, we then use \cref{thm:colentry} to conclude that hitting times are determined by $(1,1)$-\WL colours. 

\begin{proposition} \label{prop:induced}
	Define\footnote{If $\phi(G)$ is singular for a graph $G$ and $\phi \in \mathfrak{E}$ then $(\phi^{-1})(G)$ is set to the zero matrix. It follows from \cref{thm:specinv} that $\phi^{-1} \in \mathfrak{E}$.}  $X \coloneqq (I - D^{-1}A + JD/(2m))^{-1} (J - 2mD^{-1})$ where 
	$I$ is the identity, $D$ the degree, $A$ the adjacency, and $J$ the all-ones matrix map, and $m \colon G \mapsto \abs{E(G)}$ denotes the number-of-edges map.
	Then $K$ is an equitable matrix map. 
	
	Moreover, for all connected graphs $G$ with vertices $s, t \in V(G)$, $H_G(s, t) = X_G(s,t) - X_G(t,t)$.
\end{proposition}
\begin{proof}The proof is guided by \cite[Section~3]{lovasz_random_1996}.
	To ease notation, a connected graph $G$ with $n$~vertices and $m$~edges is fixed.
	Furthermore, set $M \coloneqq D^{-1} A$. Note that $M \boldsymbol{1} = \boldsymbol{1}$ and $M^T \pi = \pi$ where $\pi \coloneqq D\boldsymbol{1}/(2m)$ is the stable distribution of the random walk.
	Let finally $M' \coloneqq \boldsymbol{1} \pi^T = JD/(2m)$.
	The hitting time satisfies the following equation for all vertices $s \neq t$,
	\[
	H(s,t) = 1 + \frac{1}{d_G(s)} \sum_{r \in N_G(s)} H(r, t).
	\]
	Put in matrix form, this implies that $F \coloneqq J + MH - H$ is a diagonal matrix.
	Moreover,
	\[
	F^T \pi = J\pi + H^T(M-I)^T \pi = J\pi = \boldsymbol{1}.
	\]
	This implies that $F = 2mD^{-1}$. Hence, $H$ is determined by the equation
	\begin{equation} \label{eq:H}
		J - 2mD^{-1} = (I - M) H.
	\end{equation}
	It remains to solve this equation for $H$. However, $I-M$ is singular (as $\boldsymbol{1}$ lies in its kernel). By the Perron--Frobenius Theorem (e.g.\@ \cite{godsil_algebraic_2004}), since the graph is assumed to be connected, $\boldsymbol{1}$ spans the eigenspace of $M$ with eigenvalue $1$. This guarantees that the kernel of $I-M$ is spanned by~$\boldsymbol{1}$. Thus, for any solution $X$ of \cref{eq:H} the matrix $X + \boldsymbol{1} a^T$ for any vector $a$ is a solution as well. The matrix $H$ is nevertheless uniquely determined by \cref{eq:H} since it satisfies for every vertex $s$ that $H(s,s) = 0$. In other words, after finding a solution $X$ for \cref{eq:H}, $a$ has to be chosen such that $X + \boldsymbol{1}a^T$ is zero on the diagonal.
	
	Consider the matrix map $X \coloneqq (I - M + M')^{-1}(J-2mD^{-1})$. It is well-defined since when $N \coloneqq D^{-1/2}A D^{-1/2}$ is spectrally decomposed as $\sum_{k=1}^n \lambda_k v_kv_k^T$ with $v_1(i) = \sqrt{\pi(i)}$ and $1 = \lambda_1 > \lambda_2 \geq \dots \geq \lambda_n \geq -1$ by Perron--Frobenius,
	\[
	M' + I - M = D^{-1/2}v_1v_1^TD^{1/2} + \sum_{k=1}^n (1- \lambda_k)D^{-1/2}v_kv_k^TD^{1/2}
	\]
	observing that $\boldsymbol{1}\pi^T = D^{-1/2}v_1v_1^TD^{1/2}$. The eigenvalues $2 - \lambda_1 = 1$ and $1 - \lambda_k$ for $1 < k \leq n$ are all non-zero. Hence, $I - M + M'$ is indeed non-singular.
	Furthermore, $X$ solves \cref{eq:H}. Indeed, it can be seen that $(I-M)(I-M+M')^{-1} = I - M'$ and $(I-M')(J - 2mD^{-1}) = J-2mD^{-1}$.
	
	In order to obtain $H$ from $X$, a suitable vector $a$ has to be found such that the matrix $X + \boldsymbol{1}a^T$ is zero on the diagonal and thus equals $H$. Alternatively, $H$ can be expressed in terms of $X$ as follows:
	\[
		H(s, t) = X(s, t) - X(t,t).
	\]
	This implies the second assertions. By \cref{lem:closure}, $X$ is an equitable matrix map since it is build of the equitable matrix maps $I$, $J$, $D$, $M$, and $M'$ by matrix products and taking inverses.
\end{proof}

The proof of \cref{prop:induced} relied on the Perron--Frobenius Theorem, which characterises the adjacency spectrum of connected graphs. In order to prove \cref{prop:main} in full generality, it has to be argued that  connected components of a graph can be considered separately. This is the purpose of \cref{lemma:connected} whose proof relies on the following immediate corollary of \cref{lem:closure,thm:colentry}.

\begin{corollary} \label{fact:dist}
	For graphs $G$ and $H$, vertices $s,t \in V(G)$, and $u,v \in V(H)$, if $\mathcal{C}_G^s(t) = \mathcal{C}_G^u(v)$ then 
	$\dist_G(s,t) = \dist_{H}(u,v)$, where $\dist_G(s,t) \coloneqq \min\left\{i \geq 0 \ \middle|\ A^i(G, s, t) \neq 0\right\}$.
\end{corollary}

\begin{lemma} \label{lemma:connected}
	Let $G$ and $H$ be graphs and let $s,t \in V(G)$ and $u,v \in V(H)$ be vertices such that $\mathcal{C}_G^s(t) = \mathcal{C}_H^u(v)$.
	If $\dist_G(s,t) = \dist_{H}(u,v) < \infty$ then $\mathcal{C}_{G'}^s(t) = \mathcal{C}_{H'}^u(v)$ where $G'$ and $H'$ denote the connected components of $G$ and $H$ containing $s,t$, and $u,v$, respectively.
\end{lemma}

\begin{proof}
	Let $n \coloneqq \abs{V(G)} = \abs{V(H)}$.
	It is shown by induction on the number of iterations $i \geq 0$ that $(\mathcal{C}^s_{G})^{n+i}(x) = (\mathcal{C}^u_{H})^{n+i}(y) \implies (\mathcal{C}^s_{G'})^i(x) = (\mathcal{C}^u_{H'})^i(y)$ for all $x \in V(G')$ and $y \in V(H')$.
	For $i = 0$, the claim is clear. Suppose that  $(\mathcal{C}^s_{G})^{n+i+1}(x) = (\mathcal{C}^u_{H})^{n+i+1}(y)$ for some $x \in V(G')$ and $y \in V(H')$. Then there exists a bijection $\pi \colon V(G) \to V(H)$ such that $\atp(G, xz) = \atp(H, y\pi(z))$ and  $(\mathcal{C}^s_{G})^{n+i}(z) = (\mathcal{C}^u_{H})^{n+i}(\pi(z))$ for all $z \in V(G)$. 
	In $(\mathcal{C}^s_{G})^{n+i}$, vertices in $G'$ are distinguished from vertices not in $G'$ since with an $n$~round headstart the individualising colour of $s$ was propagated to all vertices in the same connected component.
	Hence, the bijection $\pi$ is such that $\pi(V(G')) = V(H')$. It follows that $\pi$ witnesses the following equality of multisets
	\[
		\multiset{\left( \atp(G', xr), (\mathcal{C}_G^s)^i(r)\right) \mid r \in V(G')}
		=
		\multiset{\left( \atp(H', yw), (\mathcal{C}_H^u)^i(w)\right) \mid w \in V(H')}.
	\]
	Thus, $(\mathcal{C}^s_{G'})^{i+1}(x) = (\mathcal{C}^u_{H'})^{i+1}(y)$.
\end{proof}

It remains to combine \cref{cor:cospectrality,lemma:connected,prop:induced,thm:colentry} to give a proof of \cref{thm:rw}:

\begin{proof}
	Let $G$ and $H$ be $(1,1)$-\WL indistinguishable graphs with vertices $s, t \in V(G)$ and $u, v \in V(H)$. Suppose that $\mathcal{C}_G^s(t) = \mathcal{C}_H^u(v)$.
	Observe that $G$ is connected if and only if $H$ is connected. This follows for example from the fact that $G$ and $H$ are cospectral with respect to their Laplacian matrices, cf.\@ \cite{van_dam_which_2003} and \cref{cor:cospectrality}. 
	By \cref{lemma:connected}, it may be assumed that both graphs are connected. Indeed, by \cref{fact:dist}, if $\dist_G(s,t) = \dist_H(u,v) = \infty$ then $H_G(s,t) = H_H(u,v) = \infty$. Otherwise, \cref{lemma:connected} permits to pass to the connected components containing $s,t$ and $u,v$ respectively.

	By \cref{prop:induced}, $H_G(s, t) = X_G(s, t) - X_G(t,t)$ for the $X \in \mathfrak{E}$ defined there. 
	By \cref{lem:closure}, also $X^T \in \mathfrak{E}$.
	The assumption $\mathcal{C}_G^s(t) = \mathcal{C}_H^u(v)$ implies that $\mathcal{C}_G^s(s) = \mathcal{C}_H^u(u)$ since the vertices $s$ and $u$ have distinct colours.
	By \cref{thm:colentry}, $X_G(s,s) = X_H(u,u)$ and $X^T_G(s,t) = X^T_H(u,v)$.
	Hence, $H_G(t, s) = H_H(v,u)$.
	
	This shows that if $\mathcal{C}_G^s(t) = \mathcal{C}_H^u(v)$ then $H_G(t, s) = H_H(v, u)$. To show that $(1,1)$-\WL indistinguishable graphs have the same multiset of hitting times, observe that by assumption there exist bijections $\pi \colon V(G) \to V(H)$ and $\pi_v \colon V(G) \to V(H)$ for every $v \in V(G)$ such that $\mathcal{C}^s_G(t) = \mathcal{C}^{\pi(s)}_H(\pi_s(t))$ for every $s, t \in V(G)$. This induces a bijection $\Pi \colon V(G) \times V(G) \to V(H) \times V(H)$ via $(s,t) \mapsto (\pi(s), \pi_s(t))$. By the previous observation, this bijection witnesses that
	\[
		\multiset{H_G(t, s) \mid s, t \in V(G)}
		=
		\multiset{H_H(v, u) \mid u, v \in V(H)},
	\]
	as desired.
\end{proof}

\section{Graph Spectra for $k$-\WL after $d$~Iterations}
\label{sec:genadjmat}

The purpose of this section is to outline the proof of \cref{thm:main-intro}. As already indicated, we view this result, in virtue of  \cref{thm:camb}, as a characterisation of homomorphism indistinguishability over graphs admitting certain $(k+1)$-pebble forest covers of depth at most $d$ in terms of matrix equations. The proof comprises four steps as described in \cref{recipe}. To commence, fix $k \geq 1$ and $d \geq 0$.

\subsection{Bilabelled Graphs and Augmented Homomorphism Representation}
\label{ssec:Akd}

The first step is to devise a class of bilabelled graphs such that the underlying unlabelled graphs are in $\mathcal{PFC}_k^d$, cf.\@ \cref{def:pfc-labelled}. 
For technical reasons, we consider only covers all whose leaves have the same depth \ref{b5}. It is later argued that this does not constitute a loss of generality.
As described in \cref{recipe}, we deviate from \cite{mancinska_quantum_2019,grohe_homomorphism_2021} by augmenting the labels with information about the role they play in a fixed pebble forest cover.

\begin{definition}\label{def:wlkd}
Let $\mathcal{WL}_k^d$ denote the set of all tuples $\widehat{\boldsymbol{F}} = (\boldsymbol{F}, p_\pin, p_\pout)$ such that
$\boldsymbol{F} = (F, \vec{u}, \vec{v})$ is a $(k+d,k+d)$-bilabelled graph, i.e.\@ $\vec{u}, \vec{v} \in V(F)^{k+d}$, 
and there exists a tree cover $\mathcal{F}$ of $F$ with a $(k+1)$-pebbling function $p \colon V(F) \to [k+1]$ such that
\begin{enumerate}[label=(B\arabic*)]
\item $\vec{u}_1 = \vec{v}_1$ is the root of $\mathcal{F}$ and the vertices $\vec{u}_{k+d}$ and $\vec{u}_{k+d}$ are leaves of $\mathcal{F}$,
\item $\vec{u}_i = \vec{v}_i$ for all $i \in [k]$,  $\vec{u}_1 \leq \vec{u}_2 \leq \dots \leq \vec{u}_{k+d}$, and $\vec{v}_1 \leq \vec{v}_2 \leq \dots \leq \vec{v}_{k+d}$ in $\mathcal{F}$,\label{b2}
\item for all $k \leq i < k+d$ it holds that $\vec{u}_i < \vec{u}_{i+1}$ and $\vec{v}_i < \vec{v}_{i+1}$ and that there is no vertex $x$ such that $\vec{u}_i < x < \vec{u}_{i+1}$ or $\vec{v}_i < x < \vec{v}_{i+1}$.\label{b3}
\item all unlabelled vertices $x \in V(F)$ are such that $\vec{u}_k \leq x$,\label{b4}
\item for all leaves $x$ of $\mathcal{F}$ it holds that the set $\{y \in V(F) \mid \vec{u}_k < y \leq x\}$ is of size $d$,\label{b5}
\item the leaves $\vec{u}_{k+d}$ and $\vec{v}_{k+d}$ have the least number of common ancestors among all pairs of leaves of $\mathcal{F}$, i.e.\@ writing\label{b6}
\begin{equation}\label{eq:gca}
	\gca(x,y) \coloneqq \abs{\{ z \in V(F) \mid \vec{u}_k < z \land z \leq x \land z \leq y\}}
\end{equation}
for $x, y \in V(F)$, it holds that $\gca(x, y) \geq \gca(\vec{u}_{k+d}, \vec{v}_{k+d})$ for all pairs of leaves $x,y \in V(F)$ of $\mathcal{F}$,
\item $p$ is injective on $\{\vec{u}_1, \dots, \vec{u}_k\}$,
\item $p_\pin, p_\pout \colon [k+d] \to [k+1]$ are functions such that $p_\pin(i) = p(\vec{u}_i)$ and $p_\pout(i) = p(\vec{v}_i)$ for all $i \in [k+d]$.
\end{enumerate}
$\mathcal{WL}_k^d$ furthermore contains the symbol $\bot$.
\end{definition}

Recording not only the graph and its in- and out-labels but also the value of the pebbling function at the labels allows us to introduce well-defined operations in the next section. We first augment the homomorphism representation used in \cite{mancinska_quantum_2019,grohe_homomorphism_2021} by information about the pebbling functions.

\begin{definition}\label{def:ahr}
Let $G$ be a graph. The \emph{augmented homomorphism representation} is the map
\[
\mathcal{WL}_k^d \to \mathbb{C}^{V(G)^{k+d} \times V(G)^{k+d}} \otimes \mathbb{C}^{[k+1]^{[k+d]} \times [k+1]^{[k+d]}}
\]
sending $\widehat{\boldsymbol{F}} = (F, \vec{u}, \vec{v}, p_\pin, p_\pout) \in \mathcal{WL}_k^d$
to $\boldsymbol{F}_G \otimes e_{p_\pin} e_{p_\pout}^T$ and $\bot$ to the zero matrix. Here, $\otimes$ denotes matrix Kronecker product and 
$e_p \in \mathbb{C}^{[k+1]^{[k+d]}}$ for $p \colon [k+d] \to [k+1]$ the $p$-th standard basis vector.
\end{definition}

\subsection{Operations}
\label{sec:wlkd-step2}

The second step is to define combinatorial operations on $\mathcal{WL}_k^d$ and accompanying algebraic operations respecting the augmented homomorphism representation.
Due to the correspondence between these operations, which is established in \cref{prop:ops}, we abusively use the same notation for both.
The following combinatorial operations are adaptions of standard operations on bilabelled graphs, cf.\@ \cref{sec:prelims}.

\begin{definition}[Combinatorial operations] \label{def:comop}
	Let $\widehat{\boldsymbol{F}} = (\boldsymbol{F}, p_\pin, p_\pout), \widehat{\boldsymbol{F}}' = (\boldsymbol{F}', p'_\pin, p'_\pout) \in \mathcal{WL}_k^d$.
	\begin{enumerate}
		\item The \emph{unlabelling} of $\widehat{\boldsymbol{F}} \neq \bot$ denoted by $\soe \widehat{\boldsymbol{F}}$ is the underlying unlabelled graph of $\boldsymbol{F}$.
		\item The \emph{reverse} of $\widehat{\boldsymbol{F}} = (\boldsymbol{F}, p_\pin, p_\pout)$ is $\widehat{\boldsymbol{F}}^* = (\boldsymbol{F}^*, p_\pout, p_\pin)$. The reverse of $\bot$ is $\bot$.
		\item The \emph{series composition} of $\widehat{\boldsymbol{F}}$ and $\widehat{\boldsymbol{F}}' $ is defined to be $\bot$ if $p_\pout \neq p'_\pin$ and otherwise to be $(\boldsymbol{F} \cdot \boldsymbol{F}', p_{\pin}, p'_\pout)$. The series composition of $\bot$ with any element of $\mathcal{WL}_k^d$ on either side is $\bot$.
	\end{enumerate}
\end{definition}

Unlabelling establishes a connection to unlabelled graphs. Reversal is needed for algebraic purposes, cf.\@ \cref{thm:soe}, while series composition is the operation under which finite generation will be proven in \cref{prop:fingen}. Recall the definition of $\mathcal{PFC}_k^d$ from \cref{def:pfc-labelled}.

\begin{lemma} \label{lem:wlkd-closure}
$\mathcal{WL}_k^d$ is closed under reversal and series composition. The graphs obtained from elements of $\mathcal{WL}_k^d \setminus \{\bot\}$ by unlabelling are in $\mathcal{PFC}_k^d$.
\end{lemma}

\begin{proof}The closure of $\mathcal{WL}_k^d$ under reversal and the assertion that graphs obtained from elements of $\mathcal{WL}_k^d$ by unlabelling are in $\mathcal{PFC}_k^d$ follow immediately from \cref{def:wlkd}. 
	
	It remains to show that $\mathcal{WL}_k^d$ is closed under series composition. To that end, let  $\widehat{\boldsymbol{F}} = (\boldsymbol{F}, p_\pin, p_\pout)$ and $\widehat{\boldsymbol{F}}' = (\boldsymbol{F}', p'_\pin, p'_\pout)$ be elements of $\mathcal{WL}_k^d$, wlog different from $\bot$. If $p_\pout \neq p'_\pin$ then the result of the series composition is $\bot \in \mathcal{WL}_k^d$.
	
	Otherwise, a tree cover of the graph underlying $\boldsymbol{F} \cdot \boldsymbol{F}'$ can be defined by taking the disjoint union of the tree covers of $\boldsymbol{F} = (F, \vec{u}, \vec{v})$ and $\boldsymbol{F}' = (F', \vec{u}', \vec{v}')$ satisfying \cref{def:wlkd} and identifying $\vec{v}_i$ with $\vec{u}'_i$ for all $i \in [k+d]$. The new pebbling function is defined as the union of the original pebbling functions. The assumption that  $p_\pout = p'_\pin$ guarantees that this is feasible. The axioms of \cref{def:wlkd} are easily checked. \ref{b6} is less immediate:
	
	First observe that $\gca(\vec{u}_{k+d}, \vec{v}_{k+d})$ and $\gca(\vec{u}'_{k+d}, \vec{v}'_{k+d})$ are at least $\gca(\vec{u}_{k+d}, \vec{v}'_{k+d})$. Indeed, any common ancestor of $\vec{u}_{k+d}$ and $\vec{v}'_{k+d}$ must be an ancestor of $\vec{v}_{k+d}$, which is identified with $\vec{u}'_{k+d}$. Let $z$ denote the maximal vertex such that $z \leq \vec{u}_{k+d}, \vec{v}'_{k+d}$. Any leaf $x$ of $\boldsymbol{F}$ satisfies $z \leq x$. Indeed, if $z$ and $x$ were incomparable w.r.t.\@ $\leq$ then $\gca(x, \vec{u}_{k+d}) < \gca(\vec{u}_{k+d}, \vec{v}'_{k+d})$ which contradicts the previous observation since $\gca(x, \vec{u}_{k+d}) \geq \gca(\vec{u}_{k+d}, \vec{v}_{k+d})$. The same argument applies to leaves of $\boldsymbol{F}'$. It follows that $\gca(x, y) \geq \gca(\vec{u}_{k+d}, \vec{v}'_{k+d})$ for every pair of leaves $x,y$ in the series product.
\end{proof}

\Cref{prop:ops} shows that the augmented homomorphism representation is indeed a representation of the algebraic object $\mathcal{WL}_k^d$ with reversal and series composition in terms of matrices.
\begin{lemma} \label{prop:ops}
Let $G$ be a graph. Let $\widehat{\boldsymbol{F}} = (\boldsymbol{F}, p_\pin, p_\pout), \widehat{\boldsymbol{F}}' = (\boldsymbol{F}', p'_\pin, p'_\pout) \in \mathcal{WL}_k^d$. 
\begin{enumerate}
\item Unlabelling corresponds to sum-of-entries, i.e.\@ $\hom(\soe \widehat{\boldsymbol{F}}, G ) = \soe \widehat{\boldsymbol{F}}_G$.
\item Reversal corresponds to taking adjoints, i.e.\@ $(\widehat{\boldsymbol{F}}^*)_G = (\widehat{\boldsymbol{F}}_G)^*$.
\item Series composition corresponds to matrix product, i.e.\@ $(\widehat{\boldsymbol{F}} \cdot \widehat{\boldsymbol{F}}' )_G = \widehat{\boldsymbol{F}}_G \cdot \widehat{\boldsymbol{F}}'_G$.
\end{enumerate}
\end{lemma}

\begin{proof}
	First consider the correspondence of unlabelling and sum-of-entries. Recalling the analogous argument for the homomorphism representation from \cite{mancinska_quantum_2019,grohe_homomorphism_2021},
	\begin{align*}
		\hom(\soe \widehat{\boldsymbol{F}}, G)
		&= \hom(\soe \boldsymbol{F}, G) \\
		&= \sum_{\vec{v}, \vec{v}' \in V(G)^{k+d}} \boldsymbol{F}_G(\vec{v}, \vec{v}') \\
		&= \sum_{\substack{\vec{v}, \vec{v}' \in V(G)^d, \\ p, p' \colon [k+d] \to [k+1]}} \boldsymbol{F}_G(\vec{v}, \vec{v}') e_p^T e_{p_\pin} e_{p_\pout}^T e_{p'} \\
		&= \soe \widehat{\boldsymbol{F}}_G
	\end{align*}
	since $e_p^T e_{p_\pin} e_{p_\pout}^T e_{p'} = 1$ iff $p = p_\pin$ and $p' = p_\pout$ and zero otherwise for all $p, p' \colon [k+d] \to [k+1]$.
	
	The correspondence between reversal and transposition is purely syntactical.
	
	Finally, consider series composition and matrix product. If one of the factors is $\bot$, which is mapped to the zero matrix under the augmented homomorphism representation, the statement is readily verified. Otherwise, let $\vec{u}, \vec{u}' \in V(G)^{k+d}$ and $p, p' \colon [k+d] \to [k+1]$ be arbitrary.
	\begin{align*}
		(\widehat{\boldsymbol{F}}_G \cdot \widehat{\boldsymbol{F}}'_G)(\vec{u},p; \vec{u},p')
		&= \sum_{\substack{\vec{v} \in V(G)^{k+d}, \\ q \colon [k+d] \to [k+1]}} \widehat{\boldsymbol{F}}_G(\vec{u}, p; \vec{v}, q) \widehat{\boldsymbol{F}}'_G(\vec{v}, q; \vec{u}', p') \\
		&= \sum_{\substack{\vec{v} \in V(G)^{k+d}, \\ q \colon [k+d] \to [k+1]}} \boldsymbol{F}_G(\vec{u}, \vec{v}) \boldsymbol{F}'_G(\vec{v}, \vec{u}') \delta_{p = p_\pin} \delta_{p_\pout  = q = p'_\pin} \delta_{p'_\pout = p'} \\
		&= \sum_{\vec{v} \in V(G)^{k+d}} \boldsymbol{F}_G(\vec{u}, \vec{v}) \boldsymbol{F}'_G(\vec{v}, \vec{u}') \delta_{p = p_\pin} \delta_{p_\pout = p'_\pin} \delta_{p'_\pout = p'} \\
		&= \begin{cases}
			(\boldsymbol{F} \cdot \boldsymbol{F}')_G(\vec{u},\vec{u}') \delta_{p = p_\pin}\delta_{p'_\pout = p'}, & \text{if } p_\pout = p'_\pin,\\
			0, & \text{otherwise},
		\end{cases}\\
		&= (\widehat{\boldsymbol{F}} \cdot \widehat{\boldsymbol{F}}')_G(\vec{u},p; \vec{u},p').
	\end{align*}
	The last equality holds because of the correspondence of series composition and matrix product in the non-augmented case \cite{mancinska_quantum_2019,grohe_homomorphism_2021}.
\end{proof}

\subsection{Finite Generation}
\label{ssec:fingen}
To prove that $\mathcal{WL}_k^d$ is finitely generated under series composition, we define the set $\mathcal{A}_k^d$, cf.\@ \cref{fig:intro-tdb}. Their augmented homomorphism representations yield the matrix maps featured in \cref{thm:main-intro}. For $n \in \mathbb{N}$, write $(n)$ for the tuple $(1,2,\dots, n)$.

\begin{definition}\label{defn:gadj}
	The set $\mathcal{A}_k^d$ is the subset of $\mathcal{WL}_k^d$ containing the following elements:
	For every $H \subseteq [k]$, every $\vec{h} \in H^k$ such that $\vec{h}_1 \leq \vec{h}_2 \leq \dots \leq \vec{h}_k$ and every $p_H \colon [k] \to [k+1]$ such that $p(i) = p(j)$ implies $\vec{h}_i = \vec{h}_j$ for all $i,j\in [k]$,
	\begin{enumerate}
		\item the \emph{identity graph} $(\boldsymbol{I}, p, p)$, where $\boldsymbol{I} = (I, \vec{h}(k+1)\dots(k+d), \vec{h}(k+1)\dots(k+d))$ with $V(I) = H \cup \{k+1, \dots, k+d\}$ and $E(I) = \emptyset$, for every $p \colon [k+d] \to [k+1]$ such that $p|_{[k]} = p_H$.
		\item the \emph{adjacency graph} $(\boldsymbol{A}^{ij}, p, p)$ for $i < j \in H \cup \{k+1, \dots, k+d\}$, where $\boldsymbol{A}^{ij} = (A^{ij}, \vec{h}(k+1)\dots(k+d), \vec{h}(k+1)\dots(k+d))$ with $V(A^{ij}) = H \cup \{k+1, \dots, k+d\}$ and $E(A^{ij}) = \{ij\}$, for every $p \colon [k+d] \to [k+1]$ such that $p(i) \neq p(\ell)$ for all $i < \ell \leq j$ and $p|_{[k]} = p_H$.
		\item the \emph{join graph} $(\boldsymbol{J}^\ell, p_\pin, p_\pout)$ for $k \leq \ell < k+d$, where  $\boldsymbol{J}^\ell = (J^\ell, \vec{h}(k+1)\dots(k+d), \vec{h}(k+1)\dots\ell(\ell+1)'\dots (k+d)')$ with  $V(J^\ell) = H \cup \{k+1, \dots, k+d, (\ell+1)', \dots, (k+d)'\}$ and $E(J^\ell) = \emptyset$, for every $p_\pin, p_\pout \colon [k+d] \to [k+1]$ such that $p_\pin|_{[\ell]} = p_\pout|_{[\ell]}$ and $p|_{[k]} = p_H$.
	\end{enumerate}
\end{definition}
\begin{figure}
	\centering
	\captionsetup[subfigure]{justification=centering}
	\tikzset{
		vertex/.style = {fill,circle,inner sep=0pt,minimum size=5pt},
		edge/.style = {-,thick},
		lbl/.style={color=lightgray}
	}
	\begin{subfigure}[t]{0.3 \textwidth}
		\centering
		\begin{tikzpicture}
			\node[vertex, label = {left:$(k+1)\vphantom{(k+1)'}$}, label = {right:$(k+1)'$}] (v1) {} ;
			\node[vertex, label = {left:$(k+2)\vphantom{(k+2)'}$}, label = {right:$(k+2)'$}] (v2) [below of = v1] {} ;
			\node[vertex, label = {left:$k+d\vphantom{(k+d)'}$}, label = {right:$(k+d)'$}] (vd) [below of = v2, yshift=-1cm]{} ;
			
			\draw [fill=gray, draw=gray, -stealth] (v1) -- (v2);
			\draw [fill=gray, draw=gray, -stealth] (v2) -- (vd) node [midway, xshift=-.2cm, anchor=center, rotate=270] {$\dots$} node [midway, xshift=.2cm, anchor=center, rotate=270] {$\dots$};
			
			\draw [fill=gray, draw=black, rounded corners=2pt] (-.1, 2) rectangle (.1, .5);	
			\draw [fill=gray, draw=gray, -stealth] (0, .5) -- (v1);
			\node [anchor=north east] at (-.1, 2.1) {$1\vphantom{1'}$};
			\node [anchor=north west] at (.1, 2.1) {$1'$};
			\node [anchor=south east] at (-.1, .5) {$k\vphantom{k'}$};
			\node [anchor=south west] at (.1, .5) {$k'$};
			\node [rotate=270] at (-.3, 1.25) {$\dots$};
			\node [rotate=270] at (.3, 1.25) {$\dots$};
		\end{tikzpicture}
		\caption{$\boldsymbol{I}$}
	\end{subfigure}
	\begin{subfigure}[t]{0.3\textwidth}
		\centering
		\begin{tikzpicture}
			\node[vertex, label = {left:$(k+1)\vphantom{(k+1)'}$}, label = {right:$(k+1)'$}] (v1) {} ;
			\node[vertex, label = {left:$i\vphantom{i'}$}, label = {right:$i'$}] (vi) [below of = v1] {} ;
			\node[vertex, label = {left:$j\vphantom{j'}$}, label = {right:$j'$}] (vj) [below of = vi] {} ;
			\node[vertex, label = {left:$k+d\vphantom{(k+d)'}$}, label = {right:$(k+d)'$}] (vd) [below of = vj]{} ;
			\draw[edge] (vi) to [bend left = 60] (vj) ;
			
			\draw [fill=gray, draw=gray, -stealth] (v1) -- (vi) node [midway, xshift=-.2cm, anchor=center, rotate=270] {$\dots$} node [midway, xshift=.2cm, anchor=center, rotate=270] {$\dots$};
			\draw [fill=gray, draw=gray, -stealth] (vi) -- (vj) node [midway, xshift=-.2cm, anchor=center, rotate=270] {$\dots$} node [midway, xshift=.2cm, anchor=center, rotate=270] {$\dots$};
			\draw [fill=gray, draw=gray, -stealth] (vj) -- (vd) node [midway, xshift=-.2cm, anchor=center, rotate=270] {$\dots$} node [midway, xshift=.2cm, anchor=center, rotate=270] {$\dots$};
			
			\draw [fill=gray, draw=black, rounded corners=2pt] (-.1, 2) rectangle (.1, .5);	
			\draw [fill=gray, draw=gray, -stealth] (0, .5) -- (v1);
			\node [anchor=north east] at (-.1, 2.1) {$1\vphantom{1'}$};
			\node [anchor=north west] at (.1, 2.1) {$1'$};
			\node [anchor=south east] at (-.1, .5) {$k\vphantom{k'}$};
			\node [anchor=south west] at (.1, .5) {$k'$};
			
			\node [rotate=270] at (-.3, 1.25) {$\dots$};
			\node [rotate=270] at (.3, 1.25) {$\dots$};
		\end{tikzpicture}
		\caption{$\boldsymbol{A}^{ij}$ for $1 \leq i < j \leq k+d$}
	\end{subfigure}
	\begin{subfigure}[t]{0.3 \textwidth}
		\centering
		\begin{tikzpicture}
			\node[vertex, label = {left:$(k+1)\vphantom{(k+1)'}$}, label = {right:$(k+1)'$}] (v1) {} ;
			\node[vertex, label = {left:$\ell\vphantom{\ell'}$}, label = {right:$\ell'$}] (vell) [below of = v1] {} ;
			\node[] (vemp) [below of = vell] {};
			\node[vertex, label = {left:$(\ell+1)$}] (vleft) [left of = vemp]{};
			\node[vertex, label = {right:$(\ell+1)'$}] (vright) [right of = vemp]{};
			\node[] (vemp1) [below of = vemp] {}; 
			\node[vertex, label = {left:$k+d$}] (vdl) [below of = vleft]{};
			\node[vertex, label = {right:$(k+d)'$}] (vdr) [below of = vright]{};
			
			\draw [fill=gray, draw=gray, -stealth] (v1) -- (vell) node [midway, xshift=-.2cm, anchor=center, rotate=270] {$\dots$} node [midway, xshift=.2cm, anchor=center, rotate=270] {$\dots$};
			\draw [fill=gray, draw=gray, -stealth] (vell) -- (vleft);
			\draw [fill=gray, draw=gray, -stealth] (vell) -- (vright);
			\draw [fill=gray, draw=gray, -stealth] (vleft) -- (vdl) node [midway, xshift=-.2cm, anchor=center, rotate=270] {$\dots$} node [midway, xshift=.2cm, anchor=center, rotate=270] {$\dots$};
			\draw [fill=gray, draw=gray, -stealth] (vright) -- (vdr) node [midway, xshift=-.2cm, anchor=center, rotate=270] {$\dots$} node [midway, xshift=.2cm, anchor=center, rotate=270] {$\dots$};
			
			\draw [fill=gray, draw=black, rounded corners=2pt] (-.1, 2) rectangle (.1, .5);	
			\draw [fill=gray, draw=gray, -stealth] (0, .5) -- (v1);
			\node [anchor=north east] at (-.1, 2.1) {$1\vphantom{1'}$};
			\node [anchor=north west] at (.1, 2.1) {$1'$};
			\node [anchor=south east] at (-.1, .5) {$k\vphantom{k'}$};
			\node [anchor=south west] at (.1, .5) {$k'$};
			
			\node [rotate=270] at (-.3, 1.25) {$\dots$};
			\node [rotate=270] at (.3, 1.25) {$\dots$};
		\end{tikzpicture}
		\caption{$\boldsymbol{J}^\ell$ for $k \leq \ell < k+d$}
	\end{subfigure}
	\caption{\label{fig:intro-tdb}The $(k+d,k+d)$-bilabelled graphs introduced in \cref{defn:gadj}. The numbers on the left and right of the vertices indicate in- and out-labels respectively. The relation $\leq$ of the tree cover is indicated by grey directed edges pointing from lesser to greater elements.}
\end{figure}
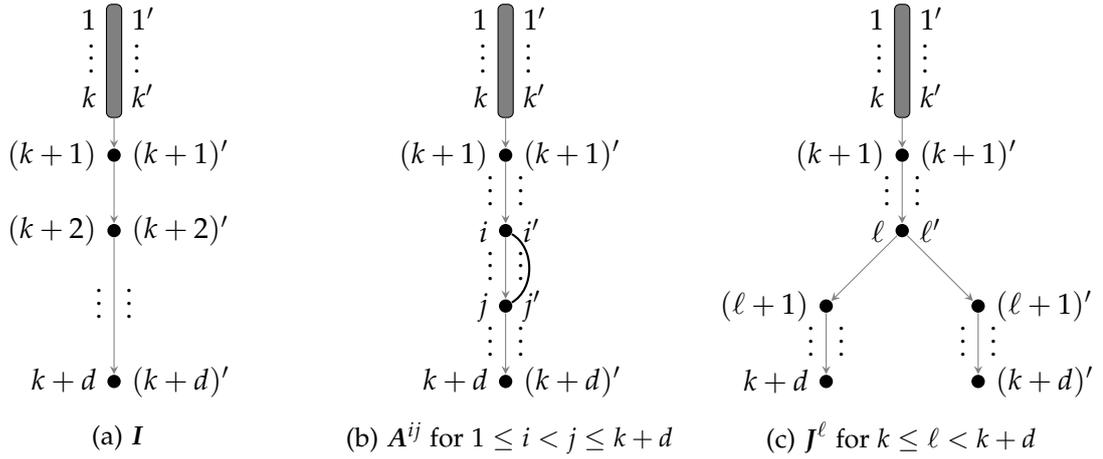

Observe that $\boldsymbol{I}$, $\boldsymbol{A}^{ij}$, and $\boldsymbol{J}^{ij}$ admit tree covers compatible with all stipulated pebbling functions, cf.\@ \cref{fig:intro-tdb}. 
Intuitively, the vertices in $H$ represent the free variables of the formulas in \cref{thm:camb}.
The set $\mathcal{A}_k^d$ comprises $O(k^k \cdot (k+d)^2 \cdot k^{2(k+d)} )$ many matrix maps of dimension $O((nk)^{2(d+k)})$ for $n$-vertex graphs. As argued in \cref{sec:results}, we consider $k$ and $d$ to be small constants in application scenarios. Having defined the generators, it can be proven that $\mathcal{WL}_k^d$ is finitely generated.

\begin{proposition} \label{prop:fingen}
For every $\widehat{\boldsymbol{F}} \in \mathcal{WL}_k^d$, there exist $\widehat{\boldsymbol{A}}^1, \dots, \widehat{\boldsymbol{A}}^r \in \mathcal{A}_k^d$ for some $r \in \mathbb{N}$ such that
$\widehat{\boldsymbol{F}}$
is the series composition of $\widehat{\boldsymbol{A}}^1, \dots, \widehat{\boldsymbol{A}}^r$, i.e.\@ $ \widehat{\boldsymbol{F}} = \widehat{\boldsymbol{A}}^1 \cdot \dots \cdot \widehat{\boldsymbol{A}}^r$.
\end{proposition}

\begin{proof}Let $\widehat{\boldsymbol{F}} = (\boldsymbol{F}, p_\pin, p_\pout) \in \mathcal{WL}_k^d$. 
	Let $\mathcal{F}$ be a tree cover and $p$ be a pebbling function of the graph underlying $\boldsymbol{F} = (F, \vec{u}, \vec{v})$ satisfying \cref{def:wlkd}.
	The proof is by induction on the number of leaves in $\mathcal{F}$. 
	
	Suppose that $\mathcal{F}$ has exactly one leaf, i.e.\@ it is a chain. 
	Then $\boldsymbol{F}$ is a series composition of $\boldsymbol{I}$ and  $\abs{E(F)}$-many bilabelled graphs of the form $\boldsymbol{A}^{ij}$, one for each edge in $F$.
		
	Suppose that $\mathcal{F}$ has more than one leaf. 
	Write $X$ for the set of all leaves $x$ of $\mathcal{F}$ other than $\vec{u}_{k+d}$ such that $\gca(\vec{u}_{k+d}, x)$ is maximal, i.e.\@ such that $\gca(\vec{u}_{k+d}, x) \geq \gca(\vec{u}_{k+d}, y)$ for all leaves $y \neq \vec{u}_{k+d}$. Let $x \in X$ be such that $\gca(x, \vec{v}_{k+d})$ is minimal, i.e.\@ $\gca(x, \vec{v}_{k+d}) \leq \gca(y, \vec{v}_{k+d})$ for all $y \in X$.
	Write $D$ for the set of vertices $z \in V(F)$ such that $z \leq \vec{u}_{k+d}$ and $z$ and $x$ are incomparable. Note that $D$ forms a chain in $\mathcal{F}$.
	
	Define $F'$ as the induced subgraph of $F$ on vertices comparable with $\vec{u}_{k+d}$. Define $\mathcal{F}'$ by restriction of $\mathcal{F}$, respectively, to $V(F')$. It can be readily verified that $\widehat{\boldsymbol{F}}' = (\boldsymbol{F}', p', p')$ with $p'(i) \coloneqq p(\vec{u}_i)$ for all $i \in [k+d]$ and $\boldsymbol{F}' = (F', \vec{u}, \vec{u})$ is in $\mathcal{WL}_k^d$. Moreover, its tree cover $\mathcal{F}'$ has only a single leaf.
	
	Define $F''$ as the graph obtained from $F$ by deleting the vertices in $D$.  Define $\mathcal{F}''$ by restriction of $\mathcal{F}$, respectively, to $V(F'')$. This forest cover is indeed a tree cover, i.e.\@ connected, since $x$ was chosen to be maximal. Write $\vec{x} \in V(F'')^{k+d}$ for the tuple such that $\vec{x}_i = \vec{u}_i = \vec{v}_i$ for all $i \in [k]$ and $\vec{x}_{k+d} = x$ satisfying \ref{b2} and \ref{b3}.
	Let $p'(i) \coloneqq p(\vec{x}_i)$ for all $i \in [k+d]$. It is claimed that $(\boldsymbol{F}'', p', p_\pout)$ is in $\mathcal{WL}_k^d$. 
	The only non-trivial property is \ref{b6}, i.e.\@ that $\gca(a, b) \geq \gca(x, \vec{v}_{k+d})$ for all leaves $a,b$ in $F''$. Distinguish cases: If $\gca(\vec{u}_{k+d}, x) > \gca(\vec{u}_{k+d}, \vec{v}_{k+d})$ then there is a vertex $y$ such that $y$ is a common ancestor of $\vec{u}_{k+d}$ and $x$ but $y$ is not an ancestor of $\vec{v}_{k+d}$. This implies that every common ancestor of $\vec{v}_{k+d}$ and $x$ is comparable with $y$ and hence $\gca(x, \vec{v}_{k+d}) = \gca(\vec{u}_{k+d}, \vec{v}_{k+d})$.
	Hence, $\gca(a, b) \geq \gca(\vec{u}_{k+d}, \vec{v}_{k+d}) = \gca(x, \vec{v}_{k+d})$ for any pair of leaves $a,b$ in $F''$.
	However, if $\gca(\vec{u}_{k+d}, x) = \gca(\vec{u}_{k+d}, \vec{v}_{k+d})$ then all leaves $a$ of $F''$ are in fact in $X$ and the claim follows readily.
	
	Let $\ell$ be maximal such that $\vec{x}_\ell = \vec{u}_\ell$. Then $\widehat{\boldsymbol{F}} = \widehat{\boldsymbol{F}}' \cdot \widehat{\boldsymbol{J}}^\ell \cdot \widehat{\boldsymbol{F}}''$. Since both $F'$ and $F''$ have less leaves than $F$, the claim follows inductively.
\end{proof}

\subsection{Matrix Equations}

We complete the proof of \cref{thm:main-intro} by appealing to the sum-of-entries version of the Specht--Wiegmann theorem, stated as \cref{thm:soe}.

\begin{proof}By \cref{cor:camb}, it suffices to show that the second assertion in \cref{thm:main-intro} is equivalent to homomorphism indistinguishability of $G$ and $H$ over $\mathcal{PFC}_k^d$.
	First suppose that the latter condition holds.
	Let $\widehat{\boldsymbol{A}}^1, \dots, \widehat{\boldsymbol{A}}^m$ be an enumeration of $\mathcal{A}_k^d$.
	Let $\mathsf{A} \coloneqq (\widehat{\boldsymbol{A}}^{1}_G, \dots, \widehat{\boldsymbol{A}}^{m}_G)$ and $\mathsf{B} \coloneqq (\widehat{\boldsymbol{A}}^{1}_H, \dots, \widehat{\boldsymbol{A}}^{m}_H)$ be the matrix sequences obtained by applying the augmented homomorphism representations with respect to $G$ and $H$.
	By \cref{thm:soe}, it suffices to show that  $\soe(w_{\mathsf{A}}) = \soe(w_{\mathsf{B}})$ for every word $w \in \Gamma_{2m}$.
	
	Since $\mathcal{A}_k^d$ is closed under the reversal operation, there exist 
	$\widehat{\boldsymbol{A}}^{i_1}, \dots, \widehat{\boldsymbol{A}}^{i_r} \in \mathcal{A}_k^d$ such that $w_{\mathsf{A}} = \widehat{\boldsymbol{A}}^{i_1}_G \cdots \widehat{\boldsymbol{A}}^{i_r}_G$ and $w_{\mathsf{B}} = \widehat{\boldsymbol{A}}^{i_1}_H \cdots \widehat{\boldsymbol{A}}^{i_r}_H$. If the product evaluates to $\bot$ then 
	$\soe(w_{\mathsf{A}}) = 0 = \soe(w_{\mathsf{B}})$ by \cref{prop:ops}. 
	Otherwise, by \cref{lem:wlkd-closure,prop:ops}, the graph underlying the product is in $\mathcal{PFC}_k^d$. 
	Hence, $\soe(w_{\mathsf{A}}) = \hom(F,G) = \hom(F, H) = \soe(w_{\mathsf{B}})$, by assumption and \cref{prop:ops}.
	
	Now suppose that the second assertion of \cref{thm:main-intro} holds, i.e.\@ there exists a pseudo-stochastic matrix $X$ 
	satisfying the stipulated equations. Since $\mathcal{A}_k^d$ contains the identity graph $\widehat{\boldsymbol{I}}$,
	\[
	\abs{V(G)}^{k+d} = \soe \widehat{\boldsymbol{I}}_G = \boldsymbol{1}^T \widehat{\boldsymbol{I}}_G \boldsymbol{1}
	= \boldsymbol{1}^T X \widehat{\boldsymbol{I}}_G \boldsymbol{1}
	= \soe \widehat{\boldsymbol{I}}_H
	= \abs{V(H)}^{k+d}.
	\]
	In other words, $G$ and $H$ have the same number of vertices. Let $F$ be an arbitrary graph in $\mathcal{PFC}_k^d$. It has to be shown that $\hom(F, G) = \hom(F, H)$.
	By padding $F$ with isolated vertices, it may be supposed that $F = \soe \widehat{\boldsymbol{F}}$ for some $\widehat{\boldsymbol{F}} \in \mathcal{WL}_k^d$. Indeed, the disjoint union $F_1 \sqcup F_2$ of two graphs $F_1$, $F_2$ satisfies the identity 
	$\hom(F_1 \sqcup F_2, G) = \hom(F_1,G) \cdot \hom(F_2,G)$ and $G$ and $H$ have the same number of vertices.
	
	Applying \cref{prop:fingen} to the graph $\widehat{\boldsymbol{F}}$, we obtain a sequence $\widehat{\boldsymbol{A}}^{i_1}, \dots, \widehat{\boldsymbol{A}}^{i_r} \in \mathcal{A}_k^d$ such that $\hom(F,G) = \soe(\widehat{\boldsymbol{A}}^{i_1}_G \cdots \widehat{\boldsymbol{A}}^{i_r}_G) = \soe(\widehat{\boldsymbol{A}}^{i_1}_G \cdots \widehat{\boldsymbol{A}}^{i_r}_G) = \hom(F, H)$ by \cref{thm:soe}.
\end{proof}

\section{Conclusion} 

We mapped out the landscape of spectral graph invariants between $1$-\WL and $2$-\WL by introducing $(1,1)$-\WL indistinguishability which subsumes many such invariants. This resolves an open question of \cite{furer_power_2010} and strengthens a result of \cite{godsil_equiarboreal_1981}.
We furthermore devised a hierarchy of spectral invariants matching the hierarchy of combinatorial invariants yielded by the Weisfeiler--Leman algorithm.
The design of graph neural networks based on these higher-order spectral invariants forms an interesting direction of future work.  

Our results contribute to the study of homomorphism indistinguishability. For characterising homomorphism indistinguishability over graphs admitting certain $k$-pebble forest covers of depth $d$, we devised a representation of labelled graphs in terms of matrices encoding information about the structure of the labels such that it affords a correspondence between combinatorial and algebraic operations. We believe that this strategy may prove useful for studying other graph classes in the future.

\paragraph{Acknowledgements}
We thank Martin Grohe for helpful conversations.
The first author is funded by the DFG Research Grants Program RA~3242/1-1 via project number 411032549. The second author is supported by the German Research Council (DFG) within Research Training Group 2236 (UnRAVeL). 

\newpage
\begin{small}
	\bibliographystyle{plainurl}

\begin{thebibliography}{10}

\bibitem{abramsky_structure_2022}
Samson Abramsky.
\newblock Structure and {Power}: an {Emerging} {Landscape}.
\newblock {\em Fundamenta Informaticae}, 186(1-4):1--26, 2022.
\newblock \href {https://doi.org/10.3233/FI-222116}
  {\path{doi:10.3233/FI-222116}}.

\bibitem{abramsky_pebbling_2017}
Samson Abramsky, Anuj Dawar, and Pengming Wang.
\newblock The {Pebbling} {Comonad} in {Finite} {Model} {Theory}.
\newblock In {\em Proceedings of the 32nd {Annual} {ACM}/{IEEE} {Symposium} on
  {Logic} in {Computer} {Science}}, {LICS} '17. IEEE Press, 2017.
\newblock \href {https://doi.org/10.1109/LICS.2017.8005129}
  {\path{doi:10.1109/LICS.2017.8005129}}.

\bibitem{abramsky_discrete_2022}
Samson Abramsky, Tomáš Jakl, and Thomas Paine.
\newblock Discrete {Density} {Comonads} and {Graph} {Parameters}.
\newblock In Helle~Hvid Hansen and Fabio Zanasi, editors, {\em Coalgebraic
  {Methods} in {Computer} {Science}}, pages 23--44, Cham, 2022. Springer
  International Publishing.
\newblock \href {https://doi.org/10.1007/978-3-031-10736-8_2}
  {\path{doi:10.1007/978-3-031-10736-8_2}}.

\bibitem{abramsky_relating_2021}
Samson Abramsky and Nihil Shah.
\newblock Relating structure and power: {Comonadic} semantics for computational
  resources.
\newblock {\em Journal of Logic and Computation}, 31(6):1390--1428, September
  2021.
\newblock \href {https://doi.org/10.1093/logcom/exab048}
  {\path{doi:10.1093/logcom/exab048}}.

\bibitem{AHU}
Alfred~V. Aho, John~E. Hopcroft, and Jeffrey~D. Ullman.
\newblock {\em {The Design and Analysis of Computer Algorithms}}.
\newblock Addison-Wesley, Reading, Mass., 1974.

\bibitem{atserias_expressive_2021}
Albert Atserias, Phokion~G. Kolaitis, and Wei{-}Lin Wu.
\newblock On the expressive power of homomorphism counts.
\newblock In {\em 36th Annual {ACM/IEEE} Symposium on Logic in Computer
  Science, {LICS} 2021, Rome, Italy, June 29 - July 2, 2021}, pages 1--13.
  {IEEE}, 2021.
\newblock \href {https://doi.org/10.1109/LICS52264.2021.9470543}
  {\path{doi:10.1109/LICS52264.2021.9470543}}.

\bibitem{AtseriasM13}
Albert Atserias and Elitza~N. Maneva.
\newblock {Sherali}--{Adams} {Relaxations} and {Indistinguishability} in
  {Counting} {Logics}.
\newblock {\em {SIAM} J. Comput.}, 42(1):112--137, 2013.
\newblock \href {https://doi.org/10.1137/120867834}
  {\path{doi:10.1137/120867834}}.

\bibitem{berkholz_tight_2017}
Christoph Berkholz, Paul~S. Bonsma, and Martin Grohe.
\newblock Tight lower and upper bounds for the complexity of canonical colour
  refinement.
\newblock {\em Theory Comput. Syst.}, 60(4):581--614, 2017.
\newblock \href {https://doi.org/10.1007/s00224-016-9686-0}
  {\path{doi:10.1007/s00224-016-9686-0}}.

\bibitem{brouwer_spectra_2012}
Andries~E. Brouwer and Willem~H. Haemers.
\newblock {\em Spectra of graphs}.
\newblock Universitext. Springer, New York, NY Dordrecht Heidelberg London,
  2012.

\bibitem{cai_optimal_1992}
Jin-Yi Cai, Martin Fürer, and Neil Immerman.
\newblock An optimal lower bound on the number of variables for graph
  identification.
\newblock {\em Combinatorica}, 12(4):389--410, 1992.
\newblock \href {https://doi.org/10.1007/BF01305232}
  {\path{doi:10.1007/BF01305232}}.

\bibitem{CHGR}
Jeff Cheeger.
\newblock {\em A Lower Bound for the Smallest Eigenvalue of the Laplacian},
  pages 195--200.
\newblock Princeton University Press, 2015.
\newblock \href {https://doi.org/10.1515/9781400869312-013}
  {\path{doi:10.1515/9781400869312-013}}.

\bibitem{chen_lectures_2018}
Gang Chen and Ilia Ponomarenko.
\newblock {\em Lectures on {Coherent} {Configurations}}.
\newblock Central China Normal University Press, Wuhan, 2018.
\newblock URL: \url{https://www.pdmi.ras.ru/~inp/ccNOTES.pdf}.

\bibitem{cvetkovic_eigenspaces_1997}
Dragoš~M. Cvetković, Peter Rowlinson, and Slobodan~K. Simić.
\newblock {\em Eigenspaces of {Graphs}}.
\newblock Cambridge University Press, 1 edition, January 1997.
\newblock \href {https://doi.org/10.1017/CBO9781139086547}
  {\path{doi:10.1017/CBO9781139086547}}.

\bibitem{cvetkovic_signless_2007}
Dragoš~M. Cvetković, Peter Rowlinson, and Slobodan~K. Simić.
\newblock Signless {Laplacians} of finite graphs.
\newblock {\em Linear Algebra and its Applications}, 423(1):155--171, May 2007.
\newblock \href {https://doi.org/10.1016/j.laa.2007.01.009}
  {\path{doi:10.1016/j.laa.2007.01.009}}.

\bibitem{dawar_lovasz_2021}
Anuj Dawar, Tomáš Jakl, and Luca Reggio.
\newblock Lov{\'{a}}sz-type theorems and game comonads.
\newblock In {\em 36th Annual {ACM/IEEE} Symposium on Logic in Computer
  Science, {LICS} 2021, Rome, Italy, June 29 - July 2, 2021}, pages 1--13.
  {IEEE}, 2021.
\newblock \href {https://doi.org/10.1109/LICS52264.2021.9470609}
  {\path{doi:10.1109/LICS52264.2021.9470609}}.

\bibitem{dawar_lovasz_2021_arxiv}
Anuj Dawar, Tomáš Jakl, and Luca Reggio.
\newblock Lov{\'{a}}sz-type theorems and game comonads.
\newblock {\em CoRR}, abs/2105.03274, 2021.
\newblock URL: \url{https://arxiv.org/abs/2105.03274}, \href
  {http://arxiv.org/abs/2105.03274} {\path{arXiv:2105.03274}}.

\bibitem{dawar_generalizations_2020}
Anuj Dawar and Danny Vagnozzi.
\newblock {Generalizations of $k$-dimensional Weisfeiler–Leman
  stabilization}.
\newblock {\em Moscow Journal of Combinatorics and Number Theory}, 9(3):229 --
  252, 2020.
\newblock \href {https://doi.org/10.2140/moscow.2020.9.229}
  {\path{doi:10.2140/moscow.2020.9.229}}.

\bibitem{DellGR18}
Holger Dell, Martin Grohe, and Gaurav Rattan.
\newblock {Lov{\'{a}}sz Meets Weisfeiler and Leman}.
\newblock In Ioannis Chatzigiannakis, Christos Kaklamanis, D{\'{a}}niel Marx,
  and Donald Sannella, editors, {\em 45th International Colloquium on Automata,
  Languages, and Programming, {ICALP} 2018, July 9-13, 2018, Prague, Czech
  Republic}, volume 107 of {\em LIPIcs}, pages 40:1--40:14. Schloss Dagstuhl -
  Leibniz-Zentrum f{\"{u}}r Informatik, 2018.
\newblock \href {https://doi.org/10.4230/LIPIcs.ICALP.2018.40}
  {\path{doi:10.4230/LIPIcs.ICALP.2018.40}}.

\bibitem{Diestel}
Reinhard Diestel.
\newblock {\em Graph Theory, 4th Edition}, volume 173 of {\em Graduate texts in
  mathematics}.
\newblock Springer, 2012.

\bibitem{dvorak_recognizing_2010}
Zdeněk Dvořák.
\newblock On recognizing graphs by numbers of homomorphisms.
\newblock {\em Journal of Graph Theory}, 64(4):330--342, August 2010.
\newblock URL: \url{http://doi.wiley.com/10.1002/jgt.20461}, \href
  {https://doi.org/10.1002/jgt.20461} {\path{doi:10.1002/jgt.20461}}.

\bibitem{feldman_weisfeiler_2022}
Or~Feldman, Amit Boyarski, Shai Feldman, Dani Kogan, Avi Mendelson, and Chaim
  Baskin.
\newblock Weisfeiler and leman go infinite: Spectral and combinatorial
  pre-colorings.
\newblock {\em CoRR}, abs/2201.13410, 2022.
\newblock URL: \url{https://arxiv.org/abs/2201.13410}, \href
  {http://arxiv.org/abs/2201.13410} {\path{arXiv:2201.13410}}.

\bibitem{furer_graph_1995}
Martin Fürer.
\newblock Graph {Isomorphism} {Testing} without {Numerics} for {Graphs} of
  {Bounded} {Eigenvalue} {Multiplicity}.
\newblock In {\em Proceedings of the {Sixth} {Annual} {ACM}-{SIAM} {Symposium}
  on {Discrete} {Algorithms}}, {SODA} '95, pages 624--631, USA, 1995. Society
  for Industrial and Applied Mathematics.
\newblock \href {https://doi.org/10.5555/313651.313828}
  {\path{doi:10.5555/313651.313828}}.

\bibitem{furer_power_2010}
Martin Fürer.
\newblock On the power of combinatorial and spectral invariants.
\newblock {\em Linear Algebra and its Applications}, 432(9):2373--2380, April
  2010.
\newblock \href {https://doi.org/10.1016/j.laa.2009.07.019}
  {\path{doi:10.1016/j.laa.2009.07.019}}.

\bibitem{SGTcol}
Chris~D. Godsil and Michael~W. Newman.
\newblock Eigenvalue bounds for independent sets.
\newblock {\em J. Comb. Theory, Ser. {B}}, 98(4):721--734, 2008.
\newblock \href {https://doi.org/10.1016/j.jctb.2007.10.007}
  {\path{doi:10.1016/j.jctb.2007.10.007}}.

\bibitem{godsil_equiarboreal_1981}
Christopher~David Godsil.
\newblock Equiarboreal graphs.
\newblock {\em Combinatorica}, 1(2):163--167, June 1981.
\newblock \href {https://doi.org/10.1007/BF02579272}
  {\path{doi:10.1007/BF02579272}}.

\bibitem{godsil_algebraic_2004}
Christopher~David Godsil and Gordon Royle.
\newblock {\em Algebraic graph theory}.
\newblock Number 207 in Graduate texts in mathematics. Springer, New York, NY,
  nachdr. edition, 2004.
\newblock \href {https://doi.org/10.1007/978-1-4613-0163-9}
  {\path{doi:10.1007/978-1-4613-0163-9}}.

\bibitem{MGplnr}
Martin Grohe.
\newblock Fixed-point logics on planar graphs.
\newblock In {\em Thirteenth Annual {IEEE} Symposium on Logic in Computer
  Science, Indianapolis, Indiana, USA, June 21-24, 1998}, pages 6--15. {IEEE}
  Computer Society, 1998.
\newblock \href {https://doi.org/10.1109/LICS.1998.705639}
  {\path{doi:10.1109/LICS.1998.705639}}.

\bibitem{grohe_excluded_12}
Martin Grohe.
\newblock Fixed-point definability and polynomial time on graphs with excluded
  minors.
\newblock {\em J. ACM}, 59(5), November 2012.
\newblock \href {https://doi.org/10.1145/2371656.2371662}
  {\path{doi:10.1145/2371656.2371662}}.

\bibitem{grohe_descriptive_2017}
Martin Grohe.
\newblock {\em Descriptive Complexity, Canonisation, and Definable Graph
  Structure Theory}, volume~47 of {\em Lecture Notes in Logic}.
\newblock Cambridge University Press, 2017.
\newblock \href {https://doi.org/10.1017/9781139028868}
  {\path{doi:10.1017/9781139028868}}.

\bibitem{grohe_counting_2020}
Martin Grohe.
\newblock Counting {Bounded} {Tree} {Depth} {Homomorphisms}.
\newblock In {\em Proceedings of the 35th {Annual} {ACM}/{IEEE} {Symposium} on
  {Logic} in {Computer} {Science}}, {LICS} '20, pages 507--520, New York, NY,
  USA, 2020. Association for Computing Machinery.
\newblock \href {https://doi.org/10.1145/3373718.3394739}
  {\path{doi:10.1145/3373718.3394739}}.

\bibitem{grohe_logic_2021}
Martin Grohe.
\newblock The {Logic} of {Graph} {Neural} {Networks}.
\newblock In {\em 36th {Annual} {ACM}/{IEEE} {Symposium} on {Logic} in
  {Computer} {Science}, {LICS} 2021, {Rome}, {Italy}, {June} 29 - {July} 2,
  2021}, pages 1--17. IEEE, 2021.
\newblock \href {https://doi.org/10.1109/LICS52264.2021.9470677}
  {\path{doi:10.1109/LICS52264.2021.9470677}}.

\bibitem{grohe_homomorphism_arxiv}
Martin Grohe, Gaurav Rattan, and Tim Seppelt.
\newblock {Homomorphism Tensors and Linear Equations}.
\newblock {\em CoRR}, abs/2111.11313, 2021.
\newblock URL: \url{https://arxiv.org/abs/2111.11313}, \href
  {http://arxiv.org/abs/2111.11313} {\path{arXiv:2111.11313}}.

\bibitem{grohe_homomorphism_2021}
Martin Grohe, Gaurav Rattan, and Tim Seppelt.
\newblock {Homomorphism} {Tensors} and {Linear} {Equations}.
\newblock In Miko{\l}aj Boja\'{n}czyk, Emanuela Merelli, and David~P. Woodruff,
  editors, {\em 49th International Colloquium on Automata, Languages, and
  Programming, {ICALP} 2022, July 4-8, 2022, Paris, France}, volume 229 of {\em
  LIPIcs}. Schloss Dagstuhl - Leibniz-Zentrum f{\"{u}}r Informatik, 2022.
\newblock \href {https://doi.org/10.4230/LIPIcs.ICALP.2022.70}
  {\path{doi:10.4230/LIPIcs.ICALP.2022.70}}.

\bibitem{Hamilton}
William~L. Hamilton.
\newblock {\em Graph Representation Learning}.
\newblock Synthesis Lectures on Artificial Intelligence and Machine Learning.
  Morgan {\&} Claypool Publishers, 2020.
\newblock \href {https://doi.org/10.2200/S01045ED1V01Y202009AIM046}
  {\path{doi:10.2200/S01045ED1V01Y202009AIM046}}.

\bibitem{immerman_canon_90}
Neil Immerman and Eric Lander.
\newblock Describing {Graphs}: {A} {First}-{Order} {Approach} to {Graph}
  {Canonization}.
\newblock In Alan~L. Selman, editor, {\em Complexity {Theory} {Retrospective}:
  {In} {Honor} of {Juris} {Hartmanis} on the {Occasion} of {His} {Sixtieth}
  {Birthday}, {July} 5, 1988}, pages 59--81. Springer New York, New York, NY,
  1990.
\newblock URL: \url{https://doi.org/10.1007/978-1-4612-4478-3_5}, \href
  {https://doi.org/10.1007/978-1-4612-4478-3\_5}
  {\path{doi:10.1007/978-1-4612-4478-3\_5}}.

\bibitem{kiefer_iteration_2020}
Sandra Kiefer and Brendan~D. McKay.
\newblock {The Iteration Number of Colour Refinement}.
\newblock In Artur Czumaj, Anuj Dawar, and Emanuela Merelli, editors, {\em 47th
  International Colloquium on Automata, Languages, and Programming (ICALP
  2020)}, volume 168 of {\em Leibniz International Proceedings in Informatics
  (LIPIcs)}, pages 73:1--73:19, Dagstuhl, Germany, 2020. Schloss
  Dagstuhl--Leibniz-Zentrum f{\"u}r Informatik.
\newblock URL: \url{https://drops.dagstuhl.de/opus/volltexte/2020/12480}, \href
  {https://doi.org/10.4230/LIPIcs.ICALP.2020.73}
  {\path{doi:10.4230/LIPIcs.ICALP.2020.73}}.

\bibitem{kiefer_planar_19}
Sandra Kiefer, Ilia Ponomarenko, and Pascal Schweitzer.
\newblock The {Weisfeiler}-{Leman} dimension of planar graphs is at most 3.
\newblock In {\em 2017 32nd {Annual} {ACM}/{IEEE} {Symposium} on {Logic} in
  {Computer} {Science} ({LICS})}, pages 1--12, Reykjavik, Iceland, June 2017.
  IEEE.
\newblock URL: \url{http://ieeexplore.ieee.org/document/8005107/}, \href
  {https://doi.org/10.1109/LICS.2017.8005107}
  {\path{doi:10.1109/LICS.2017.8005107}}.

\bibitem{lane_categories_1971}
Saunders~Mac Lane.
\newblock {\em Categories for the {Working} {Mathematician}}, volume~5 of {\em
  Graduate {Texts} in {Mathematics}}.
\newblock Springer New York, New York, NY, 1971.
\newblock URL: \url{http://link.springer.com/10.1007/978-1-4612-9839-7}, \href
  {https://doi.org/10.1007/978-1-4612-9839-7}
  {\path{doi:10.1007/978-1-4612-9839-7}}.

\bibitem{lichter_walk_2019}
Moritz Lichter, Ilia Ponomarenko, and Pascal Schweitzer.
\newblock Walk refinement, walk logic, and the iteration number of the
  {Weisfeiler}-{Leman} algorithm.
\newblock In {\em 2019 34th {Annual} {ACM}/{IEEE} {Symposium} on {Logic} in
  {Computer} {Science} ({LICS})}, pages 1--13, Vancouver, BC, Canada, June
  2019. IEEE.
\newblock URL: \url{https://ieeexplore.ieee.org/document/8785694/}, \href
  {https://doi.org/10.1109/LICS.2019.8785694}
  {\path{doi:10.1109/LICS.2019.8785694}}.

\bibitem{SBGNN}
Derek Lim, Joshua Robinson, Lingxiao Zhao, Tess~E. Smidt, Suvrit Sra, Haggai
  Maron, and Stefanie Jegelka.
\newblock Sign and basis invariant networks for spectral graph representation
  learning.
\newblock {\em CoRR}, abs/2202.13013, 2022.
\newblock URL: \url{https://arxiv.org/abs/2202.13013}, \href
  {http://arxiv.org/abs/2202.13013} {\path{arXiv:2202.13013}}.

\bibitem{Lovasz67}
L{\'{a}}szl{\'{o}} Lov{\'{a}}sz.
\newblock Operations with structures.
\newblock {\em Acta Mathematica Academiae Scientiarum Hungarica},
  18(3):321--328, September 1967.
\newblock \href {https://doi.org/10.1007/BF02280291}
  {\path{doi:10.1007/BF02280291}}.

\bibitem{lovasz_random_1996}
L{\'a}szl{\'o} Lov{\'a}sz.
\newblock Random walks on graphs: a survey.
\newblock In D.~Mikl{\'o}s, Vera~T. S{\'o}s, and T.~Sz{\"o}nyi, editors, {\em
  Combinatorics, {Paul} {Erd{\H o}s} is eighty}, volume~2 of {\em Bolyai
  {Society} {Mathematical} {Studies}}, pages 353--397. J{\'a}nos Bolyai
  Mathematical Society, Budapest, 1996.

\bibitem{lovasz_semidefinite_2009}
László Lovász and Alexander Schrijver.
\newblock Semidefinite {Functions} on {Categories}.
\newblock {\em Electron. J. Comb.}, 16(2), 2009.
\newblock \href {https://doi.org/10.37236/80} {\path{doi:10.37236/80}}.

\bibitem{mancinska_quantum_2019}
Laura Man{\v c}inska and David~E. Roberson.
\newblock Quantum isomorphism is equivalent to equality of homomorphism counts
  from planar graphs.
\newblock In {\em 2020 {IEEE} 61st {Annual} {Symposium} on {Foundations} of
  {Computer} {Science} ({FOCS})}, pages 661--672, 2020.
\newblock \href {https://doi.org/10.1109/FOCS46700.2020.00067}
  {\path{doi:10.1109/FOCS46700.2020.00067}}.

\bibitem{MontacuteS21}
Yoàv Montacute and Nihil Shah.
\newblock The {Pebble}-{Relation} {Comonad} in {Finite} {Model} {Theory}.
\newblock In Christel Baier and Dana Fisman, editors, {\em {LICS} '22: 37th
  {Annual} {ACM}/{IEEE} {Symposium} on {Logic} in {Computer} {Science},
  {Haifa}, {Israel}, {August} 2 - 5, 2022}, pages 13:1--13:11. ACM, 2022.
\newblock \href {https://doi.org/10.1145/3531130.3533335}
  {\path{doi:10.1145/3531130.3533335}}.

\bibitem{morris_sparsewl}
Christopher Morris, Gaurav Rattan, and Petra Mutzel.
\newblock Weisfeiler and leman go sparse: Towards scalable higher-order graph
  embeddings.
\newblock In H.~Larochelle, M.~Ranzato, R.~Hadsell, M.F. Balcan, and H.~Lin,
  editors, {\em Advances in Neural Information Processing Systems}, volume~33,
  pages 21824--21840. Curran Associates, Inc., 2020.
\newblock URL:
  \url{https://proceedings.neurips.cc/paper/2020/file/f81dee42585b3814de199b2e88757f5c-Paper.pdf}.

\bibitem{Morriswlgoneural}
Christopher Morris, Martin Ritzert, Matthias Fey, William~L. Hamilton, Jan~Eric
  Lenssen, Gaurav Rattan, and Martin Grohe.
\newblock Weisfeiler and leman go neural: Higher-order graph neural networks.
\newblock {\em Proceedings of the AAAI Conference on Artificial Intelligence},
  33(01):4602--4609, Jul. 2019.
\newblock URL: \url{https://ojs.aaai.org/index.php/AAAI/article/view/4384},
  \href {https://doi.org/10.1609/aaai.v33i01.33014602}
  {\path{doi:10.1609/aaai.v33i01.33014602}}.

\bibitem{neuen_exponential_2018}
Daniel Neuen and Pascal Schweitzer.
\newblock An exponential lower bound for individualization-refinement
  algorithms for graph isomorphism.
\newblock In {\em Proceedings of the 50th Annual ACM SIGACT Symposium on Theory
  of Computing}, STOC 2018, page 138–150, New York, NY, USA, 2018.
  Association for Computing Machinery.
\newblock \href {https://doi.org/10.1145/3188745.3188900}
  {\path{doi:10.1145/3188745.3188900}}.

\bibitem{perrone_notes_2021}
Paolo Perrone.
\newblock Notes on {Category} {Theory} with examples from basic mathematics,
  February 2021.
\newblock URL: \url{http://arxiv.org/abs/1912.10642}.

\bibitem{chendi_rattan}
Chendi Qian, Gaurav Rattan, Floris Geerts, Christopher Morris, and Mathias
  Niepert.
\newblock Ordered subgraph aggregation networks.
\newblock {\em CoRR}, abs/2206.11168, 2022.
\newblock \href {http://arxiv.org/abs/2206.11168} {\path{arXiv:2206.11168}},
  \href {https://doi.org/10.48550/arXiv.2206.11168}
  {\path{doi:10.48550/arXiv.2206.11168}}.

\bibitem{rattan_weisfeiler_2023}
Gaurav Rattan and Tim Seppelt.
\newblock Weisfeiler-leman and graph spectra.
\newblock In {\em Proceedings of the 2023 Annual ACM-SIAM Symposium on Discrete
  Algorithms (SODA)}, pages 2268--2285, 2023.
\newblock \href {https://doi.org/10.1137/1.9781611977554.ch87}
  {\path{doi:10.1137/1.9781611977554.ch87}}.

\bibitem{roberson_oddomorphisms_2022}
David~E. Roberson.
\newblock Oddomorphisms and homomorphism indistinguishability over graphs of
  bounded degree, June 2022.
\newblock URL: \url{http://arxiv.org/abs/2206.10321}.

\bibitem{specht_zur_1940}
Wilhelm Specht.
\newblock Zur {Theorie} der {Matrizen}. {II}.
\newblock {\em Jahresbericht der Deutschen Mathematiker-Vereinigung},
  50:19--23, 1940.
\newblock URL:
  \url{http://gdz.sub.uni-goettingen.de/dms/load/toc/?PPN=PPN37721857X_0050&DMDID=dmdlog6}.

\bibitem{spitzer_characterising_2022}
Gian~Luca Spitzer.
\newblock {\em Characterising {Fragments} of {First}-{Order} {Logic} by
  {Counting} {Homomorphisms}}.
\newblock B.{Sc}. {Thesis}, RWTH Aachen University, Aachen, 2022.

\bibitem{van_dam_which_2003}
Edwin~R. van Dam and Willem~H. Haemers.
\newblock Which graphs are determined by their spectrum?
\newblock {\em Linear Algebra and its Applications}, 373:241--272, November
  2003.
\newblock URL:
  \url{https://linkinghub.elsevier.com/retrieve/pii/S002437950300483X}, \href
  {https://doi.org/10.1016/S0024-3795(03)00483-X}
  {\path{doi:10.1016/S0024-3795(03)00483-X}}.

\bibitem{von_luxburg_tutorial_2007}
Ulrike von Luxburg.
\newblock A tutorial on spectral clustering.
\newblock {\em Statistics and Computing}, 17(4):395--416, December 2007.
\newblock \href {https://doi.org/10.1007/s11222-007-9033-z}
  {\path{doi:10.1007/s11222-007-9033-z}}.

\bibitem{weisfeiler_construction_76}
Boris Weisfeiler.
\newblock {\em On {Construction} and {Identification} of {Graphs}}, volume 558
  of {\em Lecture {Notes} in {Mathematics}}.
\newblock Springer Berlin Heidelberg, Berlin, Heidelberg, 1976.
\newblock \href {https://doi.org/10.1007/BFb0089374}
  {\path{doi:10.1007/BFb0089374}}.

\bibitem{wiegmann_necessary_1961}
N.~A. Wiegmann.
\newblock Necessary and sufficient conditions for unitary similarity.
\newblock {\em Journal of the Australian Mathematical Society}, 2(1):122--126,
  1961.
\newblock \href {https://doi.org/10.1017/S1446788700026422}
  {\path{doi:10.1017/S1446788700026422}}.

\bibitem{XHLJ19}
Keyulu Xu, Weihua Hu, Jure Leskovec, and Stefanie Jegelka.
\newblock How powerful are graph neural networks?
\newblock In {\em 7th International Conference on Learning Representations,
  {ICLR} 2019, New Orleans, LA, USA, May 6-9, 2019}. OpenReview.net, 2019.
\newblock URL: \url{https://openreview.net/forum?id=ryGs6iA5Km}.

\end{thebibliography}

\end{small}
\newpage

\appendix

\section{Full Proof of \cref{thm:camb} and \cref{cor:camb}}
\label{app:thm:camb}

To ease notation, we consider for $k \in \mathbb{N}$ homomorphisms between $k$-labelled graphs $\boldsymbol{F} = (F, \vec{u})$ and $\boldsymbol{G} = (G, \vec{v})$ which are homomorphisms $h  \colon F\to G$ such that $h(\vec{u}) = \vec{v}$ entry-wise. We write $\hom(\boldsymbol{F}, \boldsymbol{G})$ for the number of such homomorphisms.

\subsection{From Graphs to Formulas}

The following \cref{lem:forward} is an adaptation of \cite[Lemma~4]{dvorak_recognizing_2010}. Recall \cref{def:pfc-labelled}.

\begin{lemma} \label{lem:forward}
	Let $k \geq 1$ and $d \geq 0$.
	For every $k$-labelled graph $\boldsymbol{F}$ admitting a $k$-pebble forest cover of depth $d$ and every $m \geq 0$,
	there exists a formula $\varphi_{\boldsymbol{F}}^m \in\mathsf{C}_{k}^d$ with $k$ free variables such that $\hom(\boldsymbol{F},\boldsymbol{K}) = m$ iff $\boldsymbol{K} \models \varphi_{\boldsymbol{F}}^m$ for all $k$-labelled graphs $\boldsymbol{K}$. 
\end{lemma}
\begin{proof}
	The lemma is proven by induction over $d$.
	If $d= 0$ then every vertex of $\boldsymbol{F}  = (F, \vec{u})$ is labelled. 
	Hence $\hom(\boldsymbol{F},\boldsymbol{K})$ must be zero or one. Hence, if $m>1$, set $\varphi^m_{\boldsymbol{F}}$ to be false. If $m=1$, set $\varphi^m_{\boldsymbol{F}}$ to the quantifier-free formula encoding the atomic type of $\boldsymbol{u}$, i.e.\@
	\[
	\phi^1_{\boldsymbol{F}}(x_1, \dots, x_k) \coloneqq \bigwedge_{\{ \vec{u}_i, \vec{u}_j \} \in E(F)} E(x_i, x_j) \land \bigwedge_{\vec{u}_i = \vec{u}_j} (x_i = x_j).
	\]
	If $m=0$, set $\varphi^m_{\boldsymbol{F}}$ to be the negation of $\phi^1_{\boldsymbol{F}}$.
	
	For the inductive case, denote the $k$-pebbling function of $\boldsymbol{F}$ by $p$ and write $\leq$ for the forest order induced by the cover. Let $W$ be the set of minimal unlabelled vertices in $\boldsymbol{F}$, in symbols
	\[
	W \coloneqq \left\{w \in V(F) \mid w > \vec{u}_1, \dots, \vec{u}_k \land \forall x \in V(F).\ w \geq x > \vec{u}_1, \dots, \vec{u}_k \implies w = x\right\}.
	\]
	For every $w \in W$ separately and  $m \geq 0$, a formula $\phi_w^{m}$ constructed. 
	The final formula $\phi_{\boldsymbol{F}}^m$ will be assembled from these.
	Distinguish cases:
	\begin{itemize}
		\item There exists a labelled vertex $v$ such that $p(v) = p(w)$.
		
		Let $\ell$ be a label on $v$.
		Define a $k$-labelled graph $\boldsymbol{F}'$ from $\boldsymbol{F}$ as follows:
		\begin{enumerate}
			\item Delete all vertices which are incomparable with $w$.
			\item If $v$ has only one label, delete $v$. If $v$ has more than one label, remove the label $\ell$ from $v$.
			\item Add the label $\ell$ to the vertex $w$.
		\end{enumerate}
		Then $\boldsymbol{F}'$ is a $k$-labelled graph with a $k$-pebble forest cover of depth $d-1$.
		By the inductive hypothesis, for every integer $m'\geq 0$, there exists a $\mathsf{C}_{k}^{d-1}$-formula $\varphi_{\boldsymbol{F}'}^{m'}$ with $k$~free variables such $\hom(\boldsymbol{F}',\boldsymbol{K}) = m'$ if and only if $\boldsymbol{K} \models \varphi_{\boldsymbol{F}'}^{m'}$ for every $k$-labelled graph $\boldsymbol{K}$.
		Define the $\mathsf{C}_{k}^d$-formula $\varphi_{w}^{c,m}$ to be the conjunction of the following three formulas:
		\begin{itemize}
			\item the conjunction of $E(x_\ell,x_q)$ over $q \in [k]$ such that there is an edge between $ \boldsymbol{u}_\ell= v$ and $\boldsymbol{u}_q$ in $\boldsymbol{F}$, 
			\item the conjunction of $(x_\ell =x_q)$ over $q \in [k]$ such that $\boldsymbol{u}_\ell=\boldsymbol{u}_q=v$,
			and,  
			\item $\exists^{=c} x_\ell.\ \varphi_{\boldsymbol{F}'}^{m}$. 
		\end{itemize}
		
		\item The value $p(w)$ does not occur among the $p(v)$ for $v$ labelled.
		
		Since $p$ is injective on labelled vertices, there must exist a labelled vertex $v$ with at least two labels. Let $\ell \neq \ell'$ be two of these labels on $v$. 
		Consider the $k$-labelled graph $\boldsymbol{F}'$ obtained from $\boldsymbol{F}$ as follows: 
		\begin{enumerate}
			\item Delete all vertices which are incomparable with $w$.
			\item Remove the label $\ell$ from $v$.
			\item Add the label $\ell$ to the vertex $w$.
		\end{enumerate}
		By inductive hypothesis, for every integer $m'\geq 0$, there exists a $\mathsf{C}_k^{d-1}$-formula $\varphi_{\boldsymbol{F}'}^{m'}$ such that for every $k$-labelled graph $\boldsymbol{K}$, $\hom(\boldsymbol{F}',\boldsymbol{K}) = m'$ if and only if $\boldsymbol{K} \models \varphi_{\boldsymbol{F}'}^{m'}$. 
		Define the $\mathsf{C}_k^d$-formula $\varphi_{w}^{c,m}$ to be the conjunction of $\exists^{=c}x_\ell.\ \varphi_{\boldsymbol{F}'}^{m}$ and $x_{\ell} = x_{\ell'}$.
	\end{itemize}
	In both cases, the formula $\phi_w^m$ is constructed as follows.
	If $m \geq 1$, the formula $\phi_w^m$ is the disjunction over conjunctions of the form $\varphi_{w}^{c_1,m_1} \land \dots \land \varphi_{w}^{c_r,m_r} \land \exists^{=c}x_\ell.\ \neg\varphi_{\boldsymbol{F}'}^0$ such that $m = c_1 m_1 + \cdots +c_r m_r$ and $c= c_1+\cdots+c_r$ for $m_1, \dots, m_r, c_1, \dots, c_r \geq 1$ and $r \in \mathbb{N}$. If $m=0$ then set $\varphi_{w}^m \coloneqq \forall x_\ell.\ \varphi_{\boldsymbol{F}'}^0$.

	Finally, the formulas constructed for each $w \in W$ are combined.
	The formula $\phi^m_{\boldsymbol{F}}$ is the disjunction over the conjunctions $\phi^{m_1}_{w_1} \land \dots \land \phi^{m_r}_{w_r}$ where $W = \{w_1, \dots, w_r\}$ and $m = m_1 \cdots m_r$ for positive integers $m_i$ if $m > 0$ and $\bigvee_{w \in W} \phi^0_w$ otherwise.
	
	It remains to observe that the quantifier depth of $\phi_{\boldsymbol{F}}^m$ is at most the quantifier depth of $\phi_{\boldsymbol{F}'}^m$ plus one.
\end{proof}

\begin{corollary} \label{cor:forward}
	Let $k \geq 1$ and $d \geq 0$.
	For every $k$-labelled graph $\boldsymbol{F}$ admitting a $(k+1)$-pebble forest cover of depth $d$ and every $m \geq 0$,
	there exists a formula $\varphi_{\boldsymbol{F}}^m \in\mathsf{C}_{k+1}^d$ with $k$ free variables such that $\hom(\boldsymbol{F},\boldsymbol{K}) = m$ iff $\boldsymbol{K} \models \varphi_{\boldsymbol{F}}^m$ for all $k$-labelled graphs $\boldsymbol{K}$. 
\end{corollary}
\begin{proof}
	Turn $\boldsymbol{F} = (F, \vec{u})$ into a $(k+1)$-labelled graph $\boldsymbol{F}' = (F, \vec{u}\vec{u}_k)$ by duplicating the last label. This does not affect the associated pebble forest cover.
	For every $m \geq 0$, \cref{lem:forward} implies the existence of $\mathsf{C}_{k+1}^d$-formula $\varphi_{\boldsymbol{F}'}^m$ with $k+1$~free variables such that $\hom(\boldsymbol{F}',\boldsymbol{K}') = m$ if and only if $\boldsymbol{K}' \models \varphi_{\boldsymbol{F}'}^m$ for all $(k+1)$-labelled graphs $\boldsymbol{K}'$.
	
	Set $\varphi_{\boldsymbol{F}}^m(x_1, \dots, x_k) \coloneqq \varphi_{\boldsymbol{F}'}^m(x_1, \dots, x_{k+1}) \land (x_k = x_{k+1})$. This is a $\mathsf{C}_{k+1}^d$-formula with $k$~free variables. To verify the assertion, let $\boldsymbol{K}$ be a $k$-labelled graph and write $\boldsymbol{K}'$ for the $(k+1)$-labelled graph obtained from it by duplicating the last label as above. Then for every $m \geq 0$,
	\begin{align*}
		\hom(\boldsymbol{F}, \boldsymbol{K}) = m
		\iff \hom(\boldsymbol{F}', \boldsymbol{K}') = m
		\iff \boldsymbol{K}' \models \phi_{\boldsymbol{F}'}^m
		\iff \boldsymbol{K} \models \phi_{\boldsymbol{F}}^m.
	\end{align*}
	This concludes the proof.
\end{proof}

\subsection{From Formulas to Graphs}

\begin{definition}
	Let $k \geq \ell \geq 1$ and $d \geq 0$.
	An $\ell$-labelled graph $\boldsymbol{F} = (F, \vec{u})$ admits a \emph{regular $k$-pebble forest cover of depth $d$} if it has $k$-pebble forest cover of depth $d$ as defined in \cref{def:pfc-labelled} such that $\vec{u}_1 \leq \dots \leq \vec{u}_\ell$ and $p(\vec{u}_i) = i$ for all $i \in [k]$.
\end{definition}

The relevant property of such regular pebble forest covers is that if $\boldsymbol{F}$ and $\boldsymbol{F}'$ both admit $k$-pebble forest cover of depth $d$ then so does their gluing product $\boldsymbol{F} \odot \boldsymbol{F}'$. Indeed, the labelled vertices lie on a path of the cover in the same order and the pebbling functions of the covers of $\boldsymbol{F}$ and $\boldsymbol{F}'$ agree on them. Therefore, the two pebble forest covers can be combined.

Subject to the next lemmas are relational structures $F$ over the signature $\{E, I\}$ comprising the binary relation symbols $E$ and $I$. A \emph{relational structure} over $\{E, I\}$ is a tuple $(U, R^E, R^I)$ where $U$ is a finite set, the \emph{universe}, and $R^E, R^I \subseteq U \times U$. For the purpose of this work, $R^E$ and $R^I$ are always symmetric relations. A \emph{homomorphism} between $\{E, I\}$-structures $(V, S^E, S^I)$ is a map $h \colon U \to S$ such that for all $(x,y) \in R^E$ it holds that $(h(x), h(y)) \in S^E$ and similarly for $R^I$ and $S^I$. See \cite{grohe_descriptive_2017} for background on relational structures.
A \emph{$k$-labelled $\{E, I\}$-structure} is a tuple  $(U, R^E, R^I, \vec{u})$ where $\vec{u} \in U^k$. The notion of a homomorphism extends to $k$-labelled structures. The homomorphism count function $\hom(-,-)$ can be defined analogously as for graphs.

The relation symbol $E$ will be understood as the binary edge relation of the graphs which are of our ultimate interest. $I$ is an additional binary relation. 
A graph $G$ can be turned into an $\{E, I\}$-structure by interpreting $I$ as equality, i.e.\@ $I^G \coloneqq \{(v,v) \mid v \in V(G)\}$. This extends to labelled graphs.
A (labelled) $\{E, I\}$-structure admits a (regular) $k$-pebble forest cover of depth $d$ if the underlying graph admits such a cover and this cover satisfies all axioms for the $I$-relation in place of the $E$-relation, mutatis mutandis. In the language of \cite{dawar_lovasz_2021}, this cover is a cover for the relational structure.

We recall the following notions from \cite{dvorak_recognizing_2010}:
A \emph{$k$-labelled quantum relational structure} is a $\mathbb{C}$-linear combination of $k$-labelled relational structures called its \emph{constituents}. 
For $t = \sum \alpha_i \boldsymbol{F}_i$ and a $k$-labelled $\boldsymbol{G}$, write $\hom(t, \boldsymbol{G})$ for the number $\sum \alpha_i \hom(\boldsymbol{F}_i, \boldsymbol{G})$.

The following \cref{claim:dvo2,lem:backward} are straightforward variations of \cite[Lemma 5]{dvorak_recognizing_2010} and \cite[Lemma 6]{dvorak_recognizing_2010} respectively.

\begin{lemma} \label{claim:dvo2}
	Let $n \geq r \geq 1$, $k \geq \ell \geq 1$, and $d \geq 0$.
	Let $t$ be an $\ell$-labelled quantum relational structure over the signature $\{E, I\}$ whose constituents admit regular $k$-pebble forest covers of depth $d$.
	Then there exists an $\ell$-labelled quantum relational structure $t^{r, n}$ over the signature $\{E, I\}$ whose constituents admit regular $k$-pebble forest covers of depth $d$ such that for all $\ell$-labelled relational structures $\boldsymbol{G}$ over $\{E, I\}$,
	\begin{itemize}
		\item if $\hom(t, \boldsymbol{G}) \in \{0, \dots, r-1 \}$ then $\hom(t^{r, n}, G) = 0$, and
		\item if $\hom(t, \boldsymbol{G}) \in \{r, \dots,  n\}$ then $\hom(t^{r, n}, G)  = 1$.
	\end{itemize}
\end{lemma}
\begin{proof}
	Let $p(X) = \sum \alpha_i X^i$ denote a polynomial such that $p(i) = 0$ for all $i \in \{0, \dots, r-1 \} $ and $p(i) = 1$ for all $i \in \{r, \dots,  n\}$. Let $t^{r, n} \coloneqq \sum \alpha_i t^{\circ i}$. Since $\hom(t^{\odot i}, \boldsymbol{G}) = \hom(t, \boldsymbol{G})^i$ for all $i \geq 0$, this quantum relational structure has the desired properties.
\end{proof}

\begin{lemma} \label{lem:backward}
	Let $n \geq 1$, $k \geq 1$, and $d \geq 0$.
	For every formula $\phi \in \mathsf{C}_k^d$ with $k$~free variables, 
	there exists a $k$-labelled quantum relational structure over the signature $\{E, I\}$ whose constituents admit regular $k$-pebble forest covers of depth $d$
	such that for all $n$-vertex $k$-labelled relational structures $\boldsymbol{G}$ over $\{E, I\}$ where $I$ is interpreted as equality,
	\begin{itemize}
		\item if $\boldsymbol{G} \models \phi$ then $\hom(t, \boldsymbol{G}) = 1$, and
		\item if $\boldsymbol{G} \not\models \phi$ then $\hom(t, \boldsymbol{G}) = 0$.
	\end{itemize}
\end{lemma}
\begin{proof}
	By the induction on the structure of $\phi$. Let $\boldsymbol{1}$ denote the $k$-labelled relational structure over $\{E, I\}$ with universe $\{1, \dots, k\}$, without edges, and $k$ distinctly labelled vertices. It admits a regular $k$-pebble forest cover of depth zero. Distinguish cases:
	\begin{itemize}
		\item If $\phi = \mathsf{true}$ then $\boldsymbol{1}$ models $\phi$. If $\phi = \mathsf{false}$ then $0$, i.e.\@ the empty linear combination models $\phi$.
		\item If $\phi = (x_i = x_j)$ then the structure obtained from $\boldsymbol{1}$ by adding $I(i, j)$ models $\phi$,
		\item If $\phi = E(x_i, x_j)$ then the structure obtained from $\boldsymbol{1}$ by adding $E(i, j)$ models $\phi$,
		\item If $\phi = \phi_1 \land \phi_2$, let $t_1$ denote the quantum graph modelling $\phi_1$ and $t_2$ the quantum graph modelling~$\phi_2$. Their Schur product $t_1 \circ t_2$ models $\phi$. Since the pebbling function coincides with the label function on labelled vertices, this is a valid operation.
		\item If $\phi = \phi_1 \lor \phi_2$ and $t_1$ and $t_2$ are as above, then $\boldsymbol{1} - (\boldsymbol{1}-t_1) \circ (\boldsymbol{1}-t_2)$ models $\phi$.
		\item If $\phi = \neg \phi_1$ and $t_1$ is as above then $\boldsymbol{1} - t_1$ models $\phi$. 
		\item If $\phi = \exists^{\geq m}x_{i}.\ \psi$ then let $t$ denote the quantum graph modelling $\psi$.
		Let $t'$ denote the $k$-labelled quantum graph obtained from $t$ by applying the following operations to each of its constituents $\boldsymbol{F}$. Let $(\mathcal{F}, p)$ denote the regular $k$-pebble forest cover of depth $d$ of $\boldsymbol{F}$.
		\begin{enumerate}
			\item Drop the $i$-th label from the vertex $w$ carrying it. 
			Since $(\mathcal{F}, p)$ is regular, $w$ is unlabelled now. 
			\item Introduce a fresh vertex $w'$ carrying label $i$.
			\item Update $\mathcal{F}$ to $\mathcal{F}'$ by inserting $w'$ at the original position of $w$ and moving $w$ such that $w$ succeeds all labelled vertices and precedes all other unlabelled vertices.
			\item Update $p$ to $p'$ by $p'(x) \coloneqq p(x)$ for all $x \neq w'$ and $p'(w') \coloneqq p(w)$.
		\end{enumerate}
		It has to be verified that the resulting pebble forest covered labelled graph satisfies all axioms. Critical is only the validity of the pebbling function.
		Let $v \leq_{\mathcal{F}'} v'$ be such that $E(v,v')$ or $I(v,v')$. It is claimed that $p'(v) \neq p'(x)$ for all $x \in (v,v']_{\mathcal{F}'}$. 
		Wlog, $v \neq v'$.
		Distinguish cases:
		\begin{itemize}
			\item If $w \neq v$ and $w \neq v'$ then $(v, v']_{\mathcal{F}'} \setminus (v, v']_{\mathcal{F}} \subseteq \{w, w'\}$. If $w$ or $w'$ is contained in this difference then $v$ is labelled because it precedes the labelled vertex $w'$ or the least unlabelled vertex $w$. Hence, $p'(v) = p(v) \neq p(w) = p'(w) = p'(w')$. This implies that $p'(v) \neq p'(x)$ for all $x \in (v,v']_{\mathcal{F}'}$ because the analogous condition held for $\mathcal{F}$.
			\item If $w = v$ then $(v, v']_{\mathcal{F}'} \subseteq (v, v']_{\mathcal{F}}$ since $v'$ must, as an unlabelled vertex, succeed all labelled vertices. The claim follows readily.
			\item If $w = v'$ then $v$ is labelled and the claim follows since $p'$ is injective on labelled vertices and $v \neq w'$ because $w'$ neither appears in the relation $E$ nor in $I$.
		\end{itemize}
		Hence, the resulting object is a regular $k$-pebble forest cover of depth $d+1$. The $k$-labelled quantum relational structure $t'$ satisfies all desired properties.
		For $\boldsymbol{G} = (G, \vec{v})$, $\hom(t', \boldsymbol{G})$ counts the number of $v \in V(G)$ such that $G \models \psi(\vec{v}_1, \dots, \vec{v}_{i-1}, v, \vec{v}_{i+1}, \dots,\vec{v}_k)$. Apply \cref{claim:dvo2} with $r = m$. The resulting $k$-labelled quantum relational structure models $\phi$.
	\end{itemize} 
	In this construction, the depth of the structure only increases when handling quantifiers. This implies that $\boldsymbol{F}$ has indeed the desired properties.
\end{proof}

\begin{lemma}[{\cite[Proposition~23]{dawar_lovasz_2021_arxiv}}] \label{lem:dawar}
	Let $k \geq 1$ and $d \geq  0$.
	Let $\boldsymbol{F}$ be a $k$-labelled relational structure over $\{E, I\}$ admitting a $k$-pebble forest cover of depth $d$.
	Then the $k$-labelled graph $\boldsymbol{F}'$ obtained from $\boldsymbol{F}$ by identifying pairs of vertices satisfying $I$ admits a $k$-pebble forest cover of depth $d$.
\end{lemma}
\begin{proof}
	For the sake of completeness, we briefly sketch a proof.
	On {\cite[p.~15]{dawar_lovasz_2021_arxiv}} a procedure is described to construct a $k$-pebble forest cover of depth $d$ for $\boldsymbol{F}' = (F', \vec{u}')$ from the corresponding object for $\boldsymbol{F} = (F, \vec{u})$. The construction iteratively identifies pairs of vertices $u \neq v$ which satisfy $I(u,v)$ in $\boldsymbol{F}$. Since in this case $u$ and $v$ are comparable in the forest order, one may suppose wlog that $u < v$. The construction operates as follows:
	\begin{enumerate}
		\item The universe of $F'$ is $V(F) \setminus \{v\}$.
		\item The forest order $\leq'$ of $F'$ is restriction of the order $\leq$ of $F$ to the new universe.
		\item In each relation of $F$, each occurring $v$ is replaced by $u$ to form a relation of $F'$.
		\item If $v$ carried labels in $\boldsymbol{F}$ then these labels are moved to $u$ in $\boldsymbol{F}'$. 
		
		Note that in this case $u$ carried at least one label in $\boldsymbol{F}$ since $u \leq v$.
		\item The pebbling function $p$ is updated to $p'$ such that $p'(w) = p(w)$ for all $w$ such that $w < v$ or $p(w) \not\in \{p(u), p(v)\}$.
		Otherwise, $p'(w) \in \{p(u), p(v)\}$ is set according to a rule described in {\cite[p.~15]{dawar_lovasz_2021_arxiv}} such that the $p'$ is a $k$-pebbling function for $F'$.
		For our purposes, it has to be verified additionally that $p'$ is injective on labelled vertices, cf.\@ \cref{def:pfc-labelled}.
		
		After the update, $p'$ is injective on labelled vertices since $p(x) = p'(x)$ for all labelled vertices $x$ in $\boldsymbol{F}'$.
		Indeed, if $p'(x) \neq p(x)$ for some labelled vertex $x$ in $\boldsymbol{F}'$ then $x \geq v$ and $p(x) \in \{p(u), p(v)\}$.
		The first statement implies $x \neq u$ while the second statement implies $x = u$ since $p$ is injective on labelled vertices and $x \neq v$ because $v$ was deleted. A contradiction.
	\end{enumerate}
	It is argued in \cite{dawar_lovasz_2021_arxiv} that the resulting object is a $k$-pebble forest cover of depth $d$ for the updated structure.
	As argued above, the operation also preserves the properties stipulated in \cref{def:pfc-labelled}.
\end{proof}

\begin{corollary} \label{cor:backward}
	Let $n \geq 1$, $k \geq 1$, and $d \geq 0$.
	For every formula $\phi \in \mathsf{C}_{k+1}^d$ with $k$~free variables, 
	there exists a $k$-labelled quantum graph $t$ whose constituents admit $(k+1)$-pebble forest covers of depth $d$ such that for all $k$-labelled $n$-vertex graphs $\boldsymbol{G}$
	\begin{itemize}
		\item if $\boldsymbol{G} \models \phi$ then $\hom(t, \boldsymbol{G}) = 1$, and
		\item if $\boldsymbol{G} \not\models \phi$ then $\hom(t, \boldsymbol{G}) = 0$.
	\end{itemize}
\end{corollary}
\begin{proof}
	Given $\phi$, define $\phi'(x_1, \dots, x_{k+1}) \coloneqq \phi(x_1, \dots, x_k) \land (x_k = x_{k+1})$.
	This is a $\mathsf{C}_{k+1}^d$-formula with $k+1$~free variables.
	\cref{lem:backward} guarantees the existence of a $(k+1)$-labelled quantum relational structure $t'$ over the signature $\{E, I\}$
	whose constituents admit regular $(k+1)$-pebble forest covers of depth $d$ such that for all $(k+1)$-labelled structures $\boldsymbol{G}$ over $\{E, I\}$ with $I$ interpreted as equality,
	\begin{itemize}
		\item if $\boldsymbol{G} \models \phi$ then $\hom(t, \boldsymbol{G}) = 1$, and
		\item if $\boldsymbol{G} \not\models \phi$ then $\hom(t, \boldsymbol{G}) = 0$.
	\end{itemize}
	Inspecting the proof of \cref{lem:backward} shows that in all constituents $\boldsymbol{F} = (F, \vec{u})$ of $t'$ are such that $F \models I(\vec{u}_k, \vec{u}_{k+1})$.
	By applying \cref{lem:dawar} to every constituent of $t'$, 
	a $(k+1)$-labelled quantum graph $t''$ is obtained whose constituents admit $(k+1)$-pebble forest covers of depth $d$. Furthermore, in each of these constituents the $k$-th and $(k+1)$-st label coincide. Write $t$ for the $k$-labelled quantum graph obtained from $t''$ by dropping the $(k+1)$-st label in every constituent.
	Let $\boldsymbol{G} = (G, \vec{v})$ be an arbitrary $k$-labelled graph and write $\boldsymbol{G}' = (G', \vec{v})$ for the $\{E, I\}$-structure obtained from $\boldsymbol{G}$ by interpreting $I$ as equality. Hence, if $\boldsymbol{G} \models \phi$ then  $(G, \vec{v}\vec{v}_{k+1}) \models \phi'$,
	\[
	\hom(t, \boldsymbol{G})
	= \sum_{v \in V(G)} \hom(t'', (G, \vec{v}v))
	= \sum_{v \in V(G)} \hom(t', (G', \vec{v}v))
	= \hom(t', (G', \vec{v}\vec{v}_k))
	= 1.
	\]
	If $\boldsymbol{G} \not\models \phi$ then $(G, \vec{v}\vec{v}_{k+1}) \not\models \phi'$ and the above equation evaluates to zero.
\end{proof}

\subsection{Conclusion}

\begin{proof}[Proof of \cref{thm:camb}]
	The equivalence of \cref{camb3,camb2} is standard and can be found in \cite[Proof of Theorem~5.2]{cai_optimal_1992}.
	\Cref{cor:forward} yields that \cref{camb2} implies \cref{camb1}. The backward direction follows from \cref{cor:backward}.
\end{proof} 

\begin{proof}[Proof of \cref{cor:camb}]
	The forward direction follows immediately from \cref{thm:camb} via the equivalence between \Cref{camb3,camb1} since the graphs in $\mathcal{PFC}_k^d$ are precisely those underlying labelled graphs in $\mathcal{LPFC}_k^d$.

	Conversely, suppose that $G$ and $H$ are distinguished after $d$ iterations of $k$-\WL. 
	Let $c$ be a colour of $k$-tuples which occurs in $G$ and $H$ differently often. For every other colour $c' \neq c$ there exists a formula $\phi_{c,c'} \in \mathsf{C}_{k+1}^d$ with $k$~free variables which is satisfied by $k$-tuples of colour $c$ but violated by $k$-tuples of colour $c'$. Hence, the number of tuples satisfying $\phi \coloneqq \bigwedge_{c' \neq c} \phi_{c,c'}$ is different in $G$ and $H$.
	Let $t$ be a $k$-labelled quantum graph whose constituents admit $(k+1)$-pebble forest covers of depth $d$ satisfying the assertions of \cref{cor:backward} for the formula $\phi$. Hence, the sum of all numbers $\hom(t,(G,\vec{u}))$ over $\vec{u} \in V(G)^k$ is not equal to the sum of all numbers $\hom(t,(H,\vec{v}))$ over $\vec{v} \in V(H)^k$. 	
	This implies the existence of a constituent $\boldsymbol{F} = (F, \vec{u})$ of $t$ such that $\hom(F,G) \neq \hom(F,H)$. Hence, $G$ and $H$ are homomorphism distinguishable over the desired graph class. 
\end{proof}

\section{Material Omitted in \Cref{sec:onetwowl}}

\subsection{Algebraic Characterisation of $(1,1)$-\WL Indistinguishability}
\label{ssec:app-algoneonewl}

Let $\mathcal{T}^+$ denote the class of graphs $T/B$ which are obtained from a forest $T$ by contracting all vertices of a set $\emptyset \neq B \subseteq V(T)$ into a single vertex and removing possible loops and multiedges. Clearly, every forest is in $\mathcal{T}^+$ taking $B$ to be any singleton. Moreover, all cycles are in $\mathcal{T}^+$ as they can be obtained by contracting the ends of paths. Let $\mathcal{T}^\circ$ denote the set of $2$-labelled graphs $\boldsymbol{T} = (T/B, u_1u_2)$ where $T/B \in \mathcal{T}^+$, $u_1, u_2 \in V(T/B)$, and $u_1 \in B$.

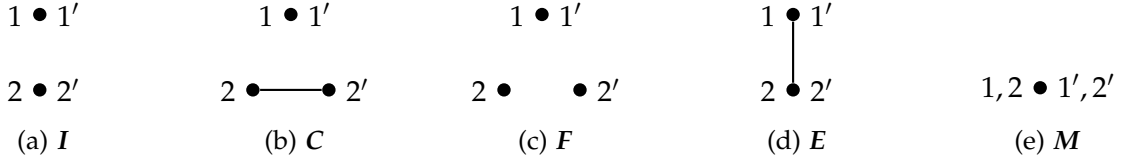
\begin{figure}[t]
\centering
\captionsetup[subfigure]{justification=centering}
\tikzset{
vertex/.style = {fill,circle,inner sep=0pt,minimum size=5pt},
edge/.style = {-,thick},
lbl/.style={color=lightgray}
}
\begin{subfigure}[t]{0.19 \textwidth}
\centering
\begin{tikzpicture}
\node[vertex] (a) [label = {left:$1\vphantom{1'}$}, label={right:$1'\vphantom{1'}$}] {} ;
\node[vertex] (b) [label = {left:$2\vphantom{2'}$}, label={right:$2'\vphantom{2'}$}] [below of= a] {}; 
\end{tikzpicture}
\caption{$\boldsymbol{I}$}
\end{subfigure}
\begin{subfigure}[t]{0.19 \textwidth}
\centering
\begin{tikzpicture}
\node[vertex] (a) [label = {left:$1\vphantom{1'}$}, label={right:$1'\vphantom{1'}$}] {} ;
\node[] (vemp) [below of = a] {};
\node[vertex] (b) [label={left:$2\vphantom{2'}$}] [below of=a, xshift=-.5cm] {}; 
\node[vertex] (c) [label={right:$2'\vphantom{2'}$}] [below of=a, xshift=.5cm] {}; 
\draw[edge] (b) -- (c) {}; 
\end{tikzpicture}
\caption{$\boldsymbol{C}$}
\end{subfigure}
\begin{subfigure}[t]{0.19 \textwidth}
 \centering 
\begin{tikzpicture}
\node[vertex] (a) [label = {left:$1\vphantom{1'}$}, label={right:$1'\vphantom{1'}$}] {} ;
\node[] (vemp) [below of = a] {};
\node[vertex] (b) [label={left:$2\vphantom{2'}$}] [below of=a, xshift=-.5cm] {}; 
\node[vertex] (c) [label={right:$2'\vphantom{2'}$}] [below of=a, xshift=.5cm] {}; 
\end{tikzpicture}
\caption{$\boldsymbol{F}$}
\end{subfigure}
\begin{subfigure}[t]{0.19 \textwidth}
\centering
\begin{tikzpicture}
\node[vertex] (a) [label = {left:$1\vphantom{1'}$}, label={right:$1'\vphantom{1'}$}] {} ;
\node[vertex] (b) [label = {left:$2\vphantom{2'}$}, label={right:$2'\vphantom{2'}$}] [below of= a] {}; 
\draw[edge] (a) -- (b) {}; 
\end{tikzpicture}
\caption{$\boldsymbol{E}$}
\end{subfigure}
\begin{subfigure}[t]{0.19 \textwidth}
 \centering
\begin{tikzpicture}
\node[vertex] (a) [label = {left:$1\vphantom{1'}, 2$}, label={right:$1'\vphantom{1'}, 2'$}] {} ;
\end{tikzpicture}
\caption{$\boldsymbol{M}$}
\end{subfigure}
\caption{The $(2,2)$-bilabelled graphs featured in \cref{thm:chhom}}
\label{fig:chhom}
\end{figure}

The set $\mathcal{T}^\circ$ is closed under parallel composition. Moreover, it is closed under series composition with the following $(2,2)$-bilabelled graphs, which are depicted in \cref{fig:chhom}:
\begin{itemize}
\item the \emph{identity graph} $\boldsymbol{I} = (I, (1,2), (1,2))$ with $V(I) = \{1,2\}$ and $E(I) = \emptyset$,
\item the \emph{connect graph} $\boldsymbol{C} = (C, (1,2), (1,2'))$ with $V(C) = \{1,2,2'\}$ and $E(C) = \{22'\}$,
\item the \emph{forget graph} $\boldsymbol{F} = (F, (1,2), (1,2'))$ with $V(F) = \{1,2,2'\}$ and $E(F) = \emptyset$,
\item the \emph{edge graph} $\boldsymbol{E} = (E, (1,2), (1,2))$ with $V(E) = \{1,2\}$ and $E(E) = \{12\}$,
\item the \emph{merge graph} $\boldsymbol{M} = (M, (1,1), (1,1))$ with $V(M) = \{1\}$ and $E(M) = \emptyset$.
\end{itemize}
Observe that all these graphs are such that under series composition the first labelled vertex is fixed.

\begin{theorem}\label{thm:chhom}
Let $G$ and $H$ be graphs. Then the following are equivalent:
\begin{enumerate}
\item $G$ and $H$ are $(1,1)$-\WL indistinguishable.\label{ch1}
\item There exists a bijection $\pi \colon V(G) \to V(H)$ such that $G_v$ and $H_{\pi(v)}$ model the same sentences of $\mathsf{C}_2$, i.e.\@ the two-variable fragment of first order logic with counting quantifiers.\label{ch2}
\item $G$ and $H$ are homomorphism indistinguishable over $\mathcal{T}^+$.\label{ch4}
\item There exists a doubly-stochastic matrix $X \in \mathbb{Q}^{V(H)^2 \times V(G)^2}$ such that $X\boldsymbol{B}_G = \boldsymbol{B}_H X$ for all $\boldsymbol{B} \in \{ \boldsymbol{I},\boldsymbol{C},\boldsymbol{G},\boldsymbol{E}, \boldsymbol{M}\}$.\label{ch5}
\end{enumerate}
\end{theorem}
\begin{proof}
The equivalence of \cref{ch1,ch2} follows readily from \cite{cai_optimal_1992}. 

That \cref{ch2} implies \cref{ch4} follows by adapting \cite[Lemma~4]{dvorak_recognizing_2010}: Intuitively, the assertion $\hom(T, G) = m$ for $T$ a tree and $m \in \mathbb{N}$ can be translated to a $\mathsf{C}_2$-sentence with variables corresponding to the vertices of $T$. To encode $\hom(T/B, G) = m$ for $T/B \in \mathcal{T}^+$, one has to replace expressions $Qx.\ \phi(x, y)$ for $Q$ any quantifier, $x, y$ variables, and $x$ representing a vertex in $B$ by $Qx.\ B(x) \land \phi(x,y)$ where $B$ is a unary predicate encoding the colour of $v$ in $G_v$. This yields \cref{ch4}.

The converse follows by adapting \cite[Lemma 6]{dvorak_recognizing_2010}: A sentence $\phi \in \mathsf{C}_2$ over the relational vocabulary of (uncoloured) graphs can be translated into a homomorphism constraint $\hom(T, G) = m$ for $T$ a tree and $m \in \{0,1\}$. If instead $\phi$ is over the relational vocabulary $\{E, B\}$ of vertex-individualised graphs then this translation yields a tree $T$ and a set $B \subseteq V(T)$ comprising the vertices corresponding to variables $x$ prescribed to satisfy $B(x)$. Since there is a unique vertex in $G$ modelling $B(x)$, the result is a homomorphism constraint $\hom(T/B, G) = m$ for $T/B \in \mathcal{T}^+$. This implies \cref{ch2}.

The equivalence of \cref{ch4,ch5} requires similar arguments as \cref{thm:main-intro}.
For two $\ell$-labelled graphs $\boldsymbol{F} = (F, \vec{u})$ and $\boldsymbol{F}' = (F', \vec{u'})$, define the \emph{gluing product} $\boldsymbol{F} \odot \boldsymbol{F}'$ as the $\ell$-labelled graph obtained by taking the disjoint union of $F$ and  $F'$ and identifying vertices $\vec{u}$ and $\vec{u}'$ element-wise, cf.\@ \cite{mancinska_quantum_2019}. Write $\boldsymbol{1}$ for the $2$-labelled $2$-vertex graph without any edges and labels on distinct vertices.

\begin{claim}\label{claim1}
The set $\mathcal{T}^\circ$ is the closure of $\{\boldsymbol{1}\}$ under multiplication with $\boldsymbol{I}$, $\boldsymbol{C}$, $\boldsymbol{G}$, $\boldsymbol{E}$, and $\boldsymbol{M}$ as well as taking gluing products.
\end{claim}
\begin{claimproof}
It is easy to see that $\mathcal{T}^\circ$ is closed under these operations and $\boldsymbol{1} \in \mathcal{T}^\circ$. The converse inclusion is shown by induction on the number of vertices of $\boldsymbol{T} \in \mathcal{T}^\circ$. If $\boldsymbol{T}$ has just a single vertices then $\boldsymbol{T} = \boldsymbol{M}\boldsymbol{1}$. If $\boldsymbol{T}$ consists of two vertices then it is one of $\boldsymbol{1}$, $\boldsymbol{E}\boldsymbol{1}$, $\boldsymbol{M}\boldsymbol{F}\boldsymbol{1}$, or $\boldsymbol{M}\boldsymbol{F}\boldsymbol{E}\boldsymbol{1}$.

Now let $\boldsymbol{T} = (T/B, u_1u_2)$ have more than two vertex. It may be assumed that $u_2 \not\in B$. Indeed, let otherwise $u'_2 \in V(T) \setminus B$ be any other vertex. Let $\boldsymbol{T}' \coloneqq (T/B, u_1u'_2)$. Then $\boldsymbol{T} = \boldsymbol{M}\boldsymbol{F}\boldsymbol{T}'$ is such that its second labelled vertex does not lie in $B$.  Distinguish cases:
\begin{itemize}
\item $u_2$ is isolated in $T/B$. Let $T' \coloneqq T \setminus u_2$ and $\boldsymbol{T}' \coloneqq (T'/B, u_1u'_2)$ for any other vertex $u'_2 \in V(T')$. Then, $\boldsymbol{T} = \boldsymbol{F}\boldsymbol{T}'$.
\item $u_2$ is of degree one  in $T/B$. If $u_2$ is connected to $u_1$ in $T/B$ then let $T' \coloneqq T \setminus u_2$ and $\boldsymbol{T}' \coloneqq (T'/B, u_1u'_2)$ for any other vertex $u'_2 \in V(T')$. Then, $\boldsymbol{T} = \boldsymbol{E} \boldsymbol{F} \boldsymbol{T}'$. If $u_2$ is connected to a vertex other than $u_1$, let $T' \coloneqq T \setminus u_2$ and $\boldsymbol{T}' \coloneqq (T'/B, u_1u'_2)$ for $u'_2 \in V(T')$ the vertex $u_2$ is connected to. Then $\boldsymbol{T} = \boldsymbol{C}\boldsymbol{T}'$.
\item $u_2$ is of degree greater than one  in $T/B$. Write $T^1, \dots, T^r$ for the connected components of $T \setminus u_2$. Let $v_i \in V(T^i)$ for $i \in [r]$ be the vertex such that $u_2v_i \in E(T)$. Observe that $V(T) = V(T^1) \sqcup \dots \sqcup V(T^r) \sqcup \{u_2\}$ and hence without loss of generality $\abs{V(T^i)} < \abs{V(T)}-1$ for all $i \in [r]$. Define $B^i$ to be the disjoint union of $B \cap V(T^i)$ and a fresh isolated vertex $v'_i$. Let $\boldsymbol{T}^i \coloneqq ((T^i \sqcup \{v'_i\})/B^i, v'_iv_i)$. Then if none of $v_i$ are in $B$ it holds that $\boldsymbol{T} = (\boldsymbol{C}\boldsymbol{T}^1) \odot \dots \odot (\boldsymbol{C}\boldsymbol{T}^r)$, otherwise $\boldsymbol{T} = \boldsymbol{E}((\boldsymbol{C}\boldsymbol{T}^1) \odot \dots \odot (\boldsymbol{C}\boldsymbol{T}^r))$.\qedhere
\end{itemize}
\end{claimproof}

Write $\mathbb{C}\mathcal{T}^\circ_G \leq \mathbb{C}^{V(G)^2}$ for the space spanned by the homomorphism vectors $\boldsymbol{T}_G$, $\boldsymbol{T} \in \mathcal{T}^\circ$. By \cref{claim1}, $\mathbb{C}\mathcal{T}^\circ_G$ is closed under Schur products and invariant under the action of $\boldsymbol{I}_G$, $\boldsymbol{C}_G$, $\boldsymbol{E}_G$, $\boldsymbol{F}_G$ and $\boldsymbol{M}_G$. By \cite[Theorem~56]{grohe_homomorphism_arxiv}, \cref{ch4} is equivalent to the existence of a doubly-stochastic map $X \colon \mathbb{C}\mathcal{T}^\circ_G \to \mathbb{C}\mathcal{T}^\circ_H$ such that $X \boldsymbol{B}_G = \boldsymbol{B}_HX$ for all $\boldsymbol{B} \in \{ \boldsymbol{I},\boldsymbol{C},\boldsymbol{G},\boldsymbol{E}, \boldsymbol{M}\}$. This is the content of \cref{ch5}.
\end{proof}

\section{A Comonadic Strategy for Homomorphism Indistinguishability}
\label{app:comonadic}

We propose to view our \cref{thm:main-intro} as the instantiation of a novel strategy for proving characterisations of homomorphism indistinguishability in terms of matrix equations for graph classes which so far have resisted the strategy outlined in \cref{recipe}. This novel strategy is inspired by the comonadic framework developed in \cite{abramsky_pebbling_2017,abramsky_relating_2021,dawar_lovasz_2021,abramsky_discrete_2022}, etc. It is in large parts generic, in the sense that it does not require prior knowledge of the graph class in questions. This is achieved by leveraging the properties of a comonad characterising the graph class.

\subsection{Preliminaries}
We first sketch the framework of comonads on the category of graphs before commenting on how it can be used to derive matrix equations. See \cite{dawar_lovasz_2021,lane_categories_1971} for further details. Let $\Graph$ denote the category of finite graphs, i.e.\@ with objects being finite graphs and morphisms being graph homomorphisms. A \emph{comonad} on $\Graph$ is a tuple $(\C, \epsilon, \delta)$ where $\C \colon \Graph \to \Graph$ is a functor and $\epsilon \colon \C \to 1_{\Graph}$ and $\delta \colon \C \to \C\C$ are natural transformations satisfying the following diagrams for all objects $A$ of $\Graph$:
\begin{equation*} \label{eq:comonad}
\begin{tikzcd}
\C A \arrow[r, "\delta_A"] \arrow[d, "\delta_A"] & \C \C A \arrow[d, "\C \delta_A"] \\
\C \C A \arrow[r, "\delta_{\C A}"]               & \C \C \C A                      
\end{tikzcd}
\quad\quad
\begin{tikzcd}
\C A \arrow[r, "\delta_A"] \arrow[d, "\delta_A"] \arrow[rd, "\id_{\C A}"] & \C \C A \arrow[d, "\C \epsilon_A"] \\
\C\C A \arrow[r, "\epsilon_{\C A}"]                                              & \C A                              
\end{tikzcd}
\end{equation*}

\begin{example}[\cite{abramsky_pebbling_2017,dawar_lovasz_2021}]
The \emph{pebbling comonad} maps a graph $A$ to the graph $\mathfrak{P}_{k,d}A$ with vertex set being the set of sequences $[(p_1, a_1), \dots, (p_\ell, a_\ell)]$ for $p_i \in [k]$, $a_i \in V(A)$, and $i \in [\ell]$, $\ell \leq d$. There is an edge between $[(p_1, a_1), \dots, (p_\ell, a_\ell)]$ and $[(p'_1, a'_1), \dots, (p'_{\ell'}, a'_{\ell'})]$ if and only if one sequence is and initial segment of the other, and if wlog the first sequence is contained in the second then $p_\ell \not\in \{p'_{\ell+1}, \dots, p'_{\ell'}\}$, and $a_\ell a'_{\ell'} \in E(A)$.
The homomorphisms $\epsilon_A \colon \mathfrak{P}_{k,d}A \to A$ sends $[(p_1, a_1), \dots, (p_\ell, a_\ell)]$ to $a_\ell$. It is readily verified that this is a morphisms in $\Graph$. The homomorphisms $\delta_A \colon \mathfrak{P}_{k,d}A \to \mathfrak{P}_{k,d}\mathfrak{P}_{k,d}A$ maps
\[
[(p_1, a_1), \dots, (p_\ell, a_\ell)] \mapsto
[(p_1, s_1), \dots, (p_\ell, s_\ell)]
\]
where $s_j = [(p_1, a_1), \dots, (p_j, a_j)]$ for all $j \leq \ell$. The triple $(\mathfrak{P}_{k,d}, \epsilon, \delta)$ specifies the pebbling comonad. 
\end{example}

A \emph{coalgebra} for a comonad $(\C, \epsilon, \delta)$ is a pair $(A, \alpha)$ where $A \in \obj \Graph$ and $\alpha \colon A \to \C A$ is a morphism such that the following diagrams commute:
\begin{equation*} \label{eq:coalgebra}
\begin{tikzcd}
A \arrow[r, "\alpha"] \arrow[d, "\alpha"] & \C A \arrow[d, "\C \alpha"] \\
\C A \arrow[r, "\delta_A"]               & \C \C A                   
\end{tikzcd}
\quad\quad
\begin{tikzcd}
A \arrow[r, "\alpha"] \arrow[rd, "\id_A"] & \C A \arrow[d, "\epsilon_A"] \\
                                          & A                           
\end{tikzcd}
\end{equation*}
The significance of coalgebras for this paper stems from the following result:
\begin{theorem}[{\cite[Section IV.B]{dawar_lovasz_2021}}]\label{thm:dawar-coalg}
Let $A \in \obj \Graph$. Then the following are equivalent:
\begin{enumerate}
\item $A$ admits a $k$-pebble forest cover of depth $\leq d$,
\item $A$ admits a $\mathfrak{P}_{k,d}$-coalgebra $A \to \mathfrak{P}_{k,d}A$.
\end{enumerate}
\end{theorem}
In fact, there is a bijective correspondence between these objects. We therefore think of coalgebras as graph decompositions.

The \emph{Eilenberg--Moore category} $\EM(\C)$ of a comonad $\C$ is the category whose objects are $\C$-coalgebras and whose morphisms are $\C$-coalgebra morphisms. A \emph{morphism $h \colon (A, \alpha) \to (B, \beta)$ of $\C$-coalgebras} is a morphisms $h \colon A \to B$ such that the following diagram commutes:
\begin{equation*} \label{eq:colalgmor}
\begin{tikzcd}
A \arrow[d, "h"] \arrow[r, "\alpha"] & \C A \arrow[d, "\C h"] \\
B \arrow[r, "\beta"]                 & \C B                  
\end{tikzcd}
\end{equation*}
Informally, the Eilenberg--Moore category is the category of decomposed graphs. The \emph{forgetful functor} $U^\C \colon \EM(\C) \to \Graph$ which sends $(A, \alpha)$ to $A$ associates a decomposed graph to its underlying graph. Thus, by \cref{thm:dawar-coalg}, a graph is in the image of $U^{\mathfrak{P}_{k,d}}$ if and only if it is admits a $k$-pebble forest cover of depth $\leq d$.

\begin{fact}[cf.\@ e.g.\@ \cite{lane_categories_1971}] \label{fact:emlimit}
Let $\C$ be a comonad over a category $\mathcal{A}$. Then $\EM(\C)$ has all colimits that $\mathcal{A}$ has. The forgetful functor $U^\C \colon \EM(\C) \to \mathcal{A}$ creates them.
\end{fact}

From the perspective of homomorphism counting, the characterisation of the adjunction underlying the comonad $\C$ in terms of homomorphism functor is of utmost importance. Let $F^\C \colon \Graph \to \EM(\C)$ denote the \emph{free functor} mapping $A$ to $(\C A, \delta_A)$, i.e.\@ $F^\C \vdash U^\C$.

\begin{fact}[{\cite[Corollary~5.4.23]{perrone_notes_2021} or e.g.\@ \cite{dawar_lovasz_2021}}] \label{fact:adjunction}
Let $(A, \alpha) \in \obj \EM(\C)$ and $G \in \obj \Graph$. Then
\[
\hom_\Graph(U^\C(A, \alpha), G) \cong \hom_{\EM(\C)}((A, \alpha), F^\C G). 
\]
Denote by $\widehat{\ell} \colon (A, \alpha) \to F^\C G$ the morphism corresponding to $\ell \colon U^\C(A, \alpha) \to G$ under this bijection.
\end{fact}

In other words, counting homomorphisms in $\Graph$ from $A$ to $G$ where $A$ admits a $\C$-coalgebra amounts to counting homomorphisms in $\EM(\C)$ from a $\C$-coalgebra of $A$ to the free $\C$-coalgebra $(\C G, \delta_G)$ of $G$.

\subsection{Bilabelled Objects and their Coalgebraic and Augmented Homomorphism Representation}
\label{sec:cat-step1}
We develop a categorical language for bilabelled objects. The following definitions resemble those in~\cite{lovasz_semidefinite_2009}.
Fix a category $\mathcal{C}$, which has all finite pushouts. An example for such a category is $\Graph$ but also $\EM(\C)$ for a comonad $\C$ on $\Graph$, cf.\@ \cref{fact:emlimit}.

\begin{definition}[Categories of Bilabelled Objects]\label{def:cat-bil}
For $L \in \obj \mathcal{C}$, define $\Lab_{\mathcal{C}}(L, K)$ to be the category with objects being pairs of $\mathcal{C}$-morphisms $L \to A \ot K$ and morphisms being $\mathcal{C}$-morphisms~$h$ such that the following diagram commutes:
\[
\begin{tikzcd}
L \arrow[r] \arrow[rd] & A \arrow[d, "h"] & K \arrow[l] \arrow[ld] \\
                       & B                &                       
\end{tikzcd}
\]
For a family of objects $\mathcal{L} \subseteq \obj \mathcal{L}$,
write $\Lab_\mathcal{C}(\mathcal{L})$ for the union of all $\Lab_\mathcal{C}(L, K)$ for $L, K \in \mathcal{L}$.
We use boldface letters $\boldsymbol{A} = (L \to A \ot K)$ to indicate bilabelled objects.
\end{definition}

An instance of a category of bilabelled objects can be constructed from a comonad $\C$ on $\Graph$ and its Eilenberg--Moore category $\EM(\C)$ as follows.
Fix a graph $L \in \obj \Graph$. Write $\mathcal{L}$ for the set of all $\C$-coalgebras $L \to \C L$ of $L$. 
Let $\Lambda(\C, L) \coloneqq \obj \Lab_{\EM(\C)}(\mathcal{L}) \sqcup \{\bot\}$ where $\bot$ is a fresh symbol. In other words, $\Lambda(\C, L)$ comprises all bilabelled objects of the form $(L, \lambda) \to (A, \alpha) \ot (L, \lambda')$ and $\bot$. The set $\Lambda(\C, L)$ is the family of bilabelled graphs of our interest.

\begin{example}
With $L$ the edge-less graph with $d$ vertices, $\Lambda(\mathfrak{P}_{k+1,d}, L)$ resembles $\mathcal{WL}_k^d$ as introduced in \cref{def:wlkd}. 
Differences arise since \cref{def:wlkd} includes only certain pebble forest covers arising in \cref{thm:camb}.
\end{example}

We proceed to define a suitable homomorphism representation. First, we generalise the homomorphism representation of \cite{mancinska_quantum_2019,grohe_homomorphism_2021} which does not encode any information about the labels.
To avoid confusion with homomorphism counts, write $\mor(L, G)$ for the set of morphisms $L \to G$ in $\Graph$.
Note that since we consider finite graphs only, this set is finite.
Recall the notation $\widehat{\ell}$, $\widehat{k}$, etc.\@ from \cref{fact:adjunction}.

\begin{definition}\label{def:em-hom-rep}
Let $G$ be a graph. The \emph{comonadic homomorphism representation} is the map
$
\Lambda(\C, L) \to \mathbb{C}^{\mor(L, G) \times \mor(K,G)}
$
mapping $\boldsymbol{A} = \left((L, \lambda) \to (A, \alpha) \ot (K, \kappa) \right)$ to the matrix $\boldsymbol{A}_G \in \mathbb{C}^{\mor(L, G) \times \mor(K, G)} $ with entry $(\ell, k)$ being the number of homomorphisms $A \to G$ such that the diagram
\begin{equation*}\label{eq:liftrep}
\begin{tikzcd}
(L,\lambda) \arrow[r] \arrow[rd, "\hat{\ell}"] & (A, \alpha) \arrow[d, dashed] & (K, \kappa) \arrow[l] \arrow[ld, "\hat{k}"'] \\
                               & F^\C G                  &                             
\end{tikzcd}
\end{equation*}
commutes. Furthermore, $\bot$ is mapped to the zero matrix.
\end{definition}

It remains to define a representation which also incorporates information about the structure of the labels, i.e.\@ their coalgebras.

\begin{definition} \label{def:cat-aug-hom-rep}
Let $G$ be a graph. The \emph{augmented homomorphism representation} associates with every element of $\Lambda(\C, L)$ an endomorphism of the $\mathbb{C}$-vector space
\[
\mathbb{C}^{\mor(L, G)} \otimes \bigoplus_{\substack{\C\text{-coalgebras} \\ \lambda \colon L \to \C L}} \mathbb{C}.
\]
It maps $\boldsymbol{A} = \left((L, \lambda) \to (A, \alpha) \ot (L, \kappa) \right)$ to $\widehat{\boldsymbol{A}}_G \coloneqq \boldsymbol{A}_G \otimes e_\lambda e_\kappa^T$ where $\boldsymbol{A}_G$ is defined in \cref{def:em-hom-rep} and $e_\lambda, e_\kappa \in \bigoplus_{\lambda \colon L \to \C L} \mathbb{C}$ are the $\lambda$-th and $\kappa$-th standard basis vectors respectively.
Furthermore, $\bot$ is mapped to the zero matrix.
\end{definition}

Note that we do not succeed to parallel the notation of \cref{sec:genadjmat} entirely. In \cref{sec:genadjmat}, $\boldsymbol{A}$ and $\widehat{\boldsymbol{A}}$ indicated different objects with different representations $\boldsymbol{A}_G$ and $\widehat{\boldsymbol{A}}_G$. Here, the underlying objects are the same, only the representations are different. We write $\boldsymbol{A} = ((L,\lambda) \to (A,\alpha) \ot (K,\kappa))$ to indicate bilabelled coalgebras, $\boldsymbol{A}_G$ to indicate the comonadic homomorphism representation, and $\widehat{\boldsymbol{A}}_G$ to indicate the augmented homomorphism representation.

\subsection{Operations}
\label{sec:cat-step2}
The familiar operations on bilabelled objects can also be defined in categorical terms.

\begin{definition}[Operations on Bilabelled Objects] \label{def:cat-ops}
\begin{enumerate}
\item Define the \emph{series composition} of $L \to A \leftarrow K$ and $K \to B \leftarrow M$, in symbols $(L \to A \leftarrow K) \cdot (K \to B \leftarrow M)$, as the unique morphisms $L \to C \leftarrow M$ induced by the pushout diagram in \cref{eq:seriescomp} where $C$ is the pushout.
\begin{equation}\label{eq:seriescomp}
\begin{tikzcd}
L \arrow[rd] &                      & K \arrow[ld] \arrow[rd] &                      & M \arrow[ld] \\
             & A \arrow[rd, dashed] &                         & B \arrow[ld, dashed] &              \\
             &                      & C                       &                      &             
\end{tikzcd}
\end{equation}
\item  Define the \emph{reverse} of $L \to A \ot K$, in symbols $(L \to A \ot K)^*$, as $K \to A \ot L$.
\end{enumerate}
\end{definition}

The operations in \cref{def:cat-ops} are compatible with the comonadic homomorphism representation.

\begin{lemma}\label{lem:cat-series-comp}
Let $\boldsymbol{A} = \left((L, \lambda) \to (A, \alpha) \ot (K, \kappa) \right)$ and $\boldsymbol{B} = \left((K, \kappa) \to (B, \beta) \ot (M, \mu) \right)$. 
Let $G$ be a graph.
Then $(\boldsymbol{A} \cdot \boldsymbol{B})_G = \boldsymbol{A}_G \cdot \boldsymbol{B}_G$.
\end{lemma}
\begin{proof}
Write $(C, \gamma)$ for the pushout of $a_2 \colon (K, \kappa) \to (A, \alpha)$ and $b_1 \colon (K, \kappa) \to (B, \beta)$.
Let $\ell \in \mor(L, G)$ and $m \in \mor(M, G)$.
Consider the following diagram. 
\[
\begin{tikzcd}
{(L, \lambda)} \arrow[rd, "a_1"] \arrow[rrddd, "\widehat{\ell}", bend right] &                          & {(K, \kappa)} \arrow[ld, "a_2"'] \arrow[rd, "b_1"] &                         & {(M, \mu)} \arrow[ld, "b_2"'] \arrow[llddd, "\widehat{m}"', bend left] \\
                                                                             & {(A, \alpha)} \arrow[rd] &                                                    & {(B, \beta)} \arrow[ld] &                                                                        \\
                                                                             &                          & {(C, \gamma)} \arrow[d, dashed]                    &                         &                                                                        \\
                                                                             &                          & F^\C G                                             &                         &                                                                       
\end{tikzcd}
\]
By the definition of the pushout and \cref{fact:adjunction}, the set of morphisms $(C, \gamma) \to F^\C G$ making the above diagram commute is in bijection with the set of triples $(k, a, b)$ where $k \in \mor(K, G)$, $a \colon (A, \alpha) \to F^\C G$, $b \colon (B, \beta) \to F^\C G$, such that $a a_2 = \widehat{k} = bb_1$, $\widehat{\ell} = a a_1$, and $\widehat{m} = b b_2$. This implies the claim.
\end{proof}

\Cref{lem:cat-series-comp} has immediate consequences for the augmented homomorphism representation.

\begin{corollary}
Let $\boldsymbol{A} = \left((L, \lambda) \to (A, \alpha) \ot (K, \kappa) \right)$ and $\boldsymbol{B} = \left((K, \kappa) \to (B, \beta) \ot (M, \mu) \right)$. 
Let $G$ be a graph.
Then $(\widehat{\boldsymbol{A} \cdot \boldsymbol{B}})_G = \widehat{\boldsymbol{A}}_G \cdot \widehat{\boldsymbol{B}}_G$.
\end{corollary}

With series composition and reversal, $\Lambda(\C, L)$ has the structure of an \emph{involution semigroup}, i.e.\@ series composition is an associative binary operation and reversal induces an involution map ${}^* \colon \Lambda(\C, L) \to \Lambda(\C, L)$ such that $(g \cdot h)^* = h^*  \cdot g^*$ and $(g^*)^* = g$ for all $g, h \in \Lambda(\C, L)$.
The series composition of $(L, \lambda_1) \to (A, \alpha) \ot (L, \lambda_2)$ and $(L, \lambda_3) \to (A, \alpha) \ot (L, \lambda_4)$ is defined via \cref{def:cat-ops} if $\lambda_2 = \lambda_3$. Otherwise, the result is defined to be $\bot$. Also the series composition of any element with $\bot$ on either side is $\bot$. The involution operation is defined as the reversal of \cref{def:cat-ops}. Again, $(\bot)^* \coloneqq \bot$.

As before, the unlabelling operation permits us to pass from a bilabelled coalgebra to the underlying graph.

\begin{definition}
The \emph{unlabelling} of $\boldsymbol{A} = \left((L, \lambda) \to (A, \alpha) \ot (K, \kappa) \right)$, in symbols $\soe \boldsymbol{A}$, is~$A$.
\end{definition}

Unlabelling corresponds to sum-of-entries under the comonadic and the augmented homomorphism representation.

\begin{lemma}
For every graph $G$ and every $\boldsymbol{A} = \left((L, \lambda) \to (A, \alpha) \ot (K, \kappa) \right)$, \[ \hom_\Graph(\soe \boldsymbol{A}, G) = \soe \boldsymbol{A}_G = \soe \widehat{\boldsymbol{A}}_G. \]
\end{lemma}
\begin{proof}
By \cref{fact:adjunction}, $\hom_\Graph(A, G) = \hom_{\EM(\C)}((A, \alpha), F^\C G)$. By \cref{eq:liftrep}, this is the sum-of-entries of $\boldsymbol{A}_G$. Recall from \cref{def:cat-aug-hom-rep}, that only one entry in the second factor of the Kronecker product $\widehat{\boldsymbol{A}}_G$ is non-zero. This entry equals one and hence $\soe \boldsymbol{A}_G = \soe \widehat{\boldsymbol{A}}_G$.
\end{proof}

\subsection{Finite Generation}
\label{sec:cat-step3}

For the case of the pebbling comonad, finite generation was proven in \cref{prop:fingen} after devising a suitable set of generators. Proving finite generation for the involution semigroup of bilabelled objects over arbitrary comonads on $\Graph$, remains a problem for further investigations. 

\begin{definition}
A subset $\Gamma \subseteq \Lambda(\C, L)$ is a \emph{set of generators} for $\Lambda(\C, L)$ if for every $\widehat{\boldsymbol{A}} \in \Lambda(\C, L)$ there exist finitely many $\widehat{\boldsymbol{B}}^1, \dots, \widehat{\boldsymbol{B}}^r \in \Gamma$ such that $\widehat{\boldsymbol{A}} = \widehat{\boldsymbol{B}}^1 \cdot \dots \cdot \widehat{\boldsymbol{B}}^r$.
\end{definition}

\subsection{Matrix Equations}
\label{sec:cat-setp4}

It remains to deduce matrix equations characterising homomorphism indistinguishability over the set of graphs underlying $\Lambda(\C, L)$.
Note that finite generation is not among the assumptions of the following theorem. It is only needed to obtain a finite system of matrix equations.

\begin{theorem} \label{thm:coalgebraic-main}
Let $\C$ be a comonad on $\Graph$. Let $L \in \operatorname{im} U^\C$ be a graph such that $\Lambda(\C, L)$ admits a set of generators $\Gamma$, which is closed under reversal. Let $G$ and $H$ be graphs. Then the following are equivalent:
\begin{enumerate}
\item $G$ and $H$ are homomorphism indistinguishable over the graphs $\soe \boldsymbol{A}$ for $ \boldsymbol{A} \in \Lambda(\C, L)$.
\item There exists a pseudo-stochastic matrix $X$ such that $X \widehat{\boldsymbol{A}}_G = \widehat{\boldsymbol{A}}_H X$ for all $\widehat{\boldsymbol{A}} \in \Gamma$.
\end{enumerate}
\end{theorem}

\begin{proof}
The backward direction is easy to see, cf.\@ proof of \cref{thm:main-intro}.
The forward direction follows from \cref{thm:soe} observing that $\soe \widehat{\boldsymbol{A}}_G = \soe \boldsymbol{A}_G$ for all $\boldsymbol{A}$ and $G$.
\end{proof}

\subsection{The Strategy}

We conclude by summarising the comonadic strategy to derive homomorphism indistinguishability characterisations in terms of matrix equations.
This approach is exemplified with less notational overhead in \cref{sec:genadjmat} for the case of the pebbling comonad $\mathfrak{P}_{k,d}$.

\begin{enumerate} \label{recipe2}
\item Define a family of labelled graphs and their homomorphism matrices

Fix a graph $L$ and define $\Lambda(\C, L)$ as in \cref{sec:cat-step1} as the family of bilabelled graphs of interest. They are of the form $(L, \lambda) \to (A, \alpha) \ot (L, \lambda')$ where $\lambda$ and $\lambda'$ are coalgebras of $L$ and $A$ is a graph of interest. The representation in terms of matrices is defined in \cref{def:cat-aug-hom-rep}. It encodes the homomorphism counts and information about the labelling, i.e.\@ the coalgebras $\lambda$ and $\lambda'$.

For the pebbling comonad, \cref{ssec:Akd} features the set $\mathcal{WL}_k^d$. It comprises tuples $(\boldsymbol{F}, p_\pin, p_\pout)$ where $\boldsymbol{F}$ is a bilabelled graph possessing a certain pebble forest cover. This resembles the above definition. The functions $p_\pin$ and $p_\pout$ play the role of $\lambda$ and $\lambda'$. The object $L$ is the graph with $d$ isolated vertices placed on a path in its forest cover.

\item Define suitable combinatorial and algebraic operations

Series composition and reversal are the relevant combinatorial operations, cf.\@ \cref{sec:cat-step2}. They correspond under the augmented homomorphism representation to matrix product and transposition. 
In the coalgebraic case, closure under these operations follows essentially from basic properties of the Eilenberg--Moore category $\EM(\C)$. Remembering the coalgebras $\lambda$ and $\lambda'$ also on the algebraic side ensures that no homomorphism matrices of graphs outside the class are produced.

For the pebbling comonad, the relevant lemmata were established in \cref{sec:wlkd-step2}.

\item Prove finite generation

Finite generation as defined in \cref{sec:cat-step3} is necessary to obtain a finite system of matrix equations in the subsequent step. This has been proven in \cref{ssec:fingen} for the pebbling comonad. Proving finite generation for arbitrary comonads on $\Graph$ remains an open problem. 

\item Recover a matrix equation using algebraic and representation-theoretic techniques.

The desired matrix equations can be obtained using the representation-theoretic framework of~\cite{grohe_homomorphism_2021}, cf.\@ \cref{thm:soe}. This resulted in \cref{thm:coalgebraic-main} in the general case and \cref{thm:main-intro} in the case of the pebbling comonad.

\end{enumerate}

\subsection{Outlook}

We propose to apply this strategy to other comonads and hence to other graph classes in the future.
In \cite{abramsky_discrete_2022}, it was recently shown that for every graph class $\mathcal{F}$ which is closed under isomorphism, finite coproducts, and summands, cf.\@ ibidem, there exists a comonad $\C$ such that $F \in \mathcal{F}$ if and only if $F$ admits a $\C$-coalgebra.
Thus, it may be argued that all graph classes of interest can be characterised using a comonad.
However, the comonads constructed in \cite{abramsky_discrete_2022} are only suitable for counting homomorphisms and fail to represent the graph decompositions, which were a conceptual ingredient for proving finite generation for \cref{thm:main-intro}.

For the future, we propose to eradicate issues with finiteness in \cref{thm:coalgebraic-main} to derive a general theorem characterising homomorphism indistinguishability over graph classes characterised by a comonad in terms of matrix equations.

\end{document}